\documentclass[12pt]{article}
\usepackage{graphicx}
\usepackage{natbib}
\usepackage[margin=1in]{geometry}
\usepackage{xcolor}
\usepackage{hyperref}
\hypersetup{
    colorlinks=true,
    linkcolor=blue,
    urlcolor=blue,
    citecolor=blue
}
\usepackage{parskip}
\usepackage{enumitem}
\usepackage{amsmath,amssymb}
\usepackage{amsthm}
\usepackage{booktabs}
\usepackage{siunitx}
\usepackage{multicol}
\usepackage{multirow}
\usepackage{float}
\usepackage{placeins}
\usepackage[font=small, labelfont=bf]{caption}
\newcommand{\floatnote}[1]{%
  \vspace{0.4em}%
  \begin{minipage}{0.94\textwidth}%
  \footnotesize\textit{Notes:} #1%
  \end{minipage}%
}
\usepackage{subcaption}
\usepackage{comment}
\usepackage{tikz}
\usetikzlibrary{positioning, shapes, arrows.meta, fit, decorations.pathreplacing}
\usepackage{setspace}
\newcommand{\R}{\mathbb{R}}
\newcommand{\N}{\mathbb{N}}
\newcommand{\E}{\mathbb{E}}
\newcommand{\Var}{\mathrm{Var}}
\newcommand{\Cov}{\mathrm{Cov}}
\newcommand{\tr}{\mathrm{tr}}
\newcommand{\argmax}{\operatorname*{arg\,max}}

\DeclareMathOperator{\vech}{vech}
\DeclareMathOperator{\softplus}{softplus}
\DeclareMathOperator{\polylog}{polylog}
\DeclareMathOperator{\chol}{chol}
\DeclareMathOperator{\ssign}{ssign}      % for "smooth softsign"
\DeclareMathOperator{\SD}{SD}

\renewcommand{\vec}[1]{\boldsymbol{#1}}
\newcommand{\mat}[1]{\boldsymbol{\mathbf{#1}}}
\newcommand{\ones}{\vec{1}}
\makeatletter
\newcommand*\bigcdot{\mathpalette\bigcdot@{.5}}
\newcommand*\bigcdot@[2]{\mathbin{\vcenter{\hbox{\scalebox{#2}{$\m@th#1\bullet$}}}}}
\makeatother

\newcommand{\numalt}{K}

\newcommand{\dimparam}{d}

\newcommand{\paramspace}{\Theta}
\newcommand{\utilcovspace}{\mathcal{X}_\numalt}
\newcommand{\utilcovspaceex}{\overline{\mathcal{X}}_\numalt}
\newcommand{\utilcovspaceexsubset}{\overline{\mathcal{X}}_\numalt^*}

\newcommand{\param}{\vec{\theta}}
\newcommand{\paramtrue}{\vec{\theta}_0}
\newcommand{\paramest}{\hat{\vec{\theta}}_n}

\newcommand{\parammle}{\tilde{\vec{\theta}}_n}
\newcommand{\parampseudo}{\vec{\theta}_0^{\dagger}}
\newcommand{\parampseudon}{\vec{\theta}_n^{\dagger}}
\newcommand{\netparam}{\vec{\phi}}

\newcommand{\gradparam}{\nabla_{\param}}

\newcommand{\util}{\vec{v}}
\newcommand{\utilstar}{\vec{v}^*}
\newcommand{\cov}{\mat{\Sigma}}
\newcommand{\covstar}{\mat{\Sigma}^*}
\newcommand{\Cstar}{C^*}
\newcommand{\centermat}{\mat{M}}

\newcommand{\choiceprob}{P}
\newcommand{\choiceprobhat}{\hat{P}}

\newcommand{\loglikobs}{m}
\newcommand{\loglikobshat}{\hat{\loglikobs}}
\newcommand{\loglikobshatn}{\loglikobshat_n}
\newcommand{\loglik}{\ell_n}
\newcommand{\loglikhat}{\hat{\ell}_n}
\newcommand{\score}{\vec{\psi}}
\newcommand{\scorehat}{\hat{\score}}
\newcommand{\scorehatn}{\scorehat_n}

\newcommand{\fisher}{\mat{J}}
\newcommand{\Amat}{\mat{A}}
\newcommand{\Bmat}{\mat{B}}
\newcommand{\sandwichmat}{\mat{V}}

\newcommand{\invarutil}{g}
\newcommand{\exceptset}{\mathcal{E}}
\newcommand{\exceptsetutil}{\mathcal{B}}
\newcommand{\exceptsetutilex}{\overline{\mathcal{B}}}

\newcommand{\permgroupsub}{S_{\numalt-1}^{(j)}}
\newcommand{\permgroupsubsub}{S_{\numalt-2}^{(1,j)}}

\newcommand{\convp}{\xrightarrow{p}}
\newcommand{\convd}{\xrightarrow{d}}

\newcommand{\Prob}{\mathrm{Pr}}

\theoremstyle{plain}
\newtheorem{theorem}{Theorem}
\newtheorem{proposition}{Proposition}
\newtheorem{lemma}{Lemma}
\newtheorem{corollary}{Corollary}

\theoremstyle{definition}
\newtheorem{definition}{Definition}
\newtheorem{assumption}{Assumption}
\theoremstyle{remark}
\newtheorem{remark}{Remark}

\usepackage{titling}
\title{Amortized Inference for Correlated Discrete Choice Models via Equivariant Neural Networks}

\author{
Easton Huch\thanks{Postdoctoral Fellow, Johns Hopkins Carey Business School, ehuch@jhu.edu.}
\and
Michael Keane\thanks{Carey Distinguished Professor, Johns Hopkins Carey Business School, mkeane14@jhu.edu.}
}

\date{November 2026\vspace{-12pt}}

\begin{document}

\maketitle

\begin{abstract}
Discrete choice models are fundamental tools in management science, economics, and marketing. Logit models dominate applied work because they yield closed-form choice probabilities, but IID Gumbel errors impose restrictive substitution patterns. We propose an amortized inference approach based on a neural network emulator that approximates choice probabilities under general, potentially correlated, error distributions. We develop a specialized neural network architecture respecting the invariance properties of discrete choice models and provide group-theoretic foundations for universal approximation from a minimal set of invariant features. Using Sobolev training, the emulator learns choice probabilities and their derivatives, enabling rapid likelihood and gradient evaluation. We show that emulator-based maximum likelihood estimators are consistent and asymptotically normal under mild approximation conditions. Simulations show substantial gains over the GHK simulator in accuracy and speed. Using NielsenIQ data on household laundry-detergent purchases, we estimate correlated-probit and correlated-Gumbel models with rich household-level random effects. The IID period-error restriction---and hence the conditional-IIA restriction underlying conventional mixed logit---is strongly rejected. The Bayesian information criterion (BIC) selects a correlated-Gumbel model with a three-factor period-error covariance structure. We find that allowing correlated period errors materially changes predicted product substitution patterns.

\vspace{1em}
\vspace{0.3em}
\end{abstract}

\section{Introduction}
\label{sec:intro}

Discrete choice models are a widely used tool in management science, economics, marketing, and other fields for understanding how individuals and organizations make decisions among finite sets of alternatives \citep{mcfadden1974conditional}. These models provide a structural framework that delivers interpretable parameters---such as willingness-to-pay, demand elasticities, and measures of perceived similarity between products---while also enabling predictions about effects of policy changes, such as price changes and product introductions/deletions.

The dominant discrete choice model in applied work is the multinomial logit (MNL), which assumes consumer utility from alternative $j$ takes the form $U_j = v_j + \epsilon_j$ (suppressing the consumer subscript), where $v_j$ is the deterministic component of utility, typically written as a linear function $\vec{x}_j^{\top} \vec{\beta}$ of alternative $j$'s attributes $\vec{x}_j$, and $\epsilon_j$ is a Type I extreme value (Gumbel) error that is independent and identically distributed (IID) across alternatives $j = 1, \ldots, \numalt$. This setup yields simple closed-form choice probabilities: the probability of choosing alternative $j$ is the ``softmax'' function $\exp(v_j) / \sum_{k=1}^{\numalt} \exp(v_k)$. This computational convenience has made MNL the default in countless applications.  

However, this convenience comes at a well-known cost. The IID Gumbel error assumption implies the restrictive independence of irrelevant alternatives (IIA) property: the odds ratio between any two alternatives is unaffected by the presence or attributes of other alternatives, which can lead to unrealistic substitution patterns. 
A popular generalization of MNL is the mixed logit model (mixed MNL), in which the utility weight $\vec{\beta}$ is allowed to be heterogeneous across consumers. Estimation is straightforward via simulation methods \citep{train2009discrete}. Mixed MNL relaxes IIA at the aggregate level, but IIA still holds at the level of an individual consumer's choices given the corresponding random utility weight.

The multinomial probit (MNP) is a well-known alternative to MNL that relaxes IIA and allows for flexible substitution patterns by assuming $\epsilon_1,\ldots,\epsilon_{\numalt}$ come from a multivariate normal distribution with cross-alternative correlations.
Despite this advantage, MNP has seen relatively limited adoption, largely due to its computational demands. MNP choice probabilities have no closed form and require evaluating multivariate normal rectangle probabilities. Estimation relies on more sophisticated simulation methods than needed for mixed MNL: the GHK simulator in a classical setting \citep{geweke1989bayesian, hajivassiliou1998simulation, keane1994solution} or MCMC in a Bayesian setting \citep{McCullochRossi1994}.

We propose a fundamentally different approach: Rather than simulating choice probabilities anew for each likelihood evaluation, we train a neural network emulator to directly approximate the choice probability function. This strategy---known as amortized inference---shifts the computational burden from inference time to a one-time training phase. Once trained, the emulator provides rapid deterministic approximations to choice probabilities via simple function calls. The amortized inference framework has proven highly successful in multiple scientific domains for inference in computationally intensive simulation models \citep{lueckmann2019likelihood,cranmer2020frontier}. Here, we adapt it to discrete choice models, extending beyond MNL and MNP to general correlated error distributions. One especially interesting case is an extension of MNL to allow for correlated Gumbel errors, thus relaxing IIA at the individual level. We refer to this as the correlated-Gumbel model.

In the econometrics literature, \citet{Norets2012} employs a related strategy, approximating the expected value function in dynamic discrete choice models via a neural network. He applies this strategy to a specific parametric model, so changes in the model require modifying and retraining the network. In contrast, our emulator operates on the generic choice problem, given only the deterministic utilities and the error distribution, enabling modifications to the parametric model (e.g., the functional form of utility) without retraining the emulator.

Our methodology relies on several new contributions. First, we develop a neural network architecture specifically designed for choice models. It respects their fundamental invariance properties, which we formalize as symmetries (group actions) of the choice-probability map.\footnote{If these properties are not embedded in the architecture, the neural network must learn them. As a continuum of models generates equivalent choice probabilities, this would slow learning considerably.} Choice probabilities are invariant to affine utility transformations (location and scale shifts), and equivariant with respect to permutations of choice alternatives. Our architecture incorporates a preprocessing transformation that reduces the size of the feature space by enforcing location and scale invariance. The processed features are then passed to a per-alternative encoder module based on the DeepSet architecture \citep{zaheer2017deep}.
The output is then concatenated across alternatives and processed through equivariant layers that respect the permutation equivariance property.

Second, we establish theoretical foundations of our architecture. We prove it can universally approximate choice probabilities on compact subsets of the space of normalized utilities and scale matrices, outside a measure-zero exceptional set. This result connects our architecture to the theory of orbit separation under group actions, extending the universal approximation results of \citet{blumsmith2025machine} on symmetric matrices to the joint space of utility vectors and centered covariance matrices. While they use Galois theory to prove generic orbit separation, we provide a direct proof based on an invariant reconstruction argument.

Third, we establish the statistical properties of maximum likelihood estimators (MLEs) formed using an emulator approximation of the true discrete choice probabilities. We show that if the emulator approximates the true log likelihood sufficiently well---specifically, if the average approximation error is $o_p(n^{-1})$---then the emulator-based estimator inherits the consistency and asymptotic normality of the exact MLE. When this condition is not met, we show that valid inference can still be obtained for a pseudo-true parameter via sandwich standard errors, treating the emulator as a working model in the quasi-ML framework. 

By careful design, our preprocessing transformations and architecture ensure that the emulator is a smooth function of its inputs: the deterministic utilities and corresponding scale (covariance) matrix. We use Sobolev training \citep{czarnecki2017sobolev} to match emulator gradients to those of the log choice probabilities. Together, these design choices enable reliable use of automatic differentiation for downstream model-fitting and inference tasks.

With a pretrained emulator, generalizing from logit to probit (or another error distribution) requires only the replacement of closed-form softmax probabilities with emulator evaluations of choice probabilities when forming the likelihood.
In the probit case, given a fixed computational budget, we show via simulations that an ML estimator using our amortized inference procedure 
matches or exceeds the performance of one using the GHK algorithm in terms of estimation error and coverage rates. Furthermore, our approach can easily handle models other than MNP where an efficient simulation algorithm like GHK is unavailable.

For concreteness, much of our exposition focuses on MNP. However, the emulator approach and supporting theory are largely agnostic to the assumed parametric form of the errors. Generalizing to other error specifications (e.g., Gumbel or multivariate-$t$) is straightforward, requiring only modest changes to the training data generation and emulator inputs.

An application to household laundry-detergent choice illustrates the empirical importance of this flexibility. Even with rich household-level random effects, IID period-specific errors are strongly rejected, and allowing correlated period-specific errors materially changes predicted substitution relative to the conditional-IIA benchmark of conventional mixed MNL.

The remainder of this paper is organized as follows. Section~\ref{sec:literature} summarizes related literature on choice modeling. Section~\ref{sec:setup} describes the family of correlated discrete choice models we consider and their properties. Section~\ref{sec:design} presents the emulator architecture and training procedure. Section~\ref{sec:theory} establishes the theoretical properties of the architecture and emulator-based estimators. Section~\ref{sec:simulation} presents simulation results.
Section~\ref{sec:application} describes an application to retail laundry-detergent purchases, and Section~\ref{sec:conclusion} concludes.

\vspace{-6pt}
\section{Related Literature}
\label{sec:literature}

The literature applying machine learning methods to discrete choice can be divided into two broad streams. The first maintains the classic random utility model (RUM), in which choice probabilities are generated by rational consumer types with different (but stable) preference orderings over the universe of choice objects, as explained by \citet{BlockMarschak1960} and \citet{McfaddenRichter1990}.\footnote{Suppose there are $\bar{\numalt}$ objects in the universe and a consumer is presented with a choice set (or ``assortment'') that contains $\numalt \le \bar{\numalt}$ elements. A key implication of RUM is that the utility a consumer derives from product $k$ is invariant to the choice set, ruling out context or assortment effects.} MNL, MNP, and mixed MNL are all members of the RUM class provided one maintains the RUM assumption that utility from an option depends only on its own characteristics. 
Key papers in this stream retain the basic MNL structure but use neural networks to generalize the functional form of utility: \citet{Bentz2000}, \citet{Sifringer2020}, \citet{WangMoZhao2020}, \citet{han2022}, and \citet{singh2023choice}. Two important recent papers extend this work: \citet{Aouad2025} develop an architecture that implements mixed MNL with a flexible distribution of taste heterogeneity (RUMnet), while \citet{Bagheri2025} generalize the Gumbel error assumption (RUM-NN).
Both use softmax-smoothed sample averages to approximate choice probabilities within the loss function, but we show our emulators are more accurate for given computation time.

\citet{Aouad2025} build on the generic approximation property of mixed-MNL models in \citet{mcfadden2000mixed}, which relies on a flexible basis expansion of utility and a nonparametric mixing distribution. In practice, these model attributes are unknown to analysts. Furthermore, this approximation property is not unique to mixed MNL. Rather, as noted by \citet{Aouad2025} and \citet{mcfadden2000mixed}, it holds for a broader class of hierarchical discrete choice models, including mixed-MNP models.  \citet{mcfadden2000mixed} use this property to justify adoption of mixed MNL with a parametric mixing distribution as a computationally convenient alternative to other models lacking closed-form choice probabilities, such as MNP. But this justification of mixed MNL becomes less compelling given a practical, general-purpose alternative like the emulators we propose here.

The second stream dispenses with the RUM structure, often motivated by a desire to relax restrictive MNL assumptions like IIA and to allow more flexible substitution patterns.
Some papers in this stream view discrete choice as a general classification problem amenable to machine learning. Examples include \citet{WangRoss2018}, \citet{Lheritier2019}, \citet{Rosenfeld2020}, \citet{ChenMisic2022}, and \citet{Chen2025}. Others maintain MNL structure at the top level---that is, choice probabilities are determined by a vector of alternative-specific utilities that enter a softmax function---but a neural net is used to construct the alternative-specific utility functions in flexible ways that deviate from RUM assumptions. For instance, the utility of alternative $j$ is allowed to depend on attributes of other alternatives to generate context effects. This is exemplified by   
\citet{WangMoZhao2021}, \citet{Wong2021}, \citet{Cai2022}, \citet{Pfannschmidt2022}, and \citet{Berbeglia2025}.

These two streams highlight a basic tension between model flexibility and interpretability: RUM models are favored for their interpretability and grounding in economic theory, but those most widely used in empirical work (MNL and mixed MNL) impose strong and potentially unrealistic constraints on substitution patterns; in particular, they assume IIA at the level of individual choices.
The first stream adds flexibility to RUM models. But these proposals typically generalize deterministic utility while retaining IID Gumbel errors, so substitution patterns still obey IIA.
RUMnet and RUM-NN are exceptions: They allow more flexible substitution patterns, but rely on sample-average approximations to choice probabilities that we show are less accurate than our emulators for a given computation time.
The second stream abandons RUM structure altogether. This leads to very flexible models, but at the cost of interpretability, making it difficult or impossible to obtain reliable inferences for many economically meaningful quantities, such as consumer welfare, willingness-to-pay measures, demand elasticities, and substitution effects.

MNP models constitute a notable exception to the above tradeoff as they allow for flexible substitution patterns via \emph{interpretable} covariance relationships.
Moreover, if additional flexibility is desired, they can be extended to allow the deterministic utilities to follow the functional form of a neural network \citep{Hruschka2007}.
Conversely, if researchers desire a more parsimonious or interpretable model, analysts can constrain the covariance structure via penalization methods as in \citet{Jiang2025} or via factor structures as we illustrate in Section \ref{sec:simulation}.
Despite these virtues, empirical applications of MNP are relatively sparse in the literature due, in large part, to the difficulty of evaluating MNP choice probabilities.

More broadly, MNP models are just one member of a large class of RUM models with correlated errors.
Other models in this class can be generated by assuming errors of the form $\vec{\epsilon} = \cov^{1/2} \vec{\epsilon}^*$, where $\vec{\epsilon}^*$ is a vector of exchangeable errors, $\cov$ is a scale matrix, and $\cov^{1/2}$ satisfies $\cov^{1/2}(\cov^{1/2})^{\top} = \cov$. This is the general class of models we study. 
For example, taking $\epsilon_j^* \overset{\text{iid}}{\sim} \text{Gumbel}(0, 1)$ provides a natural way to generalize MNL and mixed MNL to allow correlated errors.
This correlated-Gumbel model relaxes IIA at the individual level. 
As another example, we could assume $\vec{\epsilon}^* \sim \text{Multivariate-}t(\vec{0}, \mat{I}, \nu)$ to capture heavy-tailed behavior via the degrees-of-freedom parameter $\nu$ (with $\nu > 2$ so the covariance matrix exists).
\citet{Bagheri2025} also consider correlated error distributions formed in this way, but they rely on smoothed frequency simulators to approximate choice probabilities for models---such as correlated Gumbel---for which no closed form is available. We show that our emulators provide more accurate probability evaluations for the correlated-Gumbel model, for a given computation time, potentially complementing their approach.

In summary, our work provides an emulator-based amortized inference framework that gives accurate, reusable, and computationally efficient estimates of choice probabilities for RUM models with general error distributions and correlation structures.
Our framework offers a potential resolution to the flexibility--interpretability tradeoff highlighted above, enabling practical estimation of flexible choice models without sacrificing the economic insights and parsimony of RUM models.
Our approach is complementary to recent advances in machine learning methods, including the two streams discussed above. 
In many cases, these methods could be combined with emulator-based likelihood evaluations, yielding flexible models with interpretable substitution patterns and manageable computational demands. Moreover, our approach is supported by strong theoretical results, including a universal approximation guarantee and asymptotic inference under mild approximation conditions.

\vspace{-12pt}
\section{Problem Setup}
\label{sec:setup}

This section summarizes the problem setup, including the invariance properties of discrete choice models and our inferential goals.

\vspace{-12pt}
\subsection{Discrete Choice Models}
\label{sec:choice-models}

We consider a decision-maker choosing among $\numalt$ mutually exclusive alternatives. The decision-maker assigns a latent utility $U_j$ to each alternative $j \in \{1, \ldots, \numalt\}$ and selects the alternative with the highest utility:
\[
Y = \underset{{j \in \{1, \ldots, \numalt\}}}{\argmax}\; U_j.
\]
The latent utilities decompose into deterministic and stochastic components:
\vspace{-1pt}
\[
U_j = v_j + \epsilon_j, \quad j = 1, \ldots, \numalt,
\vspace{-1pt}
\]
where $v_j$ is the deterministic (systematic) utility that depends on observable characteristics, and $\epsilon_j$ is a random error capturing unobserved factors. As explained in Section \ref{sec:intro}, the deterministic utility typically takes a linear form $v_j = \vec{x}_j^\top \vec{\beta}$, where $\vec{x}_j$ is a vector of alternative-specific attributes and $\vec{\beta}$ is a parameter vector to be estimated.

We introduce a scale matrix, $\cov$, with the following structure:\footnote{Since only utility differences affect choices, MNP is often parameterized in differences relative to a base alternative \citep{keane1992}. This is convenient because Gaussianity is preserved under differencing. For non-Gaussian errors this generally fails, so we instead use the level representation in \eqref{eq:Sigma-structure}.}
\begin{equation}\label{eq:Sigma-structure}
\cov = \begin{bmatrix}
    1 & \vec{0}_{\numalt-1}^{\top} \\
    \vec{0}_{\numalt-1} & \mat{\Omega}
\end{bmatrix},
\end{equation}
where $\mat{\Omega} \in \R^{(\numalt -1) \times (\numalt - 1)}$ and $\vec{0}_{\numalt-1}$ is a zero vector of dimension $\numalt -1$.
The scale matrix transforms the errors as 
\[\vec{\epsilon} = (\epsilon_1, \ldots, \epsilon_\numalt)^\top = \cov^{1/2} \vec{\epsilon}^*,\]
where $\vec{\epsilon}^*$ is a vector of exchangeable errors and $\cov^{1/2}$ is chosen according to a predefined rule such that $\cov^{1/2} (\cov^{1/2})^{\top} = \cov$.
For general distributions, we specify that $\cov^{1/2}$ is the (unique) symmetric positive semidefinite square-root matrix.
However, in special cases described in Definition \ref{def:factorization-invariant}, the choice of factor does not matter, so we are free to select another rule, such as the lower Cholesky factor.

Under the above setup, we obtain the MNP family with $\vec{\epsilon}^* \sim \mathcal{N}(\vec{0}_{\numalt}, \mat{I}_{\numalt})$ and $\mat{\Omega}$ unconstrained, and we obtain the MNL as the special case $\mat{\Omega} = \mat{I}_{\numalt-1}$ and $\epsilon_j^* \overset{\text{iid}}{\sim}\text{Gumbel}(a, 1)$, where $a \in \R$ is arbitrary.\footnote{The constant $a$ cancels when we consider the distribution of utility differences. To obtain mean-zero errors, we can set $a$ equal to the negative of the Euler--Mascheroni constant (typically denoted $\gamma$).} We refer to the model obtained by allowing $\mat{\Omega}$ to differ from
$\mat{I}_{\numalt-1}$ while maintaining IID Gumbel $\epsilon_j^*$ as the
``correlated-Gumbel'' model. This model relaxes IIA at the individual level.
We impose the following assumption on $\vec{\epsilon}^*$.

\vspace{6pt}
\begin{assumption}\label{assumption:epsilon-star}
The error vector $\vec{\epsilon}^*$ admits a density, $f_{\vec{\epsilon}^*}: \R^{\numalt} \to [0, \infty)$, with respect to Lebesgue measure.
Moreover, $\Cov(\vec{\epsilon}^*) = C\, \mat{I}_{\numalt}$ for some $C > 0$, and the elements of $\vec{\epsilon}^*$ are exchangeable so that $f_{\vec{\epsilon}^*}(\mat{P}_{\pi} \vec{\epsilon}') = f_{\vec{\epsilon}^*}(\vec{\epsilon}')$ for all $\vec{\epsilon}' \in \R^{\numalt}$ and any permutation matrix, $\mat{P}_{\pi}$.
\end{assumption}

In the case of MNP, $f_{\vec{\epsilon}^*}$ represents the probability density function (PDF) of the $\mathcal{N}(\vec{0}_{\numalt}, \mat{I}_{\numalt})$ distribution.
Assumption \ref{assumption:epsilon-star} places weak assumptions on $f_{\vec{\epsilon}^*}$, allowing for a wide variety of error distributions, including Gumbel, multivariate-$t$ (with $\nu > 2$), and many others.

We denote the choice probability for alternative $j$ as
\[
\choiceprob_j(\util, \cov^{1/2}) = \Prob(U_j \geq U_k \text{ for all } k \neq j) = \Prob(\epsilon_k - \epsilon_j \leq v_j - v_k \text{ for all } k \neq j),
\]
where $\util = (v_1, \ldots, v_\numalt)^\top$.
In practice, we may also choose to include additional parameters governing the distribution of $\vec{\epsilon}^*$, such as the degrees-of-freedom parameter, $\nu$, described in Section \ref{sec:literature}, but these parameters are suppressed in the notation for simplicity.
The choice probability can be expressed as the following integral:
\begin{align}
\choiceprob_j(\util, \cov^{1/2}) &= \int_{\mathcal{R}_j} f_{\vec{\epsilon}^*}(\vec{\epsilon}^*) \, d\vec{\epsilon}^*,\label{eq:choice-integral}\\
\mathcal{R}_j &= \big\{\vec{\epsilon}^* \in \R^{\numalt} : v_j + \big[\cov^{1/2}\big]_{j\bigcdot} \vec{\epsilon}^* \geq v_k + \big[\cov^{1/2}\big]_{k\bigcdot} \vec{\epsilon}^* \text{ for all } k \neq j\big\},\nonumber
\end{align}
where $\mathcal{R}_j$ is the region of the error space where alternative $j$ is chosen and $\big[\cov^{1/2}\big]_{j\bigcdot}$ denotes the $j$th row of $\cov^{1/2}$ so that $\big[\cov^{1/2}\big]_{j\bigcdot} \vec{\epsilon}^* = \epsilon_j$.
In general, this integral is not analytically tractable.
In MNP models, in particular, it does not have a closed-form solution for $\numalt \geq 3$, necessitating numerical methods, such as the GHK algorithm, which simulates these probabilities via recursive conditioning. For each evaluation, GHK requires $R$ simulation draws, with the approximation error decreasing at rate $O(R^{-1/2})$.

\vspace{-9pt}
\subsection{Invariance Properties}\label{sec:invariance-properties}

Choice probabilities satisfy three fundamental invariance properties.

\textbf{Location invariance.} Adding a constant to all utilities does not change choice probabilities:
\begin{equation}\label{eq:location-invariance}
\choiceprob_j(\util + C\ones_\numalt, \cov^{1/2}) = \choiceprob_j(\util, \cov^{1/2})
\end{equation}
where $\ones_\numalt$ is the $\numalt$-vector of ones and $C \in \R$ may be a random variable. This follows because choices depend only on utility differences.

\textbf{Scale invariance.} Scaling $\util$ and $\cov^{1/2}$ preserves choice probabilities:
\begin{equation}\label{eq:scale-invariance}
\choiceprob_j(C \util, C \cov^{1/2}) = \choiceprob_j(\util, \cov^{1/2}) \quad \text{for all } C > 0.
\end{equation}
This follows because the scaled model is equivalent to the original model with $\vec{v}$ and $\vec{\epsilon}$ replaced by $C \vec{v}$ and $C \vec{\epsilon}$, respectively.

\textbf{Permutation equivariance.} Relabeling alternatives permutes choice probabilities. For any permutation $\pi$ of $\{1, \ldots, \numalt\}$, let $\mat{P}_\pi$ denote the permutation matrix satisfying $(\mat{P}_\pi \util)_i = v_{\pi(i)}$, so position $i$ in the relabeled problem corresponds to original alternative $\pi(i)$. Then
\[
\choiceprob_j(\util, \cov^{1/2}) = \choiceprob_{\pi^{-1}(j)}\big\{\mat{P}_\pi \util, (\mat{P}_\pi \cov \mat{P}_\pi^{\top})^{1/2}\big\}.
\vspace{-12pt}
\]

Unlike the first two properties, whether this property holds depends on the rule used to form $\cov^{1/2}$ from $\cov$. It holds under the symmetric positive semidefinite square-root matrix because $(\mat{P}_\pi \cov \mat{P}_\pi^{\top})^{1/2} = \mat{P}_\pi \cov^{1/2} \mat{P}_\pi^{\top}$, so the utility vector of the relabeled problem satisfies
\[
\mat{P}_\pi \util + \left(\mat{P}_\pi \cov \mat{P}_\pi^{\top}\right)^{1/2} \vec{\epsilon}^* 
\ =\  \mat{P}_\pi \util + \mat{P}_\pi \cov^{1/2} \mat{P}_\pi^{\top} \vec{\epsilon}^*
\ \overset{d}{=}\ \mat{P}_\pi \left(\util + \cov^{1/2} \vec{\epsilon}^*\right),
\]
where the equality in distribution follows from the exchangeability of $\vec{\epsilon}^*$. In this case, the relabeled problem's utility vector is simply the permutation of the original one.

These three invariance properties are satisfied for all discrete choice models we consider.
An important subset of these models also satisfies the following special property.
\vspace{12pt}

\begin{definition}\label{def:factorization-invariant}
A discrete choice model is \emph{factorization-invariant} if, whenever two factors satisfy $\cov_1^{1/2} (\cov_1^{1/2})^{\top} = \cov_2^{1/2} (\cov_2^{1/2})^{\top}$, the choice probabilities agree: $\choiceprob_j(\util, \cov_1^{1/2}) = \choiceprob_j(\util, \cov_2^{1/2})$ for all $\util \in \R^{\numalt}$ and all $j \in \{1, \ldots, \numalt\}$.
\end{definition}

Factorization invariance allows a simplification of the neural network architecture, as discussed in Section \ref{sec:architecture}.
Notably, MNP models satisfy this property, so any permissible factorization rule defining $\cov \mapsto \cov^{1/2}$ produces an equivalent model. An example of such a rule is the lower triangular Cholesky factor. More generally, this property is satisfied by the broader class of models with elliptically symmetric errors, including multivariate-$t$.

Our method applies to any model satisfying Assumption~\ref{assumption:epsilon-star}, including the correlated-Gumbel model, in which $\vec{\epsilon}^*$ has IID Gumbel components and so the model is not factorization-invariant. In such cases, different factors of the same $\cov$ can alter higher-order moments and hence choice probabilities.
Our method remains applicable provided we use the symmetric positive semidefinite square root.
Henceforth, we suppress the dependence on the factorization rule and write the choice probabilities as $\choiceprob_j(\util,\cov)$.

A generic neural network need not respect these three invariance properties. If they are not embedded in the architecture, the network must learn them. Because a continuum of models generates equivalent choice probabilities, this can substantially slow learning and impair generalization. Section \ref{sec:design} imposes these properties directly through preprocessing for location and scale invariance and an equivariant architecture for permutation equivariance.

\subsection{Inferential Goals}

Given observations $\{(y_i, \mat{X}_i)\}_{i=1}^n$, where $y_i \in \{1, \ldots, \numalt\}$ is the observed choice and $\mat{X}_i$ represents the covariates, we seek to estimate a finite-dimensional parameter vector $\param \in \paramspace \subset \R^\dimparam$. The parameter vector includes the coefficients governing the deterministic utilities, the free parameters governing $\mat{\Omega}$ (and hence $\cov$), and, where applicable, additional parameters governing $f_{\vec{\epsilon}^*}$, such as the degrees-of-freedom parameter $\nu$. Thus, for observation $i$, the model implies a deterministic utility vector $\util_i(\mat{X}_i,\param)$ and scale matrix $\cov_i(\mat{X}_i,\param)$. Section \ref{sec:preprocessing} describes how these model-implied objects are subsequently transformed into the normalized inputs supplied to the emulator.
It is worth noting that the emulator itself does not require identification: it maps any specified $(\util,\cov)$ into choice probabilities. Identification matters only when using observed choices to estimate the parameters generating $(\util,\cov)$; we return to this issue in Section \ref{sec:estimation-theory}.
The (averaged) log-likelihood function is
\[
\loglik(\param) = \frac{1}{n} \sum_{i=1}^n \log \choiceprob_{y_i}\left\{\util_i(\mat{X}_i, \param), \cov_i(\mat{X}_i, \param)\right\}.
\]
Maximum likelihood estimation requires maximizing $\loglik(\param)$ over $\paramspace$, which in turn requires repeated evaluation of $\loglik(\param)$ and its gradient $\gradparam \loglik(\param)$.
Similarly, Bayesian methods such as Hamiltonian Monte Carlo require many repeated evaluations of both $\loglik(\param)$ and $\gradparam \loglik(\param)$ throughout the estimation process.

Our goal is to construct a neural network emulator that provides fast, accurate approximations of the true choice probabilities, along with analytic gradients via automatic differentiation. The emulator is trained once on simulated data spanning the relevant input space, after which it can be used for rapid inference on many datasets and models.
\vspace{-12pt}
\section{Emulator Design}
\label{sec:design}
\vspace{-3pt}
This section outlines the emulator preprocessing transformations, neural network architecture, training procedure, and inference process.
The emulator must satisfy four desiderata:

\begin{enumerate}[itemsep=1pt, topsep=-5pt, parsep=0pt, partopsep=0pt]
\item \textbf{Respect invariance properties.} The emulator should respect the invariance properties described in Section \ref{sec:invariance-properties}: invariance with respect to location shifts and scale transformations, and equivariance with respect to permutations.

\item \textbf{Provide smooth approximations.} The emulator should be differentiable with respect to its inputs, enabling gradient-based optimization and automatic differentiation.

\item \textbf{Generalize across specifications.} A single trained emulator should work for diverse utility and covariance structures.

\item \textbf{Be computationally efficient.} Emulator evaluation should be fast enough to enable routine use in estimation.
\end{enumerate}

We achieve these goals through a combination of preprocessing transformations and a carefully designed neural network architecture.

\vspace{-9pt}
\subsection{Preprocessing: Centering and Scaling}\label{sec:preprocessing}

We address location and scale invariance through preprocessing transformations that project the inputs onto a canonical subspace.

\textbf{Centering.} Define the centering matrix $\centermat = \mat{I}_\numalt - \frac{1}{\numalt} \ones_\numalt \ones_\numalt^\top$, which we employ to project the vector of utilities, $\vec{U} \in \R^{\numalt}$, onto the subspace orthogonal to $\ones_\numalt$, resulting in transformed utilities $\tilde{\vec{U}} = \centermat \left(\vec{v} + \vec{\epsilon}\right) = \vec{v} + \vec{\epsilon} - \vec{1}_{\numalt} (\bar{v} + \bar{\epsilon}) = \vec{U} - \vec{1}_{\numalt} \bar{U}$, where $\bar{v} = \frac{1}{\numalt} \sum_{j=1}^{\numalt} v_j$, $\bar{\epsilon} = \frac{1}{\numalt} \sum_{j=1}^{\numalt} \epsilon_j$, and  $\bar{U} = \bar{v} + \bar{\epsilon}$.
These transformed utilities sum to zero within each realization: $\sum_{j=1}^{\numalt} \tilde{U}_j = 0$.

By the location invariance property in \eqref{eq:location-invariance}, this transformation does not alter the choice probabilities.
Further, under the assumed structure of $\cov$ in \eqref{eq:Sigma-structure}, the transformation $\cov \mapsto \centermat \cov \centermat^{\top}$ can be inverted as $\mat{\Omega} = \mat{F} \centermat \cov \centermat^{\top} \mat{F}^{\top} - \vec{1}_{\numalt-1} \vec{1}_{\numalt-1}^{\top}$, where $\mat{F}$ is the differencing matrix $\mat{F} = \begin{bmatrix} -\vec{1}_{\numalt-1} & \mat{I}_{\numalt-1}\end{bmatrix}$.
Consequently, distinct models remain separated and we can represent a transformed model with the following modified parameters:
\[
\tilde{\util} = \centermat \util = \util - \ones_{\numalt} \bar{v},\qquad \tilde{\cov} = \centermat \cov \centermat.
\]
\textbf{Scaling.} We normalize by the trace of $\tilde{\cov}$, producing transformed utilities $\vec{U}^* = \sqrt{\Cstar} \tilde{\vec{U}}$, where $\Cstar = \numalt / \tr\big(\tilde{\cov}\big)$.
By the scale invariance property in \eqref{eq:scale-invariance}, this normalization does not affect the choice probabilities.
Further, the map $(\tilde{\util}, \tilde{\cov}) \mapsto \left(\Cstar, \sqrt{\Cstar} \tilde{\util}, \Cstar \tilde{\cov}\right)$ is invertible; thus, the final transformed model can be represented with parameters:
\[
\Cstar = \numalt / \tr\big(\tilde{\cov}\big), \qquad
\utilstar = \sqrt{\Cstar} \, \tilde{\util}, \qquad
\covstar = \Cstar \, \tilde{\cov}.
\]
We can omit the case $\tr(\tilde{\cov}) = 0$ as this results in known, deterministic choices.\footnote{Assumption \ref{assumption:epsilon-star} implies $\Cov(\vec{U} - \vec{1}_{\numalt} \bar{U}) = C \centermat \cov \centermat = C \tilde{\cov}$ for some $C > 0$. Thus, $\tr(\tilde{\cov}) = 0$ implies that $\epsilon_j = \bar{\epsilon}$ with probability one for all $j \in \{1, \ldots, \numalt\}$.} The complete transformation produces normalized parameters $(\Cstar, \utilstar, \covstar)$ satisfying:
\begin{enumerate}
\item $\sum_{j=1}^\numalt v_j^* = 0$ (centered deterministic utilities),
\item $\covstar \ones_\numalt = \vec{0}$ (centered scale matrix),
\item $\tr(\covstar) = \numalt$ (scale-normalized),
\item $\covstar$ is positive semidefinite with rank at most $\numalt - 1$.
\end{enumerate}

This normalization reduces the size of the emulator's input space, which accelerates learning and improves generalization.
Importantly, this transformation commutes with permutation matrices: for any permutation matrix $\mat{P}_\pi$, applying the transformation to $(\mat{P}_\pi \util, \mat{P}_\pi \cov \mat{P}_\pi^\top)$ yields $(\mat{P}_\pi \utilstar, \mat{P}_\pi \covstar \mat{P}_\pi^\top)$. Thus, the preprocessing preserves permutation equivariance.

In factorization-invariant models (Definition \ref{def:factorization-invariant}), the pair $(\utilstar, \covstar)$ is sufficient to identify the choice probabilities, so the scale factor $\Cstar$ is not needed. 
For general exchangeable errors (e.g., Gumbel errors), we need two additional pieces of information: $\Cstar$ and an indicator for which alternative corresponds to the base alternative (index $1$) in the level parameterization defined in \eqref{eq:Sigma-structure}.
The latter is required to invert the preprocessing transformation to an equivalent model in the original parameterization; without it, the preprocessing transformation can result in information loss.\footnote{The inversion formula $\mat{\Omega} = \mat{F} \centermat \cov \centermat^{\top} \mat{F}^{\top} - \vec{1}_{\numalt-1} \vec{1}_{\numalt-1}^{\top}$ implicitly requires knowledge of which alternative occupies the first position. This information allows us to recover $\cov$. Although we cannot recover $\util$, we can recover $\tilde{\util} = \utilstar/\sqrt{\Cstar}$, which differs from $\util$ by an additive constant and, hence, results in equivalent choices.}

\begin{comment}
Suppose we are given two models with transformed parameters $(\Cstar_1, \utilstar_1, \covstar_1)$ and $(\Cstar_2, \utilstar_2, \covstar_2)$ such that $\utilstar_1 = \utilstar_2$ and $\covstar_1 = \covstar_2$ but $\Cstar_1 \neq \Cstar_2$.
Let $\lambda \equiv \Cstar_2/\Cstar_1$. Since $\utilstar_1=\utilstar_2$ and $\covstar_1=\covstar_2$, we have
\[
\centermat \util_1 = \sqrt{\lambda}\,\centermat \util_2,
\qquad
\centermat \cov_1 \centermat = \lambda\,\centermat \cov_2 \centermat.
\]
Now define
\[
A_1 \equiv \centermat \cov_1^{1/2},
\qquad
A_2 \equiv \sqrt{\lambda}\,\centermat \cov_2^{1/2}.
\]
Then
\[
A_1A_1^\top = \centermat \cov_1 \centermat
= \lambda\,\centermat \cov_2 \centermat
= A_2A_2^\top.
\]
By Definition~\ref{def:factorization-invariant},
\[
\choiceprob_j(\centermat\util_1,A_1)
=
\choiceprob_j(\sqrt{\lambda}\,\centermat\util_2,A_2).
\]
Finally, location invariance implies that replacing $(\util_i,\cov_i^{1/2})$ by $(\centermat\util_i,\centermat \cov_i^{1/2})$ does not change choice probabilities, and scale invariance implies that replacing $(\centermat\util_2,\centermat \cov_2^{1/2})$ by $(\sqrt{\lambda}\,\centermat\util_2,\sqrt{\lambda}\,\centermat \cov_2^{1/2})$ does not change them either. Hence
\[
\choiceprob_j(\util_1,\cov_1^{1/2})
=
\choiceprob_j(\util_2,\cov_2^{1/2})
\]
for all $j$.
\end{comment}

\vspace{-6pt}
\subsection{Neural Network Architecture}
\label{sec:architecture}

After preprocessing, we construct a neural network that maps $(\Cstar, \utilstar, \covstar)$ to choice probabilities while maintaining permutation equivariance. The architecture consists of three components: a per-alternative encoder, permutation-equivariant layers, and an output layer.
Below, we interpret $\covstar$ as a covariance matrix to simplify the exposition, referring to its diagonal and off-diagonal elements as variances and covariances, respectively.

\textbf{Per-alternative encoder.} For each alternative $j \in \{1, \ldots, \numalt\}$, we construct a representation $\vec{z}_j \in \R^{d_z}$ that captures how alternative $j$ relates to all other alternatives. This representation is built from two complementary DeepSet networks---a diagonal DeepSet and an off-diagonal DeepSet---whose outputs are combined with alternative $j$'s own features. The diagonal DeepSet processes pairwise relationships between $j$ and each other alternative, while the off-diagonal DeepSet summarizes the covariance structure among alternatives other than $j$. In the descriptions below, the features labeled as ``base inputs'' are sufficient for universal approximation of choice probabilities (see Section \ref{sec:universal-approximation}). We include additional features to improve the expressivity of the emulator, as we find this improves its predictions.

\textbf{Diagonal DeepSet.} For each pair $(j, k)$ with $k \neq j$, we construct a feature vector $\vec{d}_{jk}$ containing differentiable functions of the utilities and covariances associated with alternatives $j$ and $k$. The base inputs are the utilities $v_j^*$ and $v_k^*$, the variances $\Sigma_{jj}^*$ and $\Sigma_{kk}^*$, and the covariance $\Sigma_{jk}^*$. From these, we derive additional feature transformations that facilitate learning, including standard deviations $\sigma_j = \sqrt{\Sigma_{jj}^*}$ and $\sigma_k = \sqrt{\Sigma_{kk}^*}$, the correlation $\rho_{jk} = \Sigma_{jk}^*/(\sigma_j \sigma_k)$, and the standardized utility difference (z-score):
\[
z_{jk} = \frac{v_j^* - v_k^*}{\sqrt{\Sigma_{jj}^* + \Sigma_{kk}^* - 2\Sigma_{jk}^*}}.
\]
For models that are not factorization-invariant, we must also include an indicator function, $1_{\{k=1\}}$, to ensure that the preprocessing transformation does not result in information loss.\footnote{Section \ref{sec:preprocessing} describes how the transformed parameters can be inverted to an equivalent model with parameters $\tilde{\util}$ and $\cov$. However, this transformation can only be applied if we know which alternative occupies the first position. The indicator function $1_{\{k=1\}}$ provides exactly this information.}
Appendix \ref{sec:input-feature-details} provides the full set of input features for each module of the encoder.
It also describes transformations we applied to these features to reduce sensitivity to outliers.

The diagonal DeepSet processes these features using the architecture of \citet{zaheer2017deep}:
\[
\vec{h}_j^{\text{diag}} = \rho_{\text{diag}}\left( \frac{1}{\numalt - 1} \sum_{k \neq j} \phi_{\text{diag}}(\vec{d}_{jk})\right),\vspace{-6pt}
\]
where $\phi_{\text{diag}}$ is a multi-layer perceptron (MLP) applied to each pair, the average aggregates over all $k \neq j$, and $\rho_{\text{diag}}$ is another MLP that produces a learned nonlinear representation of the average.\footnote{\citet{zaheer2017deep} aggregate by summing. We average instead, which differs from summing only by a factor that is constant for fixed $\numalt$; averaging keeps the scale of the aggregated features stable when a single emulator is trained across several values of $\numalt$ (see Section \ref{sec:training-procedure}).} We use the notation $\phi_{\text{diag}}$ and $\rho_{\text{diag}}$ following \citet{zaheer2017deep}; the MLPs $\rho_{\text{diag}}$ and $\rho_{\text{off}}$ (defined below) can be distinguished from the correlation $\rho_{jk}$ by their subscripts.

\citet{singh2023choice} also employ a variant of the DeepSet architecture. However, their method assumes IID errors and operates on problem-specific features, which requires their neural network to be modified and retrained to accommodate new data structures.  In contrast, our architecture relies on generic RUM model inputs, such as deterministic utilities.

\textbf{Off-diagonal DeepSet.} For each pair $(k, l)$ with $k < l$ and $k, l \neq j$, we construct a feature vector $\vec{o}_{kl}$ containing differentiable, symmetric functions comparing alternatives $k$ and $l$; we require symmetry because the pairs do not possess a natural ordering. The base input is $\Sigma_{kl}^*$, and we include other derived features, such as the correlation $\rho_{kl}$, squared utility difference $(v_k^* - v_l^*)^2$, the squared z-score $z_{kl}^2$, and the variance sum $\Sigma_{kk}^* + \Sigma_{ll}^*$.

The off-diagonal DeepSet has an analogous structure to that of the diagonal DeepSet:
\[
\vec{h}_j^{\text{off}} = \rho_{\text{off}}\left( \frac{2}{(\numalt - 1)(\numalt - 2)} \sum_{\substack{k < l \\ k, l \neq j}} \phi_{\text{off}}(\vec{o}_{kl})\right),
\]
where the average is taken over the $(\numalt - 1)(\numalt - 2)/2$ pairs of rivals.

\textbf{Combining network.} The outputs of both DeepSets are concatenated with additional ``pass-through features'' $\vec{s}_j$, such as $\Cstar$, $1_{\{j=1\}}$, $v_j^*$, $\sigma_j$, and summary statistics of alternative $j$'s covariances with other alternatives. Of these features, only $\Cstar$ is required for universal approximation, but it can be omitted for factorization-invariant models.
These concatenated features are then passed to the combining MLP, $\zeta$, which produces the per-alternative representation $\vec{z}_j = \zeta\big( \vec{h}_j^{\text{diag}}, \vec{h}_j^{\text{off}}, \vec{s}_j \big)$.

\textbf{Permutation-equivariant layers.} The per-alternative representations are stacked into a matrix $\mat{Z} \in \R^{\numalt \times d_z}$ and processed through linear permutation-equivariant layers \citep{zaheer2017deep}, allowing information exchange across alternatives. These layers take the form
\[
\operatorname{act}\left( \mat{Z} \mat{W}_1 + \frac{1}{\numalt} \ones_\numalt \ones_\numalt^\top \mat{Z} \mat{W}_2 + \ones_\numalt \vec{c}^\top \right),
\]
where $\mat{W}_1, \mat{W}_2 \in \R^{d_{\text{in}} \times d_{\text{out}}}$ and $\vec{c} \in \R^{d_{\text{out}}}$ are learnable parameters, and $\operatorname{act}(\cdot)$ is a smooth activation function applied elementwise. The first term applies a per-alternative transformation, the second aggregates information across alternatives, and the third adds a shared bias (intercept).

\textbf{Output layer.} The final layer produces one logit per alternative:
\[
\vec{\mathrm{logit}} = \mat{Z}^{\text{final}} \vec{a},
\]
where $\mat{Z}^{\text{final}} \in \R^{\numalt \times d_{\text{final}}}$ is the output of the last equivariant layer and $\vec{a} \in \R^{d_{\text{final}}}$, so that the output $\vec{\mathrm{logit}}$ is a vector in $\R^{\numalt}$. We omit the other terms because they produce a constant shift that does not affect the output probabilities.

Choice probabilities are obtained by applying the softmax function to the output layer:
\[
\choiceprobhat_j = \frac{\exp(\mathrm{logit}_j)}{\sum_{k=1}^\numalt \exp(\mathrm{logit}_k)}.
\]
For any vector of valid choice probabilities $(P_1, \ldots, P_{\numalt})$ on the interior of the probability simplex $\Delta^{\numalt-1}$, there exists a corresponding vector of logits yielding $(P_1, \ldots, P_{\numalt})$ via softmax, so this parameterization entails no loss of generality.
Figure~\ref{fig:architecture} illustrates the architecture.

\begin{figure}[t]
\centering
\caption{Neural Network Emulator Architecture}
\label{fig:architecture}
\begin{tikzpicture}[
    node distance=1.0cm and 1.5cm,
    box/.style={rectangle, draw, rounded corners, minimum width=2.4cm, minimum height=0.8cm, align=center, font=\small},
    smallbox/.style={rectangle, draw, rounded corners, minimum width=2.0cm, minimum height=0.8cm, align=center, font=\footnotesize, text depth=0.25ex},
    inputbox/.style={rectangle, draw, rounded corners, minimum width=2.0cm, minimum height=0.7cm, align=center, font=\small, fill=gray!5},
    outputbox/.style={rectangle, draw, rounded corners, minimum width=2.4cm, minimum height=0.8cm, align=center, font=\small, fill=gray!10},
    arrow/.style={->, >=stealth, thick},
    dashedbox/.style={rectangle, draw, dashed, rounded corners},
]

% Input at bottom
\node[inputbox] (input) {$\left(\Cstar,\, \utilstar,\, \covstar\right)$};

% Three parallel branches
\coordinate[above=1.6cm of input] (branchcenter);
\node[smallbox, anchor=center] at ([xshift=-2.8cm]branchcenter) (diag) {Diagonal\\DeepSet};
\node[smallbox, anchor=center] at (branchcenter) (offdiag) {Off-Diagonal\\DeepSet};
\node[smallbox, anchor=center] at ([xshift=2.8cm]branchcenter) (features) {Pass-through\\Features: $\vec{s}_j$};

% Combine
\node[smallbox, above=1.0cm of offdiag] (combine) {Combining MLP};

% Per-alternative representation
\node[box, above=0.8cm of combine] (zi) {$\vec{z}_j \in \R^{d_z}$};

% Dashed box around encoder
\node[dashedbox, fit=(diag)(offdiag)(features)(combine)(zi), inner ysep=20pt, inner xsep=18pt] (encoderbox) {};

% Label for encoder box
\node[anchor=north west, font=\small\itshape, yshift=-1pt] at (encoderbox.north west) {Per-Alternative Encoder};
\node[anchor=north west, font=\small\itshape, yshift=-14pt] at (encoderbox.north west) {$j = 1, \ldots, \numalt$};

% Stack representations
\node[box, above=1.0cm of zi] (stack) {Stack: $\mat{Z} \in \R^{\numalt \times d_z}$};

% Equivariant layers
\node[box, above=0.8cm of stack] (equiv) {Equivariant Layers};

% Output
\node[outputbox, above=0.8cm of equiv] (output) {Logits $\in \R^\numalt$};

% Arrows
\draw[arrow] (input) -- (diag);
\draw[arrow] (input) -- (offdiag);
\draw[arrow] (input) -- (features);
\draw[arrow] (diag) -- (combine);
\draw[arrow] (offdiag) -- (combine);
\draw[arrow] (features) -- (combine);
\draw[arrow] (combine) -- (zi);
\draw[arrow] (zi) -- (stack);
\draw[arrow] (stack) -- (equiv);
\draw[arrow] (equiv) -- (output);

\end{tikzpicture}

\floatnote{For each alternative $j$, the per-alternative encoder combines diagonal features, off-diagonal features, and pass-through features to form $\vec{z}_j$. The stacked representations are then processed through permutation-equivariant layers to produce the logits.}
\end{figure}
\vspace{2pt}
\textbf{Activation functions.} Throughout the architecture, we use smooth (infinitely differentiable) activation functions. This ensures that the emulator output is a smooth function of the inputs, enabling accurate gradient computation for downstream model fitting. We found that the Swish activation function $\text{Swish}(x) = x / (1 + e^{-x})$ works well empirically \citep{ramachandran2018searching}, though other smooth activation functions could also be applied, such as $\tanh$ or $\softplus$ \citep{dugas2000softplus}.

\textbf{Computational complexity.} A naive implementation of the per-alternative encoder that recomputes the off-diagonal DeepSet aggregate $\sum_{k<l:\, k,l\neq j}\phi_{\text{off}}(\vec{o}_{kl})$ separately for each alternative $j$ would require $O(\numalt^3)$ operations. But most of these computations are redundant. Instead, we compute each pair contribution $\phi_{\text{off}}(\vec{o}_{kl})$ once, form the total sum over all pairs, and then subtract the terms involving alternative $j$ to obtain the required sum for that alternative. This reduces the complexity to $O(\numalt^2)$.

\textbf{Network size.} The architecture can be compact. In our simulations, we use networks with 4--64 hidden units and 0--2 hidden layers per component. In general, larger values of $\numalt$ require larger networks for accurate approximation.

\textbf{Relation to \citet{blumsmith2025machine}.} Our architecture builds on the DS-CI architecture of \citet{blumsmith2025machine}, which is designed to approximate permutation-invariant scalar functions of symmetric matrices. In contrast, our emulator approximates the full vector of choice probabilities, which depends jointly on the systematic utility vector $\util$ and the matrix $\cov$ governing the error distribution. If the alternatives are relabeled, the choice probabilities must be relabeled in the same way, so the choice-probability map is permutation-equivariant rather than permutation-invariant. These differences motivate three architectural changes. First, the equivariance property requires us to construct a separate representation for each alternative $j$, distinguishing it from its rivals $k \neq j$ in the per-alternative encoder. Second, these separate representations and the presence of $\util$ require us to pass $v_k^*$ and $\Sigma_{jk}^*$ to the diagonal DeepSet (in addition to the diagonals, $\Sigma_{kk}^*$, which DS-CI also requires). Third, whereas \citet{blumsmith2025machine} require a ``linking feature'' that combines diagonal and off-diagonal elements, our centering procedure makes this feature redundant. However, we require a different scalar-valued feature, $\Cstar$, for universal approximation (except on factorization-invariant models). Rather than processing $\Cstar$ in its own submodule, we simply pass it as an additional feature to the combining MLP.

\vspace{-6pt}
\subsection{Training Procedure}
\label{sec:training-procedure}

We train the emulator using simulated data spanning a diverse range of utility and covariance configurations, employing Sobolev training to ensure accurate approximation of both choice probabilities and their gradients.

\textbf{Data generation.} To generate each training example for our simulations in Section \ref{sec:simulation} involving the MNP model, we follow the procedure below:
\begin{enumerate}
\vspace{-3pt}
\item \textbf{Error covariance.} Draw a covariance matrix $\mat{\Omega}$ from a $\text{Wishart}(\mat{I}_{\numalt-1}/(\numalt+1), \numalt+1)$ distribution and form $\cov$ according to \eqref{eq:Sigma-structure}.
\vspace{-3pt}
\item \textbf{Raw deterministic utilities.} Draw independent utilities $\util' \sim \mathcal{N}(\vec{0}, 4\mat{I}_\numalt)$.
\vspace{-3pt}
\item \textbf{Deterministic utility covariance.} Draw a second covariance matrix $\mat{\Omega}'$ (independently from the same distribution as $\mat{\Omega}$) to form $\cov'$ as in \eqref{eq:Sigma-structure}, draw $\omega \sim 
\text{Uniform}(0,1)$, and form the convex combination $\cov'' = \omega \cov + (1 - \omega) \cov'$. 
\vspace{-3pt}
\item \textbf{Transformed deterministic utilities.} Transform the utilities to have covariance structure correlated with $\cov''$: $\util = (\cov'')^{1/2} \util'$. This yields $\Cov(\util) = 4\, \cov''$.\footnote{If desired, we can also center $\util$ as this results in an equivalent model.}
\vspace{-3pt}
\item \textbf{Preprocessing.} Apply the preprocessing transform to obtain $(\utilstar, \covstar)$ from $(\util, \cov)$.
\vspace{-3pt}
\item \textbf{Simulated choices.} Simulate $10^6$ choices from the MNP model with utilities $\util$ and scale matrix $\cov$ (or an equivalent parameterization).
\end{enumerate}

Steps 3 and 4 ensure diversity in the relationship between deterministic utilities and error covariances. When $\omega = 1$, the utilities are perfectly aligned with the covariance structure; when $\omega = 0$, they are independent. Intermediate values produce partial correlation.

For error distributions that are not factorization-invariant, we must additionally compute $\Cstar$ in step 5 and replace the MNP model with the assumed error distribution in step 6.
The other choices listed above, including the specific distributions and number of choices ($10^6$), may also be modified based on the problem domain and available computational resources.
\vspace{6pt}

\textbf{Loss function.} We employ Sobolev training \citep{czarnecki2017sobolev}, which augments the standard likelihood loss with a gradient-matching penalty. The total loss is:
\[
\mathcal{L}(\netparam) = \mathcal{L}_{\text{CE}}(\netparam) + \lambda_{\text{grad}} \mathcal{L}_{\text{grad}}(\netparam),
\]
where $\netparam$ represents the network parameters and $\lambda_{\text{grad}} > 0$ controls the weight of the gradient-matching term.  The cross-entropy component is the multinomial negative log-likelihood (up to constants):
\vspace{3pt}
\[
\mathcal{L}_{\text{CE}}(\netparam) = -\sum_{j=1}^\numalt \hat{\choiceprob}_j^{\text{sim}} \log \choiceprobhat_j(\Cstar, \utilstar, \covstar; \netparam),
\]
where $\hat{\choiceprob}_j^{\text{sim}} = n_j / \sum_{k=1}^{\numalt} n_k$ is the simulated choice frequency and $\choiceprobhat_j(\Cstar, \utilstar, \covstar; \netparam)$ is the emulator choice probability of alternative $j$.
\vspace{2pt}

For the gradient-matching component, we follow the stochastic Sobolev training approach of \citet{czarnecki2017sobolev}, which avoids computing full Jacobian matrices by instead matching directional derivatives along random directions.
Let $\vec{\delta} = (\vec{\delta}_{\util}, \vech(\mat{\Delta}_{\mat{L}}))$ denote a random direction, where $\vec{\delta}_{\util} \in \R^\numalt$, $\mat{\Delta}_{\mat{L}} \in \R^{(\numalt-1) \times (\numalt-1)}$ is lower-triangular,  and $\vech(\cdot)$ denotes the half-vectorization (lower triangular elements). The gradient-matching loss is:
\[
\mathcal{L}_{\text{grad}}(\netparam) = \sum_{j=1}^\numalt \hat{\choiceprob}_j^{\text{sim}} \left( \nabla_{\vec{\delta}} \log \choiceprobhat_j - \nabla_{\vec{\delta}} \log \choiceprob_j^{\text{target}} \right)^2,
\]
where $\nabla_{\vec{\delta}}$ denotes the directional derivative along $\vec{\delta}$. Weighting by the simulated choice frequencies $\hat{\choiceprob}_j^{\text{sim}}$ focuses learning on alternatives that are frequently chosen, whose gradients most influence parameter estimation in downstream inference tasks.

\textbf{Constraint-respecting directions.} The inputs $(\utilstar, \covstar)$ lie on the constrained set:
\begin{equation}
\label{eq:utilcovspace}
\utilcovspace = \left\{ (\utilstar, \covstar) : \sum_{j=1}^\numalt v_j^* = 0, \; \covstar \ones_\numalt = \vec{0}, \; \tr(\covstar) = \numalt, \; \covstar \succeq \mat{0} \right\}.
\end{equation}
To ensure these constraints are respected, we compute directions in the latent space of $\util$ and $\mat{L} = \chol(\mat{\Omega})$, the lower Cholesky factor of $\mat{\Omega}$.\footnote{Note that $\mat{L}$ does not necessarily correspond with $\mat{\Omega}^{1/2}$. These matrix factors play different roles. Whereas $\mat{\Omega}^{1/2}$ is used within the choice model to create covariance relationships, $\mat{L}$ is used as a computational device to ensure that the gradients respect the positive semidefiniteness of $\mat{\Omega}$ and the resulting constraints on the set $\utilcovspace$.}
We generate directions as follows:
\begin{enumerate}
\item Sample $\vec{\delta}_{\util} \sim \mathcal{N}(\vec{0}_{\numalt}, \mat{I}_\numalt)$.
\item Sample a lower-triangular matrix, $\mat{\Delta}_{\mat{L}}$, with nonzero entries in row $j$ distributed as $\mathcal{N}(\vec{0}_{j}, \mat{I}_{j} / j)$.
\item Set $\vec{\delta} = (\vec{\delta}_{\util}, \vech(\mat{\Delta}_{\mat{L}}))$ and normalize: $\vec{\delta} \leftarrow \vec{\delta} / \|\vec{\delta}\|$.
\end{enumerate}
In step 2, we divide the variances by $j$ to account for the number of nonzero entries in each row, maintaining a comparable scale across rows and columns of $\mat{\Omega} = \mat{L} \mat{L}^{\top}$.
If desired, we can remove directions that do not contribute to choice probabilities, but doing so does not affect the directional derivatives.\footnote{For example, we can set $\vec{\delta}_{\util} \leftarrow \centermat \vec{\delta}_{\util}$. However, the removed direction is orthogonal to the Jacobian of the choice probabilities, so this operation has no effect on the directional derivatives.}

\textbf{Computing directional derivatives.} Collect the latent inputs $\vec{v}$ and $\mat{L}$ in a single vector, $\vec{\xi} = (\vec{v}, \vech(\mat{L}))$, and let $\operatorname{pre}$ denote the preprocessing map from $\vec{\xi}$ to $(\Cstar, \utilstar, \covstar)$.
We approximate the emulator's directional derivatives via finite differences as
\[
\nabla_{\vec{\delta}} \log \choiceprobhat_j
\approx
\frac{\log \choiceprobhat_j\left\{\operatorname{pre}(\vec{\xi} + h \vec{\delta}); \netparam\right\} - \log \choiceprobhat_j\left\{\operatorname{pre}(\vec{\xi} - h \vec{\delta}); \netparam\right\}}{2 h},
\]
where $h > 0$ is a small step size (we use $h = 10^{-5}$).
This approach is computationally efficient, requiring only two additional forward passes per direction.
Alternatively, one could compute the directional derivative via automatic differentiation as in \citet{czarnecki2017sobolev}.
Working with the latent inputs, $\vec{\xi}$, guarantees that the resulting values of $\utilstar$ and $\covstar$ obey the constraints of $\utilcovspace$.

Target directional derivatives are computed from pre-stored Jacobians.
During data generation, we compute and store the full Jacobians $\nabla_{\vec{\xi}} \log \choiceprob_j^{\text{target}}$. During training, the target directional derivative is obtained via the inner product $\nabla_{\vec{\delta}} \log \choiceprob_j^{\text{target}} =
\vec{\delta}^{\top}\;
\nabla_{\vec{\xi}} \log \choiceprob_j^{\text{target}}$.

\textbf{Target gradient computation.} Computing target gradients for Sobolev training requires differentiating through the choice simulation process. Since discrete choices are inherently non-differentiable, we employ a smooth relaxation using a temperature-scaled softmax. For each Monte Carlo draw $r$, we form the utility vector on the original (level) scale of
\eqref{eq:Sigma-structure}, $\vec{U}^{(r)} = \util + \cov^{1/2}\vec{\epsilon}^{*(r)}$, where
$\vec{\epsilon}^{*(r)}$ is the $r$th draw of the exchangeable errors used to simulate choices
in step 6, and we approximate the hard choice indicator with:
\vspace{4pt}
\[
\tilde{Y}_j^{(r)} = \frac{\exp(\tau U_j^{(r)})}{\sum_{k=1}^\numalt \exp(\tau U_k^{(r)})},\vspace{4pt}
\]
where $\tau > 0$ is a temperature parameter. As $\tau \to \infty$, this converges to the hard choice indicator; for finite $\tau$, it provides a differentiable approximation.
We use $\tau = 5$ in our experiments, which provides a balance between gradient accuracy (low bias) and computational stability.
The target log-probability is then:
\[
\log \choiceprob_j^{\text{target}} = \log \left( \frac{1}{R} \sum_{r=1}^R \tilde{Y}_j^{(r)} \right),
\]
and the target gradient $\nabla_{\vec{\xi}} \log \choiceprob_j^{\text{target}}$ is obtained via automatic differentiation through this expression.
This approach is closely related to the ``Gumbel--softmax trick'' frequently used for gradient estimation with neural networks \citep{jang2017categorical, maddison2017concrete}.

\textbf{Training details.} Our training procedure makes use of a \emph{replay buffer}, a common technique in deep learning \citep{Mnih2015}. It stores precomputed training examples including raw inputs $(\util, \mat{L})$, simulated choice frequencies, and target Jacobians. At each optimizer step, we sample inputs from the replay buffer, sample directions, compute the loss, and update the network parameters via stochastic gradient descent or a related optimization method. The replay buffer accelerates training by allowing reuse of expensive high-precision simulations and gradient computations. We set the gradient penalty weight $\lambda_{\text{grad}} = 10^{-4}$ in our experiments; empirically, we find that this value provides sufficiently well-behaved derivatives.

\textbf{Multi-$\numalt$ training.} Remarkably, none of the weight vectors and matrices has a dimension depending on $\numalt$. The dimension $\numalt$ enters the neural network only through the input data dimension; thus, we can design a single emulator capable of processing multiple values of $\numalt$ simultaneously. The training procedures described above can then be adapted to iterate over a range of $\numalt$ values, say $\numalt \in \{3, 4, 5\}$, and the resulting emulator can then learn to approximate choice probabilities for a varying number of choices (see Section \ref{sec:sim-multi-k}). Under this configuration, we include $\numalt$ as an element of $\vec{s}_j$. 

\subsection{Inference with the Trained Emulator}

Once trained, the emulator enables rapid inference. Given a new dataset $\{(y_i, \mat{X}_i)\}_{i=1}^n$ and candidate parameter vector $\param$:

\begin{enumerate}[itemsep=1pt, topsep=2pt, parsep=0pt, partopsep=0pt]
\item Compute the deterministic utility vector $\util_i(\mat{X}_i, \param)$ and scale matrix $\cov_i(\mat{X}_i, \param)$.
\item Apply the preprocessing transformation to obtain $(\Cstar_i, \utilstar_i, \covstar_i)$.
\item Evaluate the emulator to obtain estimated choice probabilities $\choiceprobhat_{ij}$ for all alternatives.
\item Compute the log likelihood $\loglikhat(\param) = \frac{1}{n} \sum_{i=1}^n \log \choiceprobhat_{iy_i}$.
\end{enumerate}

As the preprocessing transformation and neural network are both smooth functions, the gradients of the emulator log likelihood with respect to $\param$ can be computed exactly via automatic differentiation. This enables the use of gradient-based optimizers for maximum likelihood estimation, as well as gradient-based sampling methods such as Hamiltonian Monte Carlo for Bayesian inference. The smoothness of the emulator also facilitates computation of standard errors, which can be estimated via emulator gradients (see Sections \ref{sec:estimation-theory}--\ref{sec:misspecification}).

\vspace{-6pt}
\section{Theory}
\label{sec:theory}

This section establishes the theoretical foundations for the emulator architecture and the statistical properties of emulator-based discrete choice estimators.

\vspace{-6pt}
\subsection{Universal Approximation}
\label{sec:universal-approximation}

We first show that our architecture can universally approximate choice probabilities. 
The key insight is that the base inputs processed by the per-alternative encoder's diagonal and off-diagonal DeepSets generically determine $(\utilstar,\covstar)$ up to the relevant symmetry group (i.e., permutations of rivals). In order to prove this result, we define, for each alternative $j$, the invariant function $\invarutil_j : \utilcovspace \to (\text{Multiset}, \text{Multiset})$ by:\vspace{-4pt}
\begin{align}
\mathcal{T}_j(\utilstar, \covstar) &= \{\!\{ (v_k^*, \Sigma_{kk}^*, \Sigma_{jk}^*) : k \neq j \}\!\}, \label{eq:diagonal-multiset}\\
\mathcal{O}_j(\covstar) &= \{\!\{ \Sigma_{kl}^* : 1 \leq k < l \leq \numalt, \, k \neq j, \, l \neq j \}\!\}, \label{eq:offdiagonal-multiset}\\
\invarutil_j(\utilstar, \covstar) &= (\mathcal{T}_j(\utilstar, \covstar), \mathcal{O}_j(\covstar)).\nonumber \vspace{-4pt}
\end{align}
The multiset $\mathcal{T}_j$ collects information about all the rivals of focal alternative $j$. It consists of $\numalt - 1$ triples, each containing  alternative $k$'s utility,  variance, and covariance with alternative $j$, for all $k \ne j$. It is a multiset (rather than a set or ordered list) because the order does not matter, but repeated elements do matter. The multiset $\mathcal{O}_j$ collects information about the similarity of these rivals. It consists of the off-diagonal covariances among alternatives other than $j$. These two multisets, which serve as base inputs to the encoder, generically contain enough information to reconstruct the full
normalized utility/covariance object $(\utilstar,\covstar)$, up to relabeling
of the rivals, as we now show:\footnote{It is worth noting that, while $\mathcal{T}_j$ and $\mathcal{O}_j$ do not explicitly include
$v_j^*$ or $\Sigma_{jj}^*$, these entries are pinned down by the normalization
constraints $\sum_{\ell=1}^{\numalt} v_\ell^* = 0$ and
$\covstar \ones_\numalt = \vec{0}$.}

\vspace{10pt}
\begin{theorem}[Generic Separation]
\label{thm:separation}
For each alternative $j \in \{1, \ldots, \numalt\}$, there exists a closed, measure-zero set $\exceptsetutil_j \subset \utilcovspace$ such that for any $(\utilstar_1, \covstar_1), (\utilstar_2, \covstar_2) \in \utilcovspace \setminus \exceptsetutil_j$:
\vspace{3pt}
\[
\invarutil_j(\utilstar_1, \covstar_1) = \invarutil_j(\utilstar_2, \covstar_2) \implies (\utilstar_1, \covstar_1) \text{ and } (\utilstar_2, \covstar_2) \text{ are in the same } \permgroupsub\text{-orbit},\vspace{3pt}
\]
where 
$S_{\numalt}$ is the permutation group on $\{1,\ldots,\numalt\}$
and $\permgroupsub = \{\pi \in S_{\numalt} : \pi(j)=j\}$ is the subgroup of permutations that fixes alternative $j$.
\end{theorem}

\vspace{-2pt}
\begin{proof}[Proof sketch] The proof, which we give in detail in Appendix \ref{sec:utility-covariance}, proceeds in three steps:
\vspace{-3pt}

\textbf{Step 1: Genericity conditions.} We define two genericity conditions on
$(\utilstar,\covstar)$. First, the triples
$(v_k^*,\Sigma_{kk}^*,\Sigma_{jk}^*)$ for $k \neq j$ are pairwise distinct as
elements of $\R^3$, and the elements of $\mathcal{O}_j$ are pairwise distinct. Second, for each rival $k \neq j$, among all sub-multisets of $\mathcal{O}_j$ with $\numalt - 2$ elements, exactly one has elements summing to $-\Sigma_{kk}^*-\Sigma_{kj}^*$. The row-sum
constraint $\covstar\ones_\numalt=\vec{0}$ implies that the row-$k$
sub-multiset $\{\!\{\Sigma_{k\ell}^* : \ell \neq k,j\}\!\}$ has
exactly this sum.

\textbf{Step 2: Invariant reconstruction.} When both genericity conditions hold, we can recover $(\utilstar,\covstar)$ from $\invarutil_j$ up to the relabeling of alternatives other than $j$. 
First, order the triples in $\mathcal{T}_j$ lexicographically,\footnote{Lexicographic ordering sorts triples by their first coordinate, breaking ties by the second coordinate, and then by the third. Here it is used only to impose an arbitrary canonical labeling of the rivals.}
thereby assigning labels to the rivals. Next, for each rival $k$, use the row-sum restriction and the second genericity condition to identify the unique subcollection of $\mathcal{O}_j$ corresponding to the covariances
between $k$ and the other rivals. Finally, recover each rival--rival covariance $\Sigma_{k\ell}^*$ as the common element in the recovered subcollections for rivals $k$ and $\ell$. The remaining entries $v_j^*$ and $\Sigma_{jj}^*$ are pinned down by the normalization constraints $\sum_{\ell=1}^{\numalt} v_\ell^* = 0$ and $\covstar\ones_\numalt = \vec{0}$.

\textbf{Step 3: Measure-zero exceptional set.} The set $\exceptsetutil_j$ where either condition fails is a finite union of zero sets of nontrivial analytic functions on $\utilcovspace$, the constrained set defined in \eqref{eq:utilcovspace} (see also Definition~\ref{def:global-joint-space} in Appendix~\ref{sec:utility-covariance-global}). By the principle that proper analytic subvarieties have measure zero \citep{Mityagin2020}, $\exceptsetutil_j$ has measure zero.
\end{proof}

Theorem \ref{thm:separation}, together with the Stone--Weierstrass theorem, yields universal approximation.
To show this, define the enlarged space $\utilcovspaceex = \utilcovspace \times (0, \infty)$, with the added dimension corresponding to $\Cstar$, and enlarge the exceptional set to $\exceptsetutilex_j = (\exceptsetutil_j \cup \exceptsetutil_j^{0}) \times (0, \infty)$, where
\begingroup
\addtolength{\abovedisplayskip}{-4pt}
\addtolength{\belowdisplayskip}{-4pt}
\[
\exceptsetutil_j^{0} = \left\{(\utilstar, \covstar) \in \utilcovspace : \Sigma_{jj}^* + \Sigma_{kk}^* - 2\Sigma_{jk}^* = 0 \text{ for some } k \neq j\right\}
\]
\endgroup
is the closed, measure-zero set on which the error difference $\epsilon_j - \epsilon_k$ is degenerate for some rival $k$.
On $\exceptsetutil_j^{0}$, the choice between $j$ and $k$ is determined by the sign of $v_j - v_k$, so $\choiceprob_j$ can be discontinuous there.
Outside $\exceptsetutil_j^{0}$, $\choiceprob_j$ is continuous (Lemma \ref{lem:continuity} in Appendix \ref{sec:utility-covariance}).
Enlarging the exceptional set in this way is harmless: $\exceptsetutilex_j$ remains closed, measure zero, and invariant to permutations of the rivals.\footnote{For models that are not factorization-invariant, $\choiceprob_j(\Cstar, \utilstar, \covstar)$ is defined only where the inverse preprocessing map of Section~\ref{sec:preprocessing} yields a positive semidefinite $\cov$, which restricts $\Cstar$ for a given $\covstar$. The statements below are understood on that subset of $\utilcovspaceex$. For factorization-invariant models, $\Cstar$ is irrelevant, and $\covstar$ itself can serve as the scale matrix, so no restriction is needed.}
For models that are not factorization invariant, we must also modify the function $\mathcal{T}_j$ as follows:
\begingroup
\addtolength{\abovedisplayskip}{-4pt}
\addtolength{\belowdisplayskip}{-4pt}
\begin{equation}\label{eq:diagonal-multiset-expanded}
\overline{\mathcal{T}}_j(\utilstar, \covstar) = \{\!\{ (v_k^*, \Sigma_{kk}^*, \Sigma_{jk}^*, 1_{\{k=1\}}) : k \neq j \}\!\},
\end{equation}
\endgroup
where we add the element $1_{\{k=1\}}$.
This additional element allows the preprocessing transformation in Section \ref{sec:preprocessing} to be inverted.
We then define the augmented invariant function
\begingroup
\addtolength{\abovedisplayskip}{-4pt}
\addtolength{\belowdisplayskip}{-4pt}
\[\overline{\invarutil}_j(\utilstar, \covstar, \Cstar) = (\overline{\mathcal{T}}_j(\utilstar, \covstar), \mathcal{O}_j(\covstar), \Cstar),\]
\endgroup
which, based on Theorem \ref{thm:separation}, separates $(\utilstar, \covstar, \Cstar)$ up to permutations fixing alternatives $1$ and $j$ (which may coincide) outside of the measure-zero set $\exceptsetutilex_j$: the permutation delivered by Theorem \ref{thm:separation} must map the unique triple carrying the indicator $1_{\{k=1\}} = 1$ to itself, so it fixes alternative $1$ as well as $j$.
We then have the following universal approximation property.

\vspace{9pt}
\begin{theorem}[MLP Universal Approximation]
\label{thm:universal}
Under Assumption \ref{assumption:epsilon-star}, for each alternative $j$, the choice probability function $\choiceprob_j$ can be uniformly approximated on any compact subset of $\utilcovspaceex \setminus \exceptsetutilex_j$ by a multi-layer perceptron taking the components of $\overline{\invarutil}_j$ as input.
\end{theorem}
\vspace{-5pt}

\begin{proof}[Proof sketch]
We apply the proof strategy of \citet{blumsmith2025machine}.
The choice probability $\choiceprob_j$ is continuous on $\utilcovspaceex \setminus \exceptsetutilex_j$ (Lemma \ref{lem:continuity}) and $\permgroupsubsub$-invariant, where $\permgroupsubsub$ is the subgroup of permutations fixing alternatives 1 and $j$.
Define continuous functions encoding the components of $\overline{\invarutil}_j$ (except the indicators, $1_{\{k=1\}}$, which are used only for determining indices).
By Theorem~\ref{thm:separation}, these functions separate points in the
quotient space $(\utilcovspaceex \setminus \exceptsetutilex_j)/\permgroupsubsub$.\footnote{In this quotient space, two normalized choice environments are treated as the same point if they differ only by an element of $\permgroupsubsub$, that is, by a relabeling of the alternatives other than $1$ and $j$.}
Since ($\utilcovspaceex \setminus \exceptsetutilex_j$) is a subspace of a
finite-dimensional Euclidean space, it is Hausdorff. And, as $\permgroupsubsub$ is finite and acts by coordinate permutations,
standard results imply the quotient space is also Hausdorff
\citep{bredon1972introduction}.
Hence, by the Stone--Weierstrass theorem, the algebra generated by functions encoding $\overline{\invarutil}_j$ is dense in the continuous functions on compact subsets of this quotient space.\footnote{The usual Stone--Weierstrass theorem applies to compact Hausdorff
spaces. This condition is not automatic after passing to a quotient: when
multiple original points are combined into a single quotient-point, it can become
impossible to separate distinct quotient-points by disjoint open neighborhoods.} Pulling back to ($\utilcovspaceex \setminus \exceptsetutilex_j$) and applying the universal approximation theorem for MLPs \citep{cybenko1989approximation, hornik1989multilayer} yields the result.\footnote{Here ``pulling back'' means composing a quotient-space function with the natural projection from the original space to the quotient. If $Q$ maps each $(\Cstar, \utilstar,\covstar)$ to its equivalence class under rival relabeling, then an approximation $\widetilde F$ on the quotient gives the corresponding approximation $\widetilde F\circ Q$ on $\utilcovspaceex\setminus\exceptsetutilex_j$.}  See Appendix \ref{sec:utility-covariance} for details.
\end{proof}
\vspace{-3pt}

Theorem~\ref{thm:universal} shows that no information beyond the invariant
summaries $\overline{\invarutil}_j$ is needed to approximate $\choiceprob_j$ on compact
subsets of $\utilcovspaceex \setminus \exceptsetutilex_j$. The remaining question is
whether the particular architecture in Section~\ref{sec:architecture} is rich
enough to approximate the required functions of this summary. The following corollary answers this question:

\vspace{12pt}
\begin{corollary}[Architecture Universal Approximation]
\label{cor:universal-approximation-architecture}
Let $\exceptsetutilex = \cup_{j=1}^{\numalt} \exceptsetutilex_j$. Under Assumption \ref{assumption:epsilon-star}, for each alternative $j$, the choice probability function $\choiceprob_j$ can be uniformly approximated on any compact subset of $\utilcovspaceex \setminus \exceptsetutilex_j$ by the per-alternative encoder. Further, the vector of choice probabilities $(\choiceprob_1, \ldots, \choiceprob_{\numalt})$ can be uniformly approximated on any compact subset of $\utilcovspaceex \setminus \exceptsetutilex$ by the neural network architecture of Section \ref{sec:architecture}.
\end{corollary}
\begin{proof}
Theorem~\ref{thm:universal} shows that, for each alternative $j$,
$\choiceprob_j$ can be uniformly approximated on compact subsets of
$\utilcovspaceex \setminus \exceptsetutilex_j$ by a continuous function of $\overline{\invarutil}_j$. The per-alternative encoder is constructed to process precisely these values: the diagonal DeepSet processes the values in \eqref{eq:diagonal-multiset-expanded}, the off-diagonal DeepSet processes those in \eqref{eq:offdiagonal-multiset}, and the combining MLP processes the corresponding representations with $\Cstar$.
The first statement of Corollary \ref{cor:universal-approximation-architecture} therefore follows from Theorem~\ref{thm:universal} and the universal approximation property of DeepSets: any continuous permutation-invariant function of a multiset of vectors drawn from a compact set can be uniformly approximated by the form $\rho(\sum_k \phi(\cdot))$ \citep{zaheer2017deep, segol2020universal}.\footnote{\citet{zaheer2017deep} prove this for multisets of scalars; \citet{segol2020universal} cover multisets whose elements are vectors, as required here.}
For the second statement, first note that the preceding argument applies to $\exceptsetutilex = \cup_{j=1}^{\numalt} \exceptsetutilex_j$ for all $\numalt$ alternatives simultaneously.
Further, \citet{segol2020universal} show that linear equivariant layers composed with nonlinear activation functions can uniformly approximate continuous permutation-equivariant functions on compact subsets of the feature space.
Because the architecture outputs probabilities through a softmax, it suffices to approximate the continuous, permutation-equivariant logit map $\log\{(1 - \epsilon) \choiceprob_j + \epsilon/\numalt\}$ for small $\epsilon > 0$, which is within $\epsilon$ of $\choiceprob_j$ after the softmax.
Thus, the full architecture can uniformly approximate
the vector of choice probabilities on compact subsets of $\utilcovspaceex \setminus \exceptsetutilex$.
\end{proof}

\begin{remark} The exceptional set $\exceptsetutilex$ has measure zero, so it does not affect approximation in practice: for any probability distribution over $\utilcovspaceex$ that is absolutely continuous with respect to Lebesgue measure, inputs lie in $\exceptsetutilex$ with probability zero.
\end{remark}
\vspace{6pt}

\begin{remark} The values $\Cstar$ and $1_{\{k=1\}}$ are not required for models that are factorization invariant. In this special case, the proofs can be simplified because it is sufficient to show generic separation up to the symmetry group $\permgroupsub$.
\end{remark}

\vspace{6pt}
\begin{remark} Corollary \ref{cor:universal-approximation-architecture} guarantees that the per-alternative encoder alone is sufficient for universal approximation of each $\choiceprob_j$ separately. The full architecture computes all $\numalt$ logits and applies a softmax, which ensures that the estimated probabilities sum to one; the equivariant layers allow information to flow across alternatives before this step, which improves the expressivity of the neural network.
\end{remark}

\vspace{6pt}
\begin{remark}
For very large choice sets, in order to reduce computation time, one may opt for a simpler architecture that applies the per-alternative encoder to only the chosen alternative with no equivariant layers, resulting in a single output, $P_j$. This modification maintains the universal approximation property; however, it relaxes the sum-to-one constraint in return for computational efficiency.
\end{remark}

\vspace{6pt}
\begin{remark}
Corollary \ref{cor:universal-approximation-architecture} can be extended to hold simultaneously for a finite set of $\numalt$ values under a suitable extension of the space $\utilcovspaceex$ and its exceptional sets. Section \ref{sec:sim-multi-k} presents numerical results showing that the proposed architecture can provide accurate approximations for multiple values of $\numalt$ simultaneously. This makes it simple to train the emulator to handle contexts where the choice set size varies across observations.
\end{remark}

\vspace{-9pt}
\subsection{Consistency and Asymptotic Normality}
\label{sec:estimation-theory}

We now establish that emulator-based estimators inherit the asymptotic properties of exact maximum likelihood estimators under appropriate conditions.

Let $\{(y_i, \mat{X}_i)\}_{i=1}^n$ be IID observations from the true model with parameter $\paramtrue \in \paramspace$. Define the true and emulator log-likelihood contributions as
\vspace{-3.5pt}
\begin{align*}
\loglikobs(y, \mat{X}; {\param}) &= \log \choiceprob_{y}\left\{\util(\mat{X}, \param), \cov(\mat{X}, \param)\right\}, \\
\loglikobshatn(y, \mat{X}; {\param}) &= \log \choiceprobhat_{y}\left\{\Cstar(\mat{X}, \param), \utilstar(\mat{X}, \param), \covstar(\mat{X}, \param); \netparam_n\right\},
\end{align*}
\vspace{-1.5pt}
with corresponding score functions
\vspace{-3.5pt}
\[
\score(y, \mat{X}; \param) = \gradparam \loglikobs(y, \mat{X}; {\param}), \qquad
\scorehatn(y, \mat{X}; \param) = \gradparam \loglikobshatn(y, \mat{X}; {\param}).
\]
\vspace{-1.5pt}
The true and emulator (averaged) log likelihoods are $\loglik(\param) = \frac{1}{n} \sum_{i=1}^n \loglikobs(y_i, \mat{X}_i; \param)$ and $\loglikhat(\param) = \frac{1}{n} \sum_{i=1}^n \loglikobshatn(y_i, \mat{X}_i; \param)$. Define the MLE and emulator-based estimator as
\vspace{-3.5pt}
\[
\parammle = \argmax_{\param \in \paramspace} \loglik(\param), \qquad \paramest = \argmax_{\param \in \paramspace} \loglikhat(\param).
\]
\vspace{-1.5pt}
We require that the emulator log likelihood converges to the true log likelihood.

\vspace{5.5pt}
\begin{assumption}[Emulator Approximation Quality]
\label{ass:approx}
For some $\alpha \geq 0$, the emulator satisfies
\[
\sup_{\param \in \paramspace} \frac{\left| \loglikhat(\param) - \loglik(\param) \right|}{n^{-\alpha} + n^{-\alpha/2} \|\param - \paramtrue\|} = o_p(1).
\]
\end{assumption}

\vspace{-6.5pt}
Assumption \ref{ass:approx} is implied by the uniform condition $\sup_{\param \in \paramspace} |\loglikhat(\param) - \loglik(\param)| = o_p(n^{-\alpha})$, and the two conditions are equivalent when $\alpha = 0$ and $\paramspace$ is bounded, as it is under the compactness condition in Assumption~\ref{ass:regularity}(i) below. For $\alpha > 0$, Assumption \ref{ass:approx} requires the approximation error to be $o_p(n^{-\alpha})$ only within an $n^{-\alpha/2}$-neighborhood of $\paramtrue$; elsewhere, the error may grow linearly in $\|\param - \paramtrue\|$ at rate $o_p(n^{-\alpha/2})$. This weighted form is the one that emulators naturally satisfy: a choice-probability error of order $\varepsilon$ contributes an expected log-likelihood error of order $\varepsilon^2$ at $\paramtrue$ but of order $\varepsilon \|\param - \paramtrue\|$ elsewhere (Appendix \ref{sec:approx-discussion}).\hfil\break When $\alpha = 1$, the implicit weight $1 + \sqrt{n}\|\param - \paramtrue\|$ is the one used in stochastic-differentiability conditions for extremum estimators with nonsmooth objective functions \citep{pollard1985new, pakes1989simulation, newey1994large}, here applied to the emulator's approximation error rather than to a sampling remainder. 

The $n$ referenced in Assumption \ref{ass:approx} is the sample size of the observed data. Consequently, this assumption involves an asymptotic regime in which the emulator becomes increasingly accurate as $n \to \infty$, which is achievable by progressively (a) increasing the number of simulated training examples and (b) increasing the neural network complexity (i.e., the number of layers and hidden units). Appendix \ref{sec:approx-discussion} provides more primitive sufficient conditions for Assumption \ref{ass:approx} with references to closely related work on convergence rates of neural networks and their usage in econometric theory. Notably, for models with analytic choice probabilities, we expect this assumption to require only $O(n^{\alpha + \delta})$ training examples for any $\delta > 0$.

We also require sufficient regularity of the true model.

\vspace{6pt}
\begin{assumption}[Regularity]
\label{ass:regularity}
The model satisfies:
\begin{enumerate}
    \item[(i)] $\paramtrue$ is the unique maximizer of $\E\left\{\loglikobs(y, \mat{X}; \param)\right\}$ in the compact set $\paramspace$ with $\paramtrue \in \mathrm{int}(\paramspace)$ and $\E\left\{\sup_{\param \in \paramspace} |\loglikobs(y, \mat{X}; \param)|\right\} < \infty$;
    \item[(ii)] $\loglikobs(y, \mat{X}; {\param})$ is twice continuously differentiable in $\param$ for all $(y, \mat{X})$;
    \item[(iii)] $\E\left\{\sup_{\param \in \paramspace}\|\score(y, \mat{X}; \param)\|^2\right\} < \infty$;
    \item[(iv)] $\E\left\{\sup_{\param \in \paramspace}\|\gradparam^2 \loglikobs(y, \mat{X}; {\param})\|\right\} < \infty$;
    \item[(v)] the Fisher information $\fisher(\paramtrue) = \E\left\{\score(y, \mat{X}; \paramtrue) \score(y, \mat{X}; \paramtrue)^\top\right\}$ is positive definite.
\end{enumerate}
\end{assumption}

The unique-maximizer requirement in condition (i) is an identification condition. It presumes a parameterization in which the usual location and scale indeterminacies of random utility models have been eliminated. In particular, the location normalization must rule out observationally equivalent parameterizations that differ only by a common additive shift in utilities, while the scale normalization must rule out parameterizations that differ only by a common positive rescaling of utilities. The latter may be imposed, for example, by normalizing one utility coefficient, such as the price coefficient, or by imposing an appropriate restriction on the error covariance parameters. Conditional on these normalizations, the data must contain sufficient variation to identify the remaining parameters; for example, variation in alternative-specific prices can be essential \citep{keane1992}.

\vspace{9pt}
\begin{theorem}[Consistency]
\label{thm:consistency}
Under Assumptions~\ref{ass:approx} and \ref{ass:regularity}, $\paramest \convp \paramtrue$ as $n \to \infty$.
\end{theorem}

\vspace{-6pt}
\begin{proof}
By Assumption~\ref{ass:approx} with $\alpha \geq 0$ and the compactness of $\paramspace$ (Assumption~\ref{ass:regularity}(i)), $\sup_{\param} |\loglikhat(\param) - \loglik(\param)| \convp 0$. Further, the regularity conditions of Assumption \ref{ass:regularity} are sufficient to apply the uniform law of large numbers (ULLN) so that $\sup_{\param} |\loglik(\param) - \E\left\{\loglikobs(y, \mat{X}; {\param})\right\}| \convp 0$. The triangle inequality gives uniform convergence of $\loglikhat$ to $\E\left\{\loglikobs(y, \mat{X}; {\param})\right\}$, which is uniquely maximized at $\paramtrue$ by Assumption~\ref{ass:regularity}(i). By the consistency theorem for extremum estimators \citep[Theorem 2.1]{newey1994large}, $\paramest \convp \paramtrue$.
\end{proof}

\vspace{-6pt}
Under the stronger condition $\alpha \geq 1$, the emulator-based estimator is asymptotically normal.

\vspace{6pt}
\begin{theorem}[Asymptotic Normality]
\label{thm:normality}
Under Assumptions~\ref{ass:approx} and \ref{ass:regularity} with $\alpha \geq 1$,
\[
\sqrt{n}(\paramest - \paramtrue) \convd \mathcal{N}\left\{\vec{0}, \fisher(\paramtrue)^{-1}\right\}.
\]
\end{theorem}

\vspace{-9.5pt}
The proof, given in Appendix~\ref{sec:normality-proof}, shows that $\paramest$ is asymptotically equivalent to $\parammle$.

\vspace{6pt}
\begin{remark}[Approximate Maximizers]
\label{rem:approx-max}
If a maximum is not uniquely attained, one can work with approximate maximizers satisfying $\loglikhat(\paramest) \geq \sup_{\param} \loglikhat(\param) - o_p(n^{-\alpha})$. The consistency and asymptotic normality results continue to hold; see \citet[Theorem 2.1]{newey1994large} and Appendix~\ref{sec:normality-proof}.
\end{remark}

To estimate $\fisher(\paramtrue)$, define the true and emulator outer-product-of-scores estimators:
\begin{align*}
\fisher_n(\param) &= \frac{1}{n} \sum_{i=1}^n \score(y_i, \mat{X}_i; \param) \score(y_i, \mat{X}_i; \param)^{\top}, \\
\hat{\fisher}_n(\param) &= \frac{1}{n} \sum_{i=1}^n \scorehatn(y_i, \mat{X}_i; \param) \scorehatn(y_i, \mat{X}_i; \param)^{\top}.
\end{align*}

\begin{assumption}[Gradient Approximation Quality]
\label{ass:grad-approx}
The emulator satisfies
\[
\sup_{\param \in \paramspace} \left\| \hat{\fisher}_n(\param) - \fisher_n(\param) \right\| \convp 0 \quad\text{as } n \to \infty.
\]
\end{assumption}

\vspace{-9.5pt}
Assumption~\ref{ass:grad-approx} is justified by Sobolev training, which penalizes discrepancies between emulator and target gradients. The penalization is required because training on log-likelihood values alone does not guarantee convergence of gradients. As with Assumption~\ref{ass:approx}, see Appendix \ref{sec:approx-discussion} for more primitive sufficient conditions and additional discussion.

\vspace{9pt}
\begin{proposition}[Consistent Estimation of Fisher Information]
\label{prop:fisher-estimation}
Under Assumptions~\ref{ass:approx}, \ref{ass:regularity}, and \ref{ass:grad-approx}, $\hat{\fisher}_n(\paramest) \convp \fisher(\paramtrue)$.
\end{proposition}

\begin{proof}
By the triangle inequality,
$\| \hat{\fisher}_n(\paramest) - \fisher(\paramtrue) \| \leq \| \hat{\fisher}_n(\paramest) - \fisher_n(\paramest) \| + \| \fisher_n(\paramest) - \fisher(\paramest) \| + \| \fisher(\paramest) - \fisher(\paramtrue) \|.$
The first term vanishes by Assumption~\ref{ass:grad-approx}. The second term converges in probability to zero as the regularity conditions of Assumption \ref{ass:regularity} are sufficient to apply the ULLN to $\fisher_n$.  The third term vanishes by continuity of $\fisher$ and consistency of $\paramest$.
\end{proof}

The emulator score $\scorehatn(y_i, \mat{X}_i; \param)$ is computable via automatic differentiation through the smooth neural network. Standard errors for $\paramest$ are obtained as the square roots of the diagonal elements of $\hat{\fisher}_n(\paramest)^{-1}/n$.

An alternative estimator based on the Hessian,
$\tilde{\fisher}_n = -\frac{1}{n} \sum_{i=1}^n \gradparam^2 \loglikobshatn(y_i, \mat{X}_i; \paramest)$,
requires second-derivative accuracy not directly targeted by first-order Sobolev training. The score-based estimator $\hat{\fisher}_n(\paramest)$ is therefore preferred.

\subsection{Inference Under Misspecification}
\label{sec:misspecification}

When the emulator does not perfectly approximate the true likelihood, we can still obtain valid inference by treating the emulator as a working model, following the quasi-maximum likelihood framework \citep{white1982maximum}. In this section, we treat the neural network as fixed with $n$, dropping the $n$ subscript from the observation log likelihood, $\loglikobshat(y, \mat{X}; \param)$, and score function, $\scorehat(y, \mat{X}; \param)$.
Define the pseudo-true parameter as the population maximizer of the emulator log likelihood:
\[
\parampseudo = \argmax_{\param \in \paramspace} \E\left\{ \loglikobshat(y, \mat{X}; {\param}) \right\}.
\]
When the emulator closely approximates the true likelihood, $\parampseudo \approx \paramtrue$. The asymptotic distribution depends on two matrices:
\vspace{-2pt}
\[
\Amat(\param) = -\E\left\{ \gradparam^2 \loglikobshat(y, \mat{X}; {\param}) \right\}, \qquad
\Bmat(\param) = \E\left\{ \scorehat(y, \mat{X}; \param) \scorehat(y, \mat{X}; \param)^\top \right\}.
\]
\vspace{-2pt}
Under correct specification (i.e., a perfect emulator), the information matrix equality gives $\Amat(\paramtrue) = \Bmat(\paramtrue) = \fisher(\paramtrue)$. Under misspecification, these matrices differ, necessitating the sandwich covariance form. 
Under regularity conditions stated in Appendix~\ref{sec:regularity-misspec-appendix}, we obtain the following results.

\vspace{6pt}
\begin{theorem}[Consistency Under Misspecification]
\label{thm:consistency-misspec}
Under Assumptions~\ref{ass:pseudo-identification} and \ref{ass:regularity-misspec}, $\paramest \convp \parampseudo$.
\end{theorem}

\vspace{6pt}
\begin{theorem}[Asymptotic Normality Under Misspecification]
\label{thm:sandwich}
Under Assumptions~\ref{ass:pseudo-identification} and \ref{ass:regularity-misspec},
\[
\sqrt{n}(\paramest - \parampseudo) \convd \mathcal{N}\left\{\vec{0}, \Amat(\parampseudo)^{-1} \Bmat(\parampseudo) \Amat(\parampseudo)^{-\top}\right\}.
\]
\end{theorem}
\vspace{-3pt}

The proofs of Theorems~\ref{thm:consistency-misspec} and~\ref{thm:sandwich} both follow standard M-estimator arguments; see Appendix~\ref{sec:proof-sandwich}. The sandwich covariance $\sandwichmat = \Amat(\parampseudo)^{-1} \Bmat(\parampseudo) \Amat(\parampseudo)^{-\top}$ can be consistently estimated by $\hat{\sandwichmat}_n = \hat{\Amat}_n^{-1} \hat{\Bmat}_n \hat{\Amat}_n^{-\top}$, where
\vspace{-5pt}
\[
\hat{\Amat}_n = -\frac{1}{n} \sum_{i=1}^n \gradparam^2 \loglikobshat(y_i, \mat{X}_i; \paramest), \qquad \hat{\Bmat}_n = \frac{1}{n} \sum_{i=1}^n \scorehat(y_i, \mat{X}_i; \paramest) \scorehat(y_i, \mat{X}_i; \paramest)^{\top}.
\]
Both quantities are computable via automatic differentiation.
Note that $\hat{\Bmat}_n$ corresponds to $\hat{\fisher}_n(\paramest)$ from the previous section under an asymptotic regime with a fixed emulator.

Theorem~\ref{thm:sandwich} provides valid inference for $\parampseudo$, not necessarily $\paramtrue$. If we employ an asymptotic regime where the emulator varies with $n$, then the pseudo-true parameter also varies with $n$, so we denote it as $\parampseudon$. When $\alpha \in [0, 1)$ in Assumption~\ref{ass:approx}, we have $\parampseudon \convp \paramtrue$ as $n \to \infty$, provided the population-level approximation error $\sup_{\param \in \paramspace} |\E\{\loglikobshatn(y, \mat{X}; \param)\} - \E\{\loglikobs(y, \mat{X}; \param)\}|$ vanishes in probability, which follows from the primitive conditions of Appendix~\ref{sec:approx-discussion}; the argument of Theorem~\ref{thm:consistency} then applies to the population maximizer. However, this convergence may be too slow relative to $\sqrt{n}$ sampling variability. The sandwich standard errors remain valid for $\parampseudon$, but confidence intervals should not be interpreted as targeting $\paramtrue$. Even if $\Amat(\parampseudon)$ and $\Bmat(\parampseudon)$ both approximate $\fisher(\paramtrue)$, the centering at $\parampseudon \neq \paramtrue$ introduces asymptotic bias. Valid inference for $\paramtrue$ requires $\alpha \geq 1$ (Theorem~\ref{thm:normality}). Alternatively, with a fixed emulator, one can use the emulator estimate as a starting value for Newton--Raphson refinement based on a more accurately simulated likelihood \citep[Theorem~5]{kristensen2017higher}.
Under appropriate convergence and simulation-accuracy conditions, this may substantially reduce emulator-induced bias while confining expensive simulated likelihood evaluations to the final refinement stage.

Similarly, Proposition \ref{prop:fisher-estimation} provides conditions under which $\hat{\Bmat}_n \convp \fisher(\paramtrue)$. However, our first-order Sobolev training procedure provides no guarantee that the emulator Hessian will approximate the Hessian of the true model; thus, $\hat{\Amat}_n$ need not converge to any quantity from the true model. However, Theorem~\ref{thm:sandwich} requires only that $\Amat(\parampseudo)$ be nonsingular, not that it match the true model. When the emulator additionally provides accurate second derivatives, the sandwich covariance matrix reduces to $\fisher(\paramtrue)^{-1}$.

\vspace{-2pt}
\section{Simulation Study}
\label{sec:simulation}
\vspace{-2pt}

We evaluate the performance of emulator-based MNP estimators across a range of sample sizes, numbers of alternatives, and model specifications.
We compare their statistical performance and computational requirements to estimators based on the GHK simulator.

\vspace{-8pt}
\subsection{Simulation Design}
\label{sec:simulation-design}
\vspace{-2pt}

\textbf{Model specification.} We consider MNP models with $\numalt \in \{3, 5, 10\}$ alternatives and two covariates. Alternative $1$ serves as the reference alternative, with its utility normalized to zero. For each non-reference alternative $j \in \{2, \ldots, \numalt\}$, the deterministic utility is $v_j = \vec{x}_j^\top \vec{\beta}$, where $\vec{x}_j \in \R^2$ is a vector of alternative-specific attributes and $\vec{\beta} = (0, 1)^\top$ is the true coefficient vector. We consider two specifications for $\cov$, denoting them as Dense and Factor.

The Dense specification assumes a dense structure for $\cov_{2:\numalt,2:\numalt}$, the lower $(\numalt-1) \times (\numalt-1)$ submatrix of $\cov$. For identification, we set $\Sigma_{22} = 1$ and $\cov_{\bigcdot 1} = \cov_{1 \bigcdot}^{\top} = \vec{0}_{\numalt}$; the latter equation sets the variance and covariances of the reference alternative to zero.\footnote{This normalization differs from \eqref{eq:Sigma-structure}, which sets $\Sigma_{11} = 1$. The difference is immaterial here because MNP is factorization invariant (Definition~\ref{def:factorization-invariant}), so the MNP emulator requires only $(\utilstar, \covstar)$ as inputs and never needs to recover the level structure in \eqref{eq:Sigma-structure}.}

The Factor specification assumes a one-factor structure for $\cov_{2:\numalt,2:\numalt}$ of the form $\cov_{2:\numalt,2:\numalt} = \mat{\Psi} + \vec{\gamma} \vec{\gamma}^{\top}$, where $\vec{\gamma} = (\gamma_2, \ldots, \gamma_{\numalt})^{\top} \in \R^{\numalt-1}$ and $\mat{\Psi} = \mathrm{diag}(\Psi_{22}, \ldots, \Psi_{\numalt\numalt}) \in \R^{(\numalt - 1) \times (\numalt - 1)}$ is a diagonal matrix of uniquenesses. We index the elements of $\vec{\gamma}$ and $\mat{\Psi}$ by alternative, so that $\gamma_j$ and $\Psi_{jj}$ correspond to alternative $j \in \{2, \ldots, \numalt\}$. For identification, we impose an anchored normalization that fixes the first non-reference uniqueness and factor loading:
\[
\Psi_{22} = 0, \qquad \gamma_2 = 1.
\]
As with the Dense specification, we set $\cov_{\bigcdot 1} = \cov_{1 \bigcdot}^{\top} = \vec{0}_{\numalt}$. These restrictions imply $\Sigma_{22} = 1$, fixing the scale of the model, and $\Sigma_{2j} = \gamma_j$ for $j \geq 3$, so the factor loadings directly encode covariances with alternative 2. Given these constraints, we parameterize the model in terms of $\numalt - 2$ free uniquenesses $(\Psi_{33}, \ldots, \Psi_{\numalt\numalt})$ and $\numalt - 2$ free loadings $(\gamma_3, \ldots, \gamma_{\numalt})$.\footnote{The $\Psi_{22}=0,\ \gamma_2=1$ restriction is stronger than is required for theoretical identification. The weaker restriction $\gamma_2^2+\Psi_{22}=1$ is sufficient to fix the overall utility scale, but in practice we found the likelihood to be very flat along directions that reallocate variance between $\gamma_2^2$ and $\Psi_{22}$ while making offsetting changes in the other factor loadings. The anchored restriction removes this source of practical weak identification.}

Table \ref{tab:parameter-counts} shows the number of nonredundant parameters by type for the Dense and Factor specifications. With the anchored normalization, the Factor specification is exactly identified for $\numalt = 3$ and overidentified for $\numalt \geq 4$. For $\numalt=10$, the Factor specification results in fewer than half as many parameters as the Dense specification (18 vs.\ 46). 

\begin{table}[htb]
    \centering
\captionsetup{skip=4pt}
\caption{Parameter Counts by Covariance Specification}
\label{tab:parameter-counts}
    \small
    \begin{tabular}{crccS[table-format=2.0]S[table-format=2.0]}
            \toprule
            \multirow{2}{*}{\shortstack{Covariance ($\cov$)\\Specification}} & $\numalt$ & Coefficients & \multirow{2}{*}{\shortstack{Variances\\(Uniquenesses)}} & {\multirow{2}{*}{\shortstack{Covariances\\(Loadings)}}} & {Total} \\
            & & & & & \\
            \midrule
            \multirow{3}{*}{Dense}  & 3  & 2 & 1 & 1  & 4  \\
                                    & 5  & 2 & 3 & 6  & 11 \\
                                    & 10 & 2 & 8 & 36 & 46 \\
            \addlinespace
            \multirow{3}{*}{Factor} & 3  & 2 & 1 & 1  & 4  \\
                                    & 5  & 2 & 3 & 3  & 8  \\
                                    & 10 & 2 & 8 & 8  & 18 \\
            \bottomrule
    \end{tabular}

\vspace{4pt}
\end{table}

\textbf{Data generation.} We generate covariates $\vec{x}_{ij} \overset{\text{iid}}{\sim} \mathcal{N}(\vec{0}, \mat{I}_2)$ for $i = 1, \ldots, n$ and alternatives $j = 2, \ldots, \numalt$, with $\vec{x}_{i1} = \vec{0}$ for the reference alternative. We draw the covariates once for each combination of sample size and DGP and hold them fixed across replications and estimation methods. Each replication draws new true covariance parameters and new choices.

For the Dense specification, we draw the covariance matrix according to a $\text{Wishart}(\mat{I}_{\numalt-1}, \numalt + 10)$ distribution and then normalize it so that $\Sigma_{22} = 1$, resulting in a PDF given by
\[
\frac{\Gamma(s) \cdot |\cov_{2:\numalt,2:\numalt}|^{5}}{2^{s} \Gamma_{\numalt-1}\left(\frac{\numalt+10}{2}\right) \cdot \left\{\frac{1}{2}\mathrm{tr}(\cov_{2:\numalt,2:\numalt})\right\}^{s}},
\]
where $s = (\numalt + 10)(\numalt-1)/2$ and $\Gamma_{\numalt-1}(\cdot)$ denotes the $(\numalt-1)$-dimensional multivariate gamma function (see Appendix \ref{sec:scaled-wishart}).

For the Factor specification, we sample the free uniquenesses on the log scale, drawing $\log(\Psi_{jj}) \overset{\text{iid}}{\sim} \mathcal{N}(0, 0.1^2)$ for $j = 3, \ldots, \numalt$, and exponentiate to obtain $\Psi_{jj} > 0$. For the factor loadings, we draw $\gamma_j \overset{\text{iid}}{\sim} \mathcal{N}(0, 0.2^2)$ for $j = 3, \ldots, \numalt$. The fixed values $\Psi_{22} = 0$ and $\gamma_2 = 1$ complete the parameterization.

Given the covariates and true parameters, we compute deterministic utilities and simulate choices $y_i \in \{1, \ldots, \numalt\}$ according to one of two data-generating processes (DGPs): MNP or Emulator. The MNP DGP generates choices from the exact MNP model, drawing errors $\vec{\epsilon}_i \sim \mathcal{N}(\vec{0}, \cov)$ and selecting the utility-maximizing alternative. The Emulator DGP generates choices by sampling directly from the categorical distribution defined by the emulator's predicted choice probabilities. When coupled with the Emulator DGP, the emulator-based estimation method is correctly specified, enabling a comparison of emulator-based estimators under both correct specification (Emulator DGP) and approximate inference (MNP DGP).

\textbf{Estimation.} To ensure positive definiteness of $\cov$, we estimate the Dense specification on the scale of the Cholesky factor, with the diagonal elements optimized on the log scale to enforce positivity. Similarly, we estimate the Factor specification with the free uniquenesses on the log scale, optimizing $\log(\Psi_{jj})$ for $j = 3, \ldots, \numalt$.

We form estimators using two methods, Emulator and GHK, corresponding to whether we approximate the log-likelihood function using a pretrained emulator or the GHK simulator. We fit both methods to data generated by the MNP DGP; we also fit the Emulator method to data generated by the Emulator DGP to analyze its performance under correct specification. We apply the GHK method with $R \in \{10, 50, 250\}$ draws, representing low, moderate, and high computational effort.

Both methods, Emulator and GHK, perform regularized maximum likelihood estimation. The regularization consists of two components: a barrier penalty and a pseudo-prior. The barrier penalty is applied to $\log(\Psi_{jj})$ for the Factor specification and to the log-diagonal Cholesky parameters for the Dense specification. For a latent parameter $\vartheta$, the barrier penalty is $\lambda_{\text{reg}} \cdot \log(\vartheta + 8)/n$. The pseudo-prior terms take the form $\lambda_{\text{reg}} \cdot \log\left\{f_{\text{DGP}}(\param)\right\} /n$, where $f_{\text{DGP}}(\param)$ is the density function of $\param$ under the data-generating process. We add these penalties together to form a combined penalty of the form $\lambda_{\text{reg}} \cdot r(\param)/n$, where $r(\param)$ sums the barrier and log-prior penalties and $\lambda_{\text{reg}}$ is a regularization hyperparameter.

We then estimate $\param$ as the maximizer of the sum of the (approximate) log likelihood and $\lambda_{\text{reg}} \cdot r(\param)/n$. We take $\lambda_{\text{reg}} = 0.1$ so that the penalty is small and asymptotically negligible. In practice, we find that these penalties lead to more reliable convergence of the model-fitting algorithm and help address weak identification of the covariance parameters, especially when $n$ is small relative to $\numalt$. We account for these penalties in the standard error computation by estimating $\Cov(\paramest)$ as $\big\{n \cdot \hat{\fisher}_n(\paramest) - \lambda_{\text{reg}} \cdot \nabla_{\param}^2\, r(\paramest)\big\}^{-1}$. Again, the impact of the penalization is asymptotically negligible because the term $n \cdot \hat{\fisher}_n(\paramest)$ dominates asymptotically.

For all methods, we maximize the (approximate) regularized log likelihood using L-BFGS optimization \citep{liu1989}. In practice, we found it necessary to limit the minimum step size (we used $10^{-6}$) and allow small decreases in the objective function to enable the optimizer to escape local optima.

\textbf{Experimental design.} We consider sample sizes $n \in \{1{,}000, 10{,}000, 100{,}000\}$ and conduct 100 replications for each configuration. We compute performance metrics aggregated across all parameter types. We compute Monte Carlo standard errors for the coverage and average timings using scaled empirical standard deviations. For the remaining metrics, we apply bootstrap resampling (across repetitions) with 1{,}000 bootstrap samples to estimate the Monte Carlo standard errors.

\textbf{Computational environment.} All simulations are implemented in PyTorch \citep{paszke2019pytorch} and executed on a personal laptop with 48 GB RAM and an Apple M4 Pro processor. We train the emulators using the Apple Metal Performance Shaders (MPS). To obtain comparable estimation timings, however, we perform the rest of the simulation on CPU.

\vspace{-6pt}
\subsection{Simulation Results}
\label{sec:simulation-results}

We train the emulators for $400{,}000$ optimizer steps using $400{,}000$ pregenerated training examples, each consisting of $10^6$ choices. Details on the optimizer, learning rate schedule, and network architecture are provided in Appendix~\ref{sec:emulator-training}. Training completed in 2.8, 4.8, and 14.0 hours for $\numalt = 3$, $5$, and $10$, respectively, representing a modest level of computational effort.

Figure \ref{fig:loss} displays the training and evaluation loss for the emulators over time. At each value of $\numalt \in \{3, 5, 10\}$, the training loss decreases consistently over time even at the end of training, indicating that further training would likely lead to additional improvements. The evaluation loss closely tracks the training loss, implying that the emulator generally exhibits little overfitting to the simulated training probabilities.
We compute the evaluation loss every 200 optimizer steps; because it is computed less frequently, the evaluation loss curve displays much less visible variation than the training loss curve.
For each emulator, we retain the weights that produce the lowest evaluation loss.

\begin{figure}[!t]
\centering
\caption{Training and Evaluation Loss by Number of Alternatives}
\label{fig:loss}
\begin{subfigure}[b]{0.328\textwidth}
\centering
\includegraphics[width=\textwidth]{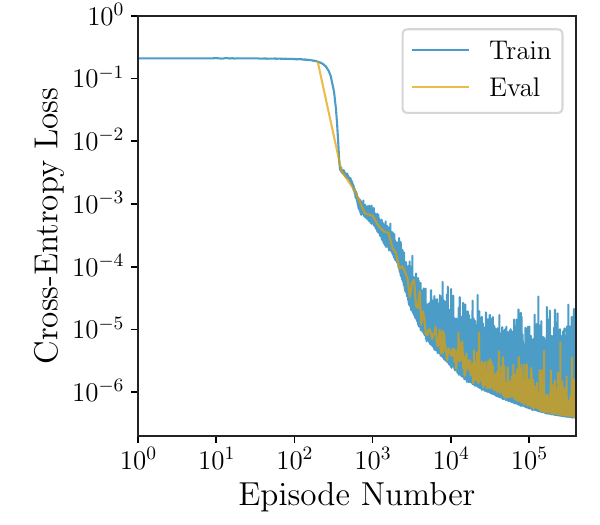}
\caption{$\numalt = 3$}
\label{fig:loss:3}
\end{subfigure}
\hfill
\begin{subfigure}[b]{0.328\textwidth}
\centering
\includegraphics[width=\textwidth]{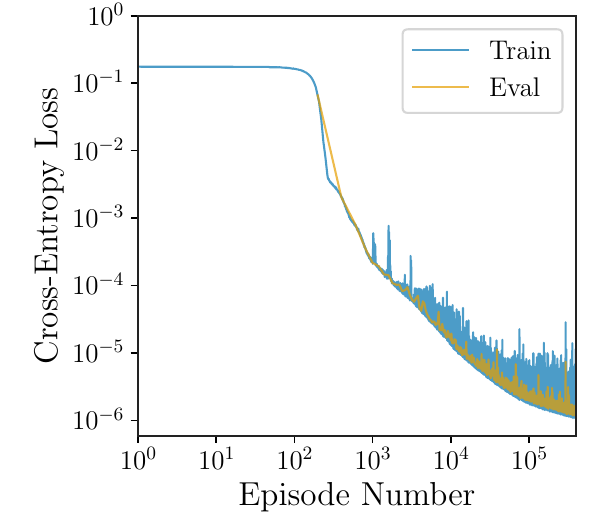}
\caption{$\numalt = 5$}
\label{fig:loss:5}
\end{subfigure}
\hfill
\begin{subfigure}[b]{0.328\textwidth}
\centering
\includegraphics[width=\textwidth]{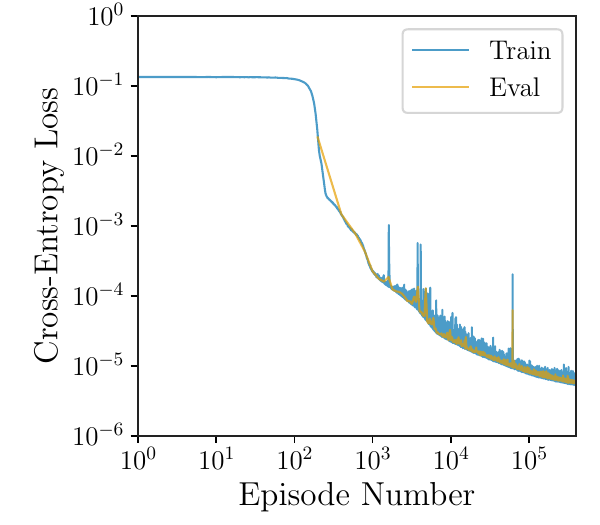}
\caption{$\numalt = 10$}
\label{fig:loss:10}
\end{subfigure}

\floatnote{Blue lines show training loss and orange lines show evaluation loss over $400{,}000$ optimizer steps for $\numalt \in \{3,5,10\}$. Evaluation loss is computed every 200 optimizer steps.}
\end{figure}
\FloatBarrier

\begin{table}[!b]
\centering
\caption{Dense MNP Simulation Results: $\numalt=3$}
\label{tab:emulator_vs_ghk:3:all}
\footnotesize
\begin{tabular}{@{}rllrrrrr@{}}
\toprule
\multicolumn{1}{c}{$n$} & \multicolumn{1}{l}{DGP} & \multicolumn{1}{l}{Method} & \multicolumn{1}{c}{RMSE} & \multicolumn{1}{c}{RMS Bias} & \multicolumn{1}{c}{SE Ratio} & \multicolumn{1}{c}{Coverage} & \multicolumn{1}{c}{Time (s)} \\
\midrule
\multirow{5}{*}{1,000} & Emulator & Emulator & 0.157 (0.016) & 0.012 (0.009) & 1.01 (0.05) & 0.94 (0.01) & 0.11 (0.01) \\
 & MNP & Emulator & 0.208 (0.033) & 0.050 (0.019) & 0.95 (0.08) & 0.95 (0.01) & 0.11 (0.00) \\
 & MNP & GHK(10) & 0.202 (0.031) & 0.047 (0.019) & 0.96 (0.08) & 0.94 (0.02) & 0.08 (0.00) \\
 & MNP & GHK(50) & 0.203 (0.032) & 0.048 (0.018) & 0.96 (0.09) & 0.95 (0.01) & 0.14 (0.00) \\
 & MNP & GHK(250) & 0.207 (0.033) & 0.049 (0.019) & 0.95 (0.08) & 0.95 (0.01) & 0.40 (0.01) \\
\midrule
\multirow{5}{*}{10,000} & Emulator & Emulator & 0.052 (0.006) & 0.007 (0.004) & 0.93 (0.07) & 0.93 (0.01) & 0.36 (0.00) \\
 & MNP & Emulator & 0.054 (0.009) & 0.008 (0.005) & 0.98 (0.11) & 0.94 (0.01) & 0.36 (0.00) \\
 & MNP & GHK(10) & 0.052 (0.008) & 0.010 (0.003) & 1.00 (0.09) & 0.94 (0.01) & 0.29 (0.00) \\
 & MNP & GHK(50) & 0.052 (0.009) & 0.007 (0.004) & 1.00 (0.10) & 0.95 (0.01) & 0.42 (0.00) \\
 & MNP & GHK(250) & 0.053 (0.009) & 0.008 (0.005) & 0.99 (0.11) & 0.94 (0.01) & 0.84 (0.01) \\
\midrule
\multirow{5}{*}{100,000} & Emulator & Emulator & 0.016 (0.002) & 0.001 (0.001) & 0.97 (0.08) & 0.94 (0.01) & 1.69 (0.01) \\
 & MNP & Emulator & 0.016 (0.002) & 0.001 (0.001) & 1.00 (0.07) & 0.97 (0.01) & 1.73 (0.01) \\
 & MNP & GHK(10) & 0.020 (0.002) & 0.007 (0.001) & 0.83 (0.06) & 0.83 (0.02) & 0.87 (0.01) \\
 & MNP & GHK(50) & 0.016 (0.002) & 0.002 (0.001) & 1.03 (0.08) & 0.97 (0.01) & 1.77 (0.02) \\
 & MNP & GHK(250) & 0.016 (0.002) & 0.001 (0.001) & 1.06 (0.07) & 0.98 (0.01) & 6.76 (0.09) \\
\bottomrule
\end{tabular}

\floatnote{Performance metrics are averaged across coefficients, variances, and covariances. Monte Carlo standard errors are shown in parentheses.}
\end{table}

\noindent Table \ref{tab:emulator_vs_ghk:3:all} compares the performance of the emulator and GHK estimation methods for $\numalt=3$. The performance metrics are root mean squared error (RMSE), root mean squared bias (RMS Bias), the ratio of estimated to empirical standard errors (SE Ratio), 95\% confidence interval coverage rates (Coverage), and elapsed model-fitting time in seconds (Time (s)). The emulator performs well across all sample sizes, nearly matching GHK(50) and GHK(250) in statistical performance but at a lower computational cost in most comparisons. GHK(10) significantly undercovers at $n=100{,}000$ and exhibits substantially higher RMS Bias than the other methods. This poor performance is likely an artifact of the GHK simulator providing biased estimates on the log scale due to Jensen's inequality; for sufficiently large $n$, this bias is eventually nonnegligible relative to $1/\sqrt{n}$ sampling error.

Table \ref{tab:emulator_vs_ghk:5:all} displays similar results for $\numalt=5$. The emulator is competitive with GHK(50) and GHK(250) at all sample sizes. Compared to GHK(50), the compute time for the Emulator method is slightly lower for $n=1{,}000$, approximately equal for $n = 10{,}000$, and roughly twice as high for $n = 100{,}000$.

\begin{table}[ptbh]
\centering
\caption{Dense MNP Simulation Results: $\numalt=5$}
\label{tab:emulator_vs_ghk:5:all}
\footnotesize
\begin{tabular}{@{}rllrrrrr@{}}
\toprule
\multicolumn{1}{c}{$n$} & \multicolumn{1}{l}{DGP} & \multicolumn{1}{l}{Method} & \multicolumn{1}{c}{RMSE} & \multicolumn{1}{c}{RMS Bias} & \multicolumn{1}{c}{SE Ratio} & \multicolumn{1}{c}{Coverage} & \multicolumn{1}{c}{Time (s)} \\
\midrule
\multirow{5}{*}{1,000} & Emulator & Emulator & 0.243 (0.017) & 0.032 (0.012) & 1.11 (0.06) & 0.94 (0.01) & 0.41 (0.01) \\
 & MNP & Emulator & 0.282 (0.020) & 0.045 (0.017) & 1.04 (0.04) & 0.93 (0.01) & 0.42 (0.01) \\
 & MNP & GHK(10) & 0.278 (0.015) & 0.074 (0.017) & 0.95 (0.03) & 0.89 (0.01) & 0.28 (0.00) \\
 & MNP & GHK(50) & 0.280 (0.017) & 0.051 (0.017) & 1.02 (0.04) & 0.92 (0.01) & 0.58 (0.01) \\
 & MNP & GHK(250) & 0.284 (0.020) & 0.046 (0.017) & 1.03 (0.04) & 0.93 (0.01) & 1.94 (0.02) \\
\midrule
\multirow{5}{*}{10,000} & Emulator & Emulator & 0.086 (0.006) & 0.010 (0.004) & 0.99 (0.05) & 0.94 (0.01) & 2.24 (0.02) \\
 & MNP & Emulator & 0.084 (0.005) & 0.009 (0.004) & 1.05 (0.04) & 0.96 (0.01) & 2.27 (0.02) \\
 & MNP & GHK(10) & 0.095 (0.006) & 0.032 (0.006) & 0.88 (0.03) & 0.89 (0.01) & 0.85 (0.01) \\
 & MNP & GHK(50) & 0.086 (0.006) & 0.013 (0.005) & 1.02 (0.04) & 0.95 (0.01) & 2.24 (0.03) \\
 & MNP & GHK(250) & 0.084 (0.005) & 0.010 (0.004) & 1.06 (0.04) & 0.96 (0.01) & 4.86 (0.05) \\
\midrule
\multirow{5}{*}{100,000} & Emulator & Emulator & 0.028 (0.004) & 0.003 (0.002) & 0.97 (0.09) & 0.95 (0.01) & 17.42 (0.14) \\
 & MNP & Emulator & 0.029 (0.002) & 0.001 (0.001) & 0.98 (0.05) & 0.95 (0.01) & 17.68 (0.16) \\
 & MNP & GHK(10) & 0.056 (0.004) & 0.024 (0.003) & 0.49 (0.03) & 0.60 (0.02) & 3.73 (0.03) \\
 & MNP & GHK(50) & 0.030 (0.002) & 0.004 (0.001) & 0.93 (0.05) & 0.92 (0.01) & 8.29 (0.07) \\
 & MNP & GHK(250) & 0.028 (0.002) & 0.002 (0.001) & 1.00 (0.06) & 0.95 (0.01) & 34.70 (0.43) \\
\bottomrule
\end{tabular}

\floatnote{Performance metrics are averaged across coefficients, variances, and covariances. Monte Carlo standard errors are shown in parentheses.}
\end{table}

\begin{table}[ptbh]
\centering
\caption{Dense MNP Simulation Results: $\numalt=10$}
\label{tab:emulator_vs_ghk:10:all}
\footnotesize
\begin{tabular}{@{}rllrrrrr@{}}
\toprule
\multicolumn{1}{c}{$n$} & \multicolumn{1}{l}{DGP} & \multicolumn{1}{l}{Method} & \multicolumn{1}{c}{RMSE} & \multicolumn{1}{c}{RMS Bias} & \multicolumn{1}{c}{SE Ratio} & \multicolumn{1}{c}{Coverage} & \multicolumn{1}{c}{Time (s)} \\
\midrule
\multirow{5}{*}{1,000} & Emulator & Emulator & 0.319 (0.014) & 0.192 (0.011) & 1.45 (0.07) & 0.94 (0.01) & 2.23 (0.02) \\
 & MNP & Emulator & 0.342 (0.016) & 0.208 (0.012) & 1.87 (0.38) & 0.94 (0.01) & 2.18 (0.02) \\
 & MNP & GHK(10) & 0.388 (0.017) & 0.216 (0.013) & 0.29 (0.06) & 0.21 (0.02) & 5.90 (0.38) \\
 & MNP & GHK(50) & 0.363 (0.016) & 0.215 (0.013) & 0.67 (0.06) & 0.53 (0.03) & 13.24 (0.62) \\
 & MNP & GHK(250) & 0.351 (0.017) & 0.214 (0.012) & 1.44 (0.12) & 0.91 (0.01) & 14.63 (0.39) \\
\midrule
\multirow{5}{*}{10,000} & Emulator & Emulator & 0.171 (0.009) & 0.072 (0.007) & 1.52 (0.08) & 0.96 (0.01) & 19.18 (0.27) \\
 & MNP & Emulator & 0.198 (0.011) & 0.081 (0.010) & 1.68 (0.13) & 0.95 (0.01) & 16.79 (0.21) \\
 & MNP & GHK(10) & 0.235 (0.014) & 0.123 (0.010) & 1.32 (0.35) & 0.79 (0.02) & 8.70 (0.29) \\
 & MNP & GHK(50) & 0.203 (0.010) & 0.090 (0.009) & 1.27 (0.07) & 0.91 (0.01) & 34.10 (0.85) \\
 & MNP & GHK(250) & 0.199 (0.010) & 0.079 (0.010) & 1.40 (0.07) & 0.95 (0.01) & 47.98 (0.74) \\
\midrule
\multirow{5}{*}{100,000} & Emulator & Emulator & 0.070 (0.004) & 0.013 (0.003) & 1.08 (0.04) & 0.95 (0.01) & 143.74 (1.18) \\
 & MNP & Emulator & 0.082 (0.005) & 0.019 (0.004) & 0.93 (0.03) & 0.93 (0.01) & 143.66 (1.28) \\
 & MNP & GHK(10) & 0.138 (0.009) & 0.073 (0.006) & 0.48 (0.02) & 0.60 (0.01) & 31.11 (0.34) \\
 & MNP & GHK(50) & 0.085 (0.005) & 0.027 (0.004) & 0.89 (0.03) & 0.90 (0.01) & 80.62 (0.80) \\
 & MNP & GHK(250) & 0.081 (0.005) & 0.020 (0.004) & 0.97 (0.03) & 0.93 (0.01) & 348.26 (4.36) \\
\bottomrule
\end{tabular}

\floatnote{Performance metrics are averaged across coefficients, variances, and covariances. Monte Carlo standard errors are shown in parentheses.}
\end{table}

Table \ref{tab:emulator_vs_ghk:10:all} displays performance results for $\numalt=10$. The emulator performs comparably to GHK(250) with significantly less computation. At $n=1{,}000$, the emulator fits the model in less than 3 seconds on average, outpacing even GHK(10). The computation time of the emulator falls between that of GHK(10) and GHK(50) at $n=10{,}000$ and between GHK(50) and GHK(250) at $n=100{,}000$. Compared to Tables \ref{tab:emulator_vs_ghk:3:all} and \ref{tab:emulator_vs_ghk:5:all}, the standard error estimates are less reliable for the smaller sample sizes; the problem is amplified if we omit the penalty on $\cov$. The standard errors improve substantially with the sample size, reaching SE Ratios of 0.93--1.08 for the Emulator and GHK(250) methods at $n=100{,}000$.

\vspace{12pt}
\begin{table}[ptbh]
\centering
\caption{Factor MNP Simulation Results: $\numalt=3$}
\label{tab:emulator_vs_ghk:3-factor:all}
\footnotesize
\begin{tabular}{@{}rllrrrrr@{}}
\toprule
\multicolumn{1}{c}{$n$} & \multicolumn{1}{l}{DGP} & \multicolumn{1}{l}{Method} & \multicolumn{1}{c}{RMSE} & \multicolumn{1}{c}{RMS Bias} & \multicolumn{1}{c}{SE Ratio} & \multicolumn{1}{c}{Coverage} & \multicolumn{1}{c}{Time (s)} \\
\midrule
\multirow{5}{*}{1,000} & Emulator & Emulator & 0.089 (0.004) & 0.022 (0.006) & 1.17 (0.05) & 0.94 (0.01) & 0.11 (0.00) \\
 & MNP & Emulator & 0.095 (0.004) & 0.007 (0.005) & 1.10 (0.04) & 0.96 (0.01) & 0.10 (0.00) \\
 & MNP & GHK(10) & 0.096 (0.004) & 0.011 (0.005) & 1.08 (0.04) & 0.96 (0.01) & 0.08 (0.00) \\
 & MNP & GHK(50) & 0.094 (0.004) & 0.007 (0.005) & 1.11 (0.04) & 0.96 (0.01) & 0.15 (0.00) \\
 & MNP & GHK(250) & 0.095 (0.004) & 0.007 (0.005) & 1.10 (0.04) & 0.96 (0.01) & 0.41 (0.01) \\
\midrule
\multirow{5}{*}{10,000} & Emulator & Emulator & 0.038 (0.002) & 0.003 (0.002) & 1.00 (0.04) & 0.93 (0.01) & 0.39 (0.00) \\
 & MNP & Emulator & 0.038 (0.002) & 0.007 (0.003) & 1.00 (0.05) & 0.96 (0.01) & 0.39 (0.00) \\
 & MNP & GHK(10) & 0.040 (0.002) & 0.013 (0.002) & 0.99 (0.05) & 0.95 (0.01) & 0.32 (0.00) \\
 & MNP & GHK(50) & 0.038 (0.002) & 0.008 (0.003) & 1.00 (0.05) & 0.96 (0.01) & 0.46 (0.01) \\
 & MNP & GHK(250) & 0.038 (0.002) & 0.007 (0.003) & 1.00 (0.05) & 0.95 (0.01) & 0.91 (0.01) \\
\midrule
\multirow{5}{*}{100,000} & Emulator & Emulator & 0.011 (0.001) & 0.001 (0.001) & 1.07 (0.05) & 0.96 (0.01) & 1.86 (0.01) \\
 & MNP & Emulator & 0.011 (0.001) & 0.001 (0.001) & 1.04 (0.05) & 0.96 (0.01) & 1.85 (0.01) \\
 & MNP & GHK(10) & 0.016 (0.001) & 0.009 (0.001) & 0.88 (0.04) & 0.81 (0.02) & 0.93 (0.01) \\
 & MNP & GHK(50) & 0.012 (0.001) & 0.002 (0.001) & 1.04 (0.04) & 0.95 (0.01) & 1.96 (0.02) \\
 & MNP & GHK(250) & 0.011 (0.001) & 0.001 (0.001) & 1.04 (0.05) & 0.96 (0.01) & 7.23 (0.09) \\
\bottomrule
\end{tabular}

\vspace{6pt}

\floatnote{Performance metrics are averaged across coefficients, uniquenesses, and factor loadings. Monte Carlo standard errors are shown in parentheses.}
\end{table}

Table \ref{tab:emulator_vs_ghk:3-factor:all} displays analogous results for the Factor specification with $\numalt=3$. Performance metrics are averaged over the nonredundant elements of $\vec{\beta}$, $\mat{\Psi}$, and $\vec{\gamma}$. The Emulator's performance mirrors that of GHK(250) but at a much lower computational cost.

Tables~\ref{tab:emulator_vs_ghk:5-factor:all} and~\ref{tab:emulator_vs_ghk:10-factor:all} display analogous results for 5 and 10 alternatives, respectively. In both cases, the Emulator performs comparably to GHK(50) and GHK(250) despite not being trained specifically for factor models. The Emulator is comparable to or faster than GHK(50) for $n \in \{1{,}000, 10{,}000\}$. At $n = 100{,}000$, the Emulator's computation time falls between that of GHK(50) and GHK(250).

\begin{table}[ptbh]
\centering
\caption{Factor MNP Simulation Results: $\numalt=5$}
\label{tab:emulator_vs_ghk:5-factor:all}
\footnotesize
\begin{tabular}{@{}rllrrrrr@{}}
\toprule
\multicolumn{1}{c}{$n$} & \multicolumn{1}{l}{DGP} & \multicolumn{1}{l}{Method} & \multicolumn{1}{c}{RMSE} & \multicolumn{1}{c}{RMS Bias} & \multicolumn{1}{c}{SE Ratio} & \multicolumn{1}{c}{Coverage} & \multicolumn{1}{c}{Time (s)} \\
\midrule
\multirow{5}{*}{1,000} & Emulator & Emulator & 0.128 (0.004) & 0.016 (0.005) & 1.20 (0.04) & 0.97 (0.01) & 0.32 (0.01) \\
 & MNP & Emulator & 0.132 (0.003) & 0.012 (0.004) & 1.19 (0.03) & 0.97 (0.01) & 0.32 (0.00) \\
 & MNP & GHK(10) & 0.134 (0.004) & 0.017 (0.005) & 1.12 (0.03) & 0.95 (0.01) & 0.23 (0.00) \\
 & MNP & GHK(50) & 0.131 (0.003) & 0.012 (0.004) & 1.19 (0.03) & 0.98 (0.01) & 0.47 (0.01) \\
 & MNP & GHK(250) & 0.132 (0.003) & 0.012 (0.004) & 1.19 (0.03) & 0.97 (0.01) & 1.57 (0.02) \\
\midrule
\multirow{5}{*}{10,000} & Emulator & Emulator & 0.059 (0.002) & 0.007 (0.003) & 1.01 (0.04) & 0.95 (0.01) & 1.98 (0.03) \\
 & MNP & Emulator & 0.059 (0.003) & 0.007 (0.003) & 1.04 (0.04) & 0.96 (0.01) & 1.96 (0.02) \\
 & MNP & GHK(10) & 0.062 (0.003) & 0.011 (0.003) & 0.93 (0.04) & 0.93 (0.01) & 0.73 (0.01) \\
 & MNP & GHK(50) & 0.059 (0.003) & 0.004 (0.003) & 1.02 (0.04) & 0.95 (0.01) & 1.89 (0.03) \\
 & MNP & GHK(250) & 0.059 (0.003) & 0.007 (0.003) & 1.04 (0.04) & 0.95 (0.01) & 4.93 (0.07) \\
\midrule
\multirow{5}{*}{100,000} & Emulator & Emulator & 0.019 (0.001) & 0.002 (0.001) & 1.01 (0.03) & 0.96 (0.01) & 15.97 (0.18) \\
 & MNP & Emulator & 0.019 (0.001) & 0.002 (0.001) & 1.04 (0.03) & 0.96 (0.01) & 15.96 (0.18) \\
 & MNP & GHK(10) & 0.029 (0.001) & 0.017 (0.001) & 0.75 (0.03) & 0.75 (0.02) & 3.53 (0.04) \\
 & MNP & GHK(50) & 0.019 (0.001) & 0.004 (0.001) & 1.02 (0.03) & 0.95 (0.01) & 7.58 (0.08) \\
 & MNP & GHK(250) & 0.019 (0.001) & 0.002 (0.001) & 1.04 (0.03) & 0.96 (0.01) & 30.74 (0.45) \\
\bottomrule
\end{tabular}

\floatnote{Performance metrics are averaged across coefficients, uniquenesses, and factor loadings. Monte Carlo standard errors are shown in parentheses.}
\end{table}

\begin{table}[ptbh]
\centering
\caption{Factor MNP Simulation Results: $\numalt=10$}
\label{tab:emulator_vs_ghk:10-factor:all}
\footnotesize
\begin{tabular}{@{}rllrrrrr@{}}
\toprule
\multicolumn{1}{c}{$n$} & \multicolumn{1}{l}{DGP} & \multicolumn{1}{l}{Method} & \multicolumn{1}{c}{RMSE} & \multicolumn{1}{c}{RMS Bias} & \multicolumn{1}{c}{SE Ratio} & \multicolumn{1}{c}{Coverage} & \multicolumn{1}{c}{Time (s)} \\
\midrule
\multirow{5}{*}{1,000} & Emulator & Emulator & 0.188 (0.004) & 0.023 (0.004) & 1.14 (0.02) & 0.96 (0.00) & 1.08 (0.01) \\
 & MNP & Emulator & 0.198 (0.005) & 0.030 (0.005) & 1.10 (0.02) & 0.96 (0.01) & 1.03 (0.01) \\
 & MNP & GHK(10) & 0.220 (0.005) & 0.059 (0.006) & 0.92 (0.02) & 0.91 (0.01) & 0.76 (0.01) \\
 & MNP & GHK(50) & 0.205 (0.005) & 0.039 (0.005) & 1.06 (0.02) & 0.95 (0.01) & 2.46 (0.03) \\
 & MNP & GHK(250) & 0.199 (0.005) & 0.030 (0.005) & 1.11 (0.02) & 0.96 (0.01) & 5.80 (0.09) \\
\midrule
\multirow{5}{*}{10,000} & Emulator & Emulator & 0.080 (0.002) & 0.011 (0.003) & 1.22 (0.04) & 0.97 (0.01) & 6.95 (0.06) \\
 & MNP & Emulator & 0.081 (0.002) & 0.010 (0.002) & 1.22 (0.04) & 0.97 (0.01) & 6.83 (0.06) \\
 & MNP & GHK(10) & 0.103 (0.002) & 0.049 (0.003) & 0.95 (0.03) & 0.88 (0.01) & 3.92 (0.05) \\
 & MNP & GHK(50) & 0.087 (0.002) & 0.019 (0.003) & 1.14 (0.04) & 0.96 (0.01) & 14.67 (0.15) \\
 & MNP & GHK(250) & 0.084 (0.002) & 0.012 (0.003) & 1.20 (0.04) & 0.97 (0.01) & 23.93 (0.23) \\
\midrule
\multirow{5}{*}{100,000} & Emulator & Emulator & 0.042 (0.002) & 0.003 (0.002) & 0.95 (0.04) & 0.94 (0.01) & 62.10 (0.51) \\
 & MNP & Emulator & 0.037 (0.001) & 0.010 (0.002) & 1.13 (0.05) & 0.97 (0.01) & 59.63 (0.55) \\
 & MNP & GHK(10) & 0.071 (0.002) & 0.047 (0.002) & 0.65 (0.02) & 0.59 (0.02) & 19.50 (0.22) \\
 & MNP & GHK(50) & 0.039 (0.001) & 0.012 (0.001) & 1.04 (0.04) & 0.93 (0.01) & 48.68 (0.65) \\
 & MNP & GHK(250) & 0.036 (0.001) & 0.004 (0.001) & 1.11 (0.05) & 0.96 (0.01) & 190.72 (2.50) \\
\bottomrule
\end{tabular}

\floatnote{Performance metrics are averaged across coefficients, uniquenesses, and factor loadings. Monte Carlo standard errors are shown in parentheses.}
\end{table}

Across all scenarios, the Emulator method performs competitively with or significantly better than GHK. We expect that the relative performance of the emulator would increase with (a) additional training and (b) implementation on specialized deep-learning hardware, such as graphics processing units (GPUs). The latter is likely to be especially advantageous for the emulator because its computations are trivially parallelizable and highly optimized on modern GPUs.

\subsection{Multi-\texorpdfstring{$\numalt$}{K} Training Accuracy}
\label{sec:sim-multi-k}

We conclude this section by providing numerical results from training a neural network emulator to approximate choice probabilities for $\numalt \in \{3, 4, 5\}$ simultaneously. We use the same neural network as we used in the previous section for $\numalt=5$, with $\numalt$ added as a pass-through feature. As above, we train for $400{,}000$ optimizer steps, but we loop over training examples with $\numalt=3$, $4$, and $5$ alternatives within each optimizer step. The training process completed in 10.3 hours.

\begin{figure}[ptbh]
\centering
\caption{Multi-$\numalt$ Emulator Training and Accuracy}
\label{fig:345}
\begin{subfigure}[b]{0.49\textwidth}
\centering
\includegraphics[width=\textwidth]{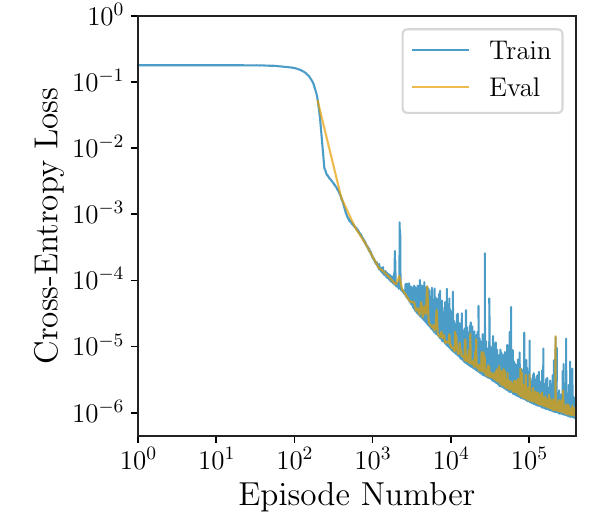}
\caption{Training Loss}
\label{fig:loss:345}
\end{subfigure}
\hfill
\begin{subfigure}[b]{0.49\textwidth}
\centering
\includegraphics[width=\textwidth]{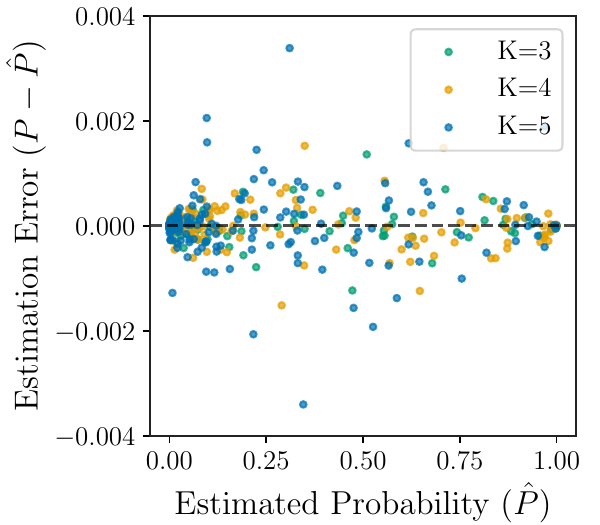}
\caption{Estimation Accuracy}
\label{fig:errors:345}
\end{subfigure}

\floatnote{Panel (a) shows training loss (blue) and evaluation loss (orange) over $400{,}000$ optimizer steps for $\numalt \in \{3,4,5\}$. Panel (b) shows choice-probability estimation errors for 400 simulated examples at each value of $\numalt$.}
\end{figure}

Figure \ref{fig:loss:345} shows the training and evaluation loss steadily decreasing throughout the training procedure, approaching levels similar to those of Figure \ref{fig:loss}. Figure \ref{fig:errors:345} plots the estimated choice probabilities against estimation errors for $\numalt \in \{3, 4, 5\}$ with 400 simulated examples each. The small errors indicate that the model accurately approximates the choice probabilities across all three values of $\numalt$ simultaneously.

\section{Application to Retail Laundry-Detergent Purchases}
\label{sec:application}

% BEGIN INLINE APPLICATION (formerly application.tex)
We apply our proposed amortized inference approach to models of consumer choice among laundry detergents. The application compares panel-data mixed probit and mixed logit models with rich permanent taste heterogeneity and IID period errors to probit and Gumbel models that retain this heterogeneity while allowing the period errors to be correlated.

\subsection{Data Description}
\label{sec:application-data}

The data set is derived from the NielsenIQ Retail Scanner and Consumer Panel data sets provided by the University of Chicago Kilts Center for Marketing.
We use data from the time period 2021--2024.
We construct alternatives by aggregating universal product codes (UPCs) to the level of (brand, form) pairs, where the form is either liquid, powder, or unit dose.
The latter includes, for example, Tide pods and related products.
After aggregating to the brand-form level, we limit the set of alternatives to those representing at least 2\% market share by total revenue in the Consumer Panel data, which results in the ten alternatives shown in Table \ref{tab:brand-forms} and an Other category comprising 11.63\% of the observed choices.
No powders exceed 2\% market share, so all powder-form detergents are placed in the Other category.

\begin{table}[htbp]
    \centering
\caption{Laundry Detergent Alternatives}
\label{tab:brand-forms}

{\small
\begin{tabular}{lcrrrrr}
\toprule
Brand & Form & Share (\%) & Price/Load (\$) & Coupon (\%) & Promo (\%) & Missing (\%) \\
\midrule
Tide & Liquid & 27.69 & 0.1334 & 5.52 & 6.25 & 0.00 \\
A\&H & Liquid & 14.83 & 0.0821 & 2.97 & 3.37 & 0.51 \\
Tide & Unit Dose & 12.16 & 0.3428 & 4.86 & 5.13 & 0.06 \\
Other & - & 11.63 & 0.1565 & 2.69 & 1.69 & 0.14 \\
All & Liquid & 10.46 & 0.1168 & 5.94 & 2.96 & 0.47 \\
Purex & Liquid & 5.82 & 0.0687 & 4.46 & 2.34 & 11.08 \\
Persil & Liquid & 5.20 & 0.1393 & 8.10 & 2.80 & 1.45 \\
Gain & Liquid & 5.09 & 0.1140 & 4.34 & 4.33 & 0.26 \\
All & Unit Dose & 2.56 & 0.2695 & 4.19 & 1.65 & 4.34 \\
Gain & Unit Dose & 2.44 & 0.3071 & 4.13 & 3.14 & 1.19 \\
Masked & Liquid & 2.14 & 0.0979 & 4.73 & 1.31 & 12.31 \\
\bottomrule
\end{tabular}
}

\floatnote{A\&H denotes Arm \& Hammer. Share is the percentage of choice occasions on which the alternative is selected, and Price/Load is the average price per load. Coupon and Promo are sample averages of the weighted coupon and promotion variables defined in Section \ref{sec:application-models}. Missing is the percentage of choice sets for which a scanner price is unavailable. Masked denotes anonymized store brands. Figures are national averages over all the retail chains in the data.}
\end{table}

The Consumer Panel data set provides the observed choices.\footnote{When consumers purchase more than one detergent in a single shopping trip, we break ties based on the total price paid net of coupons. If the prices are also exactly equal, we break the tie randomly, selecting a single alternative among all tied alternatives.}
We merge this data set with the Retail Scanner data to reconstruct plausible choice sets at each choice occasion.
To do so, we match observed choices with scanner-data prices at the same store within two weeks of the purchase, with preference given to prices observed closest to the purchase week.
This matching procedure filters out purchases at stores not included in the Retail Scanner data set.
In cases where a valid price is not available within two weeks of the purchase data, we mark the alternative as unavailable and omit it from the choice set.
We filter out a very small number of choice occasions in which Tide liquid is unavailable.

These processing steps result in 33,284 households and 192,001 individual choice occasions.
We further filter down to households with 10--30 purchases, resulting in 4,900 households and 78,047 choice occasions.
The 10-purchase lower bound is motivated by \citet{keane1997}, who finds little sensitivity of estimated state dependence to the treatment of initial conditions in a panel averaging about 10 purchase occasions. The 30-purchase upper bound removes a small number of households with an unusually high number of purchases.  

\subsection{Description of Choice Models}
\label{sec:application-models}

We fit panel-data choice models with household-level random effects, in which household $i$'s utility for alternative $j$ at time $t$ is
\begin{align}
U_{ijt}=&\  
    \alpha_j + \beta_p p_{ijt} + \beta_c c_{ijt} + \beta_m m_{ijt} + \label{eq:fixed-effects}\\
    &\ + \beta_{BF} L_{ijt}^{BF} + \beta_B L_{ijt}^B + \beta_F L_{ijt}^F\label{eq:loyalty}\\
    &\ + a_{ij} + b_{i,p} p_{ijt} + b_{i,c} c_{ijt} + b_{i,m} m_{ijt} + \epsilon_{ijt} \label{eq:random-effects}.
\end{align}
In \eqref{eq:fixed-effects}, $\alpha_j$ is an alternative-specific constant (normalized to zero for $j=1$: Tide liquid), and $\beta_p$, $\beta_c$, and $\beta_m$ are mean coefficients on price, coupon availability, and promotions, respectively.
The latter three variables are constructed as weighted averages of the per-load price in dollars, the coupon value as a percent of total price, and an indicator for whether the UPC was being promoted (set to one if either the ``feature'' or ``display'' flags are set in the NielsenIQ data); the coupon variable is exponentially smoothed with weight 0.5.

The terms in \eqref{eq:loyalty} represent utility contributions from loyalty to brand-form combinations (BF), individual brands (B), and the product form (F).
Following \citet{guadagni1983logit}, we define these variables as exponentially weighted averages:
\[
L_{ijt}^g = (1 - \rho) \sum_{\ell=1}^{t-1} \rho^{\ell - 1} I_{ij,t-\ell}^g,
\]
where $g \in \{BF, B, F\}$, $\rho = 0.9$, and $I_{ij,t-\ell}^g$ is an indicator for whether the household selected an alternative whose value of $g$ (brand-form, brand, or form) matches that of alternative $j$. We initialize the loyalty measures to zero prior to the first observed purchase.

The coefficients $a_{ij}$, $b_{i,p}$, $b_{i,c}$, and $b_{i,m}$ are household-level random effects.
We set $a_{i1}=0$ and assume that
\[
    \left(\vec{a}_{i,2:11}^{\top},\, b_{i,p},\, b_{i,c},\, b_{i,m}\right)^{\top}
    \overset{iid}{\sim}
    \mathcal{N}(\vec{0}_{13}, \mat{C}),
\]
where $\mat{C} \in \R^{13 \times 13}$ is an estimated covariance matrix.

Finally, we assume that $\vec{\epsilon}_{it} = \cov^{1/2} \vec{\epsilon}^*$ as described in Section \ref{sec:choice-models}, and we test two different models for the distribution of $\vec{\epsilon}^*$: standard normal and $\text{Gumbel}(-\gamma, 1)$, corresponding to the MNP and correlated-Gumbel models, respectively.
We parameterize $\mat{\Omega} \in \R^{10 \times 10}$ according to a factor structure as follows:
\[
\mat{\Omega} = \mat{\Psi} + \mat{\Gamma} \mat{\Gamma}^{\top},
\]
where $\mat{\Psi} \in \R^{10 \times 10}$ is a diagonal matrix of uniquenesses and $\mat{\Gamma} \in \R^{10 \times Q}$ is a lower-triangular loading matrix with $Q \in \N$ factors.
As in Section \ref{sec:simulation}, we impose an anchored factor parameterization. We set $\Gamma_{1,1}=1$, constrain the remaining main-diagonal elements of $\mat{\Gamma}$ to be non-negative, and set $\Psi_{qq}=0$ for all $q \leq Q$. With $\Psi_{1,1}=0$, $\Gamma_{1,1}=1$ fixes the overall utility scale, while the remaining diagonal sign restrictions fix factor signs. The zero-uniqueness restrictions generalize the one-factor anchoring device.\footnote{As in the one-factor case, the anchored restrictions are stronger than required for theoretical identification, but help eliminate a very flat likelihood along directions that reallocate variance between factor loadings and uniquenesses.}
The mixed-MNL model sets $\mat{\Omega}=\mat{I}$ while allowing $\mat{C}$ to be unrestricted, with $\epsilon_j^* \overset{iid}{\sim}\text{Gumbel}(a, 1)$.

We estimate the resulting models via penalized maximum likelihood as in Section \ref{sec:simulation} with the log likelihoods computed via emulator evaluations.
Appendix \ref{sec:application-details} provides additional details on the emulators and model-fitting procedure.

\subsection{Results}
\label{sec:application-results}

Table \ref{tab:agg-results} summarizes the quality of the model fit via the Bayesian information criterion (BIC) for various combinations of the model family (MNP vs. correlated-Gumbel) and the period-specific covariance structure $\mat{\Omega}$.
The latter consists of either (i) IID errors, (ii) independent heteroscedastic errors, (iii) factor models with 1--7 factors, or (iv) a dense covariance structure (all elements of $\mat{\Omega}$ are free except $\Omega_{1,1}$ is fixed at one).

\begin{table}[htb]
\centering
\caption{Application Model Fit and Computation Time}
\label{tab:agg-results}
{\small
\begin{tabular}{llrcc}
\toprule
Family & Cov. Type & \multicolumn{1}{c}{BIC} & \# Params. & Time (h) \\
\midrule
\multirow{10}{*}{MNP} & IID & 192,545.0 & 107 & 0.2 \\
 & Independent & 191,846.0 & 116 & 2.3 \\
 & 1-Factor & 191,635.2 & 125 & 0.9 \\
 & 2-Factor & 191,639.2 & 133 & 1.5 \\
 & 3-Factor & \textbf{191,539.6} & 140 & 1.1 \\
 & 4-Factor & 191,568.3 & 146 & 0.8 \\
 & 5-Factor & 191,608.8 & 151 & 0.6 \\
 & 6-Factor & 191,619.5 & 155 & 1.7 \\
 & 7-Factor & 191,645.0 & 158 & 0.3 \\
 & Dense & 191,564.7 & 161 & 7.6\vspace{6pt}\\
\multirow{10}{*}{correlated-Gumbel} & IID & 192,106.8 & 107 & 0.9 \\
 & Independent & 191,379.0 & 116 & 2.9 \\
 & 1-Factor & 191,120.1 & 125 & 3.1 \\
 & 2-Factor & 191,155.4 & 133 & 2.4 \\
 & 3-Factor & \textbf{191,092.6} & 140 & 5.0 \\
 & 4-Factor & 191,122.4 & 146 & 1.9 \\
 & 5-Factor & 191,162.5 & 151 & 3.0 \\
 & 6-Factor & 191,191.5 & 155 & 2.7 \\
 & 7-Factor & 191,217.0 & 158 & 1.5 \\
 & Dense & 191,097.3 & 161 & 8.1 \\
\bottomrule
\end{tabular}
}

\floatnote{Cov. Type denotes the covariance structure of the period-specific errors. Bold BIC values indicate the best fit within each model family. The columns \# Params. and Time (h) denote the total number of model parameters and the model-fitting time in hours.}
\end{table}
BIC selects a 3-factor covariance structure for both the MNP and correlated-Gumbel models, with the preferred correlated-Gumbel model outperforming the preferred MNP model by about 447 BIC units.
Both models improve the BIC by hundreds of units by allowing heteroscedastic errors and adding correlation structure via factors.
\FloatBarrier
In particular, asymptotic likelihood ratio tests produce $\chi^2$ statistics in the range 280--810 for Independent vs. IID and 1-Factor vs. Independent covariance structures for \textit{both} the MNP and correlated-Gumbel models, providing strong evidence against IID  period-specific errors in this application.\footnote{Across specifications, the models contain 107--161 parameters and fit in 0.2--8.1 hours on an NVIDIA H100 GPU.}
Thus, the mixed-MNL specification, which imposes IID period-specific Gumbel errors, is rejected in favor of the correlated-Gumbel model, providing strong evidence against conditional IIA even after allowing for correlated household random effects.

Having shown that the IID restriction is strongly rejected, we next assess whether relaxing it materially changes predicted substitution patterns.
To do so, we conduct a representative-consumer product-deletion experiment designed to isolate the role of the period-error structure.
For each preferred model, we calibrate a fixed vector of deterministic utilities, $\util$, so that baseline choice probabilities reproduce the observed product shares in Table \ref{tab:brand-forms}.
We then delete product $k$ while holding the utilities and period-error distribution of the surviving products fixed, and compute the diversion rate
\[
D_{k \to j} = \frac{P_j^{(-k)} - P_j}{P_k}, \qquad j \neq k,
\]
which is the fraction of choices originally going to $k$ that are diverted to $j$.
Under conventional mixed MNL, conditional on persistent tastes, IIA implies
\[
D_{k \to j}^{\text{ML}} = \frac{P_j}{1 - P_k}.
\]
We therefore compute the diversion ratio
\[
R_{k \to j} = \frac{D_{k \to j}}{D_{k \to j}^{\text{ML}}}.
\]
We compute the diversion ratios using 100 million independent Monte Carlo draws.  Values above one indicate stronger substitution than the mixed-MNL IIA benchmark, while values below one indicate weaker substitution.
Figure \ref{fig:diversion-ratios} reports these ratios for the 3-factor MNP and correlated-Gumbel models.
\vspace{6pt}

\begin{figure}[H]
\centering
\caption{Diversion Ratios Relative to Mixed MNL}
\label{fig:diversion-ratios}
\begin{subfigure}[b]{0.49\textwidth}
\centering
\includegraphics[width=\textwidth]{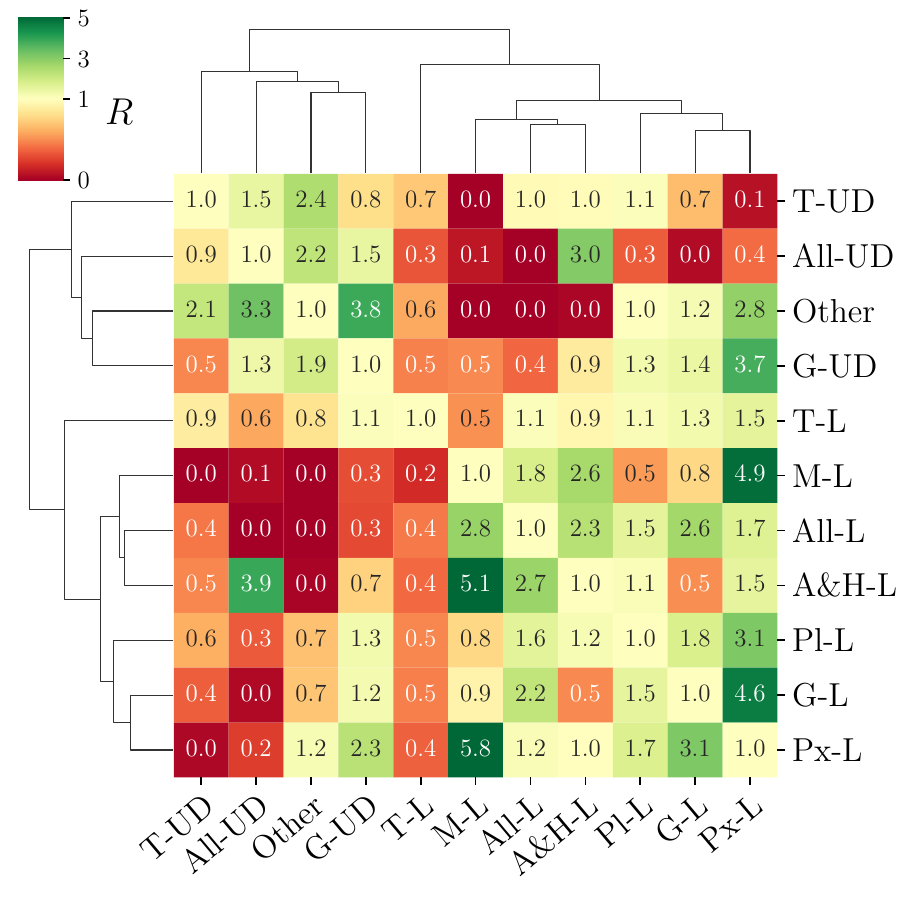}
\caption{MNP}
\label{fig:diversion-ratio-mnp}
\end{subfigure}
\hfill
\begin{subfigure}[b]{0.49\textwidth}
\centering
\includegraphics[width=\textwidth]{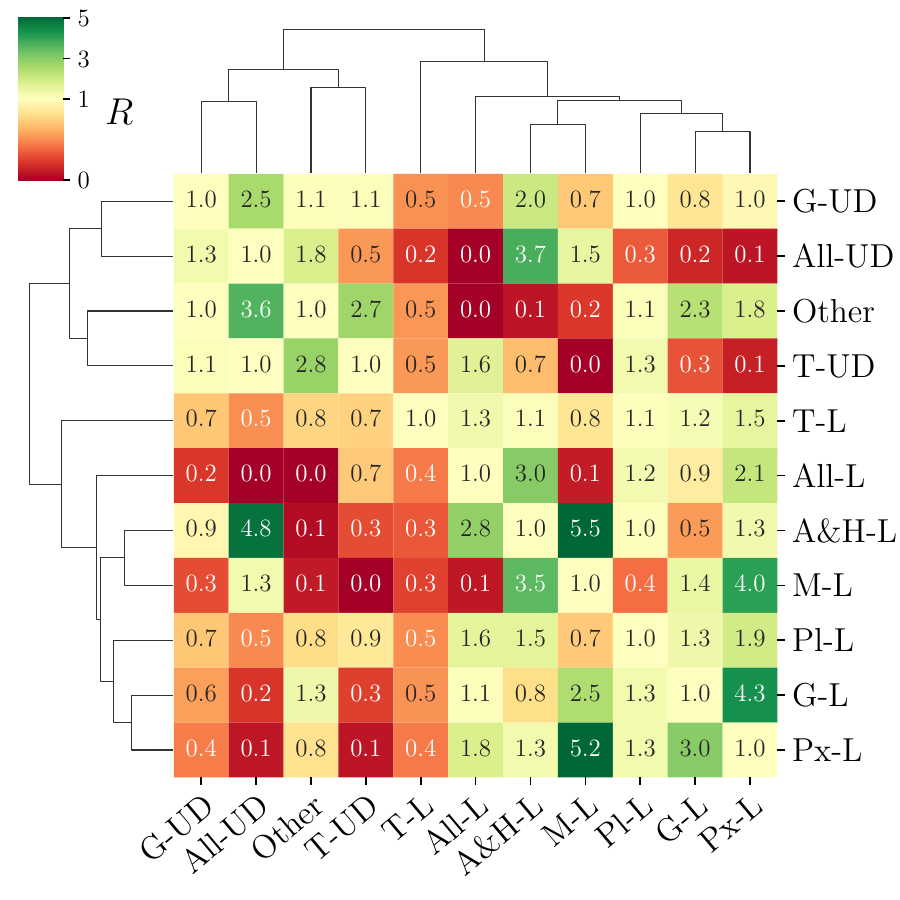}
\caption{Correlated-Gumbel}
\label{fig:diversion-ratios-gumbel}
\end{subfigure}

\floatnote{$R=D/D^{\mathrm{ML}}$, where $D^{\mathrm{ML}}$ is the diversion implied by IIA. Values above (below) one indicate more (less) diversion than under IIA.
Dendrograms use complete-linkage clustering based on $\SD(\epsilon_j,\epsilon_k)$. L denotes liquid and U unit dose; brand abbreviations are Tide (T), Arm \& Hammer (A\&H), Purex (Px), Persil (Pl), Gain (G), and Masked (M).}
\end{figure}
\newpage

Several interesting patterns emerge in Figure \ref{fig:diversion-ratios}.
First, both models exhibit large, systematic departures from IIA, as evidenced by the values of $R_{k \to j}$ far from one. This includes many values less than 0.1, shown in red, and many values greater than 2.0, shown in green.
Second, the alternatives cluster in interpretable ways. The dendrograms along the side of Figure \ref{fig:diversion-ratios} provide a clustering of the alternatives using complete linkage and a distance matrix defined by specifying distances as $\SD(\epsilon_j, \epsilon_k)$.
Notice that the unit dose alternatives cluster together in both models, as do the liquid forms of Arm \& Hammer, Gain, Purex, and Persil.
Third, on average, alternatives exhibit a much smaller degree of substitution across forms (liquid vs. unit dose) than they do within the same form.
Similar patterns emerge under the corresponding dense models.
Overall, the results in Figure \ref{fig:diversion-ratios} produce interesting, actionable insights that would be missed in a model assuming IID period-specific errors.
% END INLINE APPLICATION

\section{Conclusion}
\label{sec:conclusion}

This paper develops an amortized inference approach that makes flexible correlated-error discrete choice models practical to estimate.
Rather than repeatedly approximating high-dimensional choice probabilities at each likelihood evaluation, we train an equivariant neural network to approximate the generic choice-probability map and then reuse the resulting emulator throughout estimation.
This approach accommodates large choice sets, rich covariance structures, and general error distributions while retaining the economic structure and interpretability of RUM models.

Because the emulator operates on the generic choice problem, it is agnostic to the assumed parametric form of the deterministic utilities or the accompanying scale (covariance) matrix.
Given a pretrained emulator, this generality enables drop-in adoption without the need for custom architecture changes or additional training.
In particular, generalizing from logit to probit requires only the replacement of closed-form softmax probabilities with emulator evaluations.
Other correlated error distributions are similarly straightforward.

Our proposed architecture and training process feature a preprocessing transformation, DeepSet modules \citep{zaheer2017deep}, equivariant layers, smooth activation functions, Sobolev training \citep{czarnecki2017sobolev}, and a sufficient equivariant feature set. Together, these design choices produce smooth approximations of choice probabilities that respect their invariance properties and accelerate learning via parameter sharing.
Without careful attention to all these components, we find that the emulator may produce inaccurate predictions, jeopardizing emulator-based parameter estimates and their accompanying inferential procedures.
We show theoretically that the architecture possesses a universal approximation property, extending recent theoretical results for symmetric matrices \citep{blumsmith2025machine} to the joint space of utility vectors and centered scale (covariance) matrices.

In Monte Carlo simulations, we find that emulator-based estimators perform as well as or better than GHK-based estimators in terms of both statistical performance and computational requirements when estimating MNP models.
These findings support our theoretical results showing that emulator-based estimators inherit the properties of exact MLEs under mild approximation conditions.
We focus on MNP models in the simulation study to enable comparison to the GHK simulator, which is directly applicable only to Gaussian errors.
The computational benefits of the amortized inference approach are amplified by the fact that emulator evaluations are trivially parallelizable and highly optimized on modern deep learning hardware, such as GPUs.

We apply our amortized inference procedures to household choice among laundry detergents, estimating correlated-probit and correlated-Gumbel specifications with rich household-level random effects.
This allows us to compare conventional IID period-error specifications with models that add flexible cross-product correlation in the transitory shocks while retaining the same rich specification of persistent heterogeneity.
In this application, BIC and likelihood ratio tests strongly reject IID period-specific errors, providing evidence against the conditional-IIA restriction underlying conventional mixed MNL for these data.
Relaxing this restriction materially changes predicted substitution patterns, with many diversion ratios differing sharply from the proportional-diversion benchmark implied by mixed MNL.
Thus, in the laundry-detergent application, the additional flexibility enabled by our approach is economically consequential rather than merely a statistical improvement in fit.

Several limitations merit discussion. Practitioners must ensure sufficient overlap between the distribution of inputs in their application and the training distribution. Additionally, our asymptotic theory requires the emulator approximation error to vanish at an appropriate rate, meaning applications with very large samples may require more extensively trained emulators. Lastly, models with correlated errors require identifying restrictions not needed for independent logit-based models. Although our approach enables simple and efficient computation of choice probabilities for these models, the data requirements and modeling decisions to ensure proper identification remain a separate challenge.

The framework also allows numerous extensions. Our emulator architecture and training procedure could be specialized to specific settings, such as factor-structured scale matrices or mixture distributions. We expect these settings could be handled by reconfiguring the modules in our architecture and, in some cases, generalizing the preprocessing transformation. Our emulators could also be used in tandem with other modern choice modeling architectures, such as those described in Section \ref{sec:literature}. Previous work has applied neural-network-assisted estimation approaches to other econometric models with intractable likelihood functions, such as dynamic discrete choice \citep{Norets2012} and consumer search \citep{wei2025}, in the context of specific parametric models. Our approach is therefore complementary to these model-specific methods: where they exploit structure tailored to a particular economic model, amortized inference targets a reusable generic probability map that can be combined with such modeling architectures. The amortized inference approach presented here could likely be generalized to many of these settings, enabling efficient amortized inference without modifying or retraining the neural network.

\section*{Acknowledgments}
\vspace{-6pt}
The emulators for Section \ref{sec:application} were trained at the Advanced Research Computing at Hopkins (ARCH) core facility (\url{rockfish.jhu.edu}).
The authors acknowledge the use of AI tools to assist with mathematical derivations, computer code, and writing.
To respect the confidentiality of the data in the case study, Section \ref{sec:application} was written manually without providing AI tools any access to the underlying data.
All content was manually reviewed for accuracy.

\vspace{-12pt}
\section*{Data Availability}
\vspace{-6pt}
All reported empirical results are from our own analyses using data from Nielsen Consumer LLC and marketing databases provided through the NielsenIQ Datasets at the Kilts Center for Marketing Data Center at the University of Chicago Booth School of Business. The conclusions drawn from the NielsenIQ data are our own and do not reflect the views of NielsenIQ; NielsenIQ is not responsible for, had no role in, and was not involved in analyzing or preparing the results reported herein.

The NielsenIQ Retail Scanner and Consumer Panel data used in the empirical application are proprietary and cannot be redistributed by the authors. Researchers seeking to reproduce the empirical analysis must obtain access independently through the Kilts Center for Marketing at the University of Chicago Booth School of Business, subject to applicable terms of use.

\vspace{-12pt}
\section*{Code Availability}
\vspace{-6pt}
Computer code to reproduce the results in Sections \ref{sec:simulation} and \ref{sec:application} is available at \url{https://github.com/eastonhuch/amortized-inference-discrete-choice}.

\newpage

\bibliography{references}

\appendix
\newpage

\setcounter{page}{1}
\renewcommand{\thepage}{A-\arabic{page}}

\begin{center}
\Large\textbf{Supplementary Materials for ``Amortized Inference for Correlated Discrete Choice Models via Equivariant Neural Networks''}

\vspace{12pt}

\normalsize
Easton Huch\\
Postdoctoral Fellow, Johns Hopkins Carey Business School

\vspace{0.5em}

Michael Keane\\
Carey Distinguished Professor, Johns Hopkins Carey Business School
\end{center}

\vspace{20pt}

Appendix \ref{sec:separation} presents results and proofs regarding generic separation and universal approximation. Appendix \ref{sec:utility-covariance} proves the results given in Section \ref{sec:universal-approximation}. Appendix \ref{sec:estimation-proofs} discusses the emulator approximation assumptions and provides proofs of the estimation theory results in Sections \ref{sec:estimation-theory} and \ref{sec:misspecification}. Appendix \ref{sec:simulation-details} provides details on the emulator input features, emulator training, and simulation study. Appendix \ref{sec:gumbel-sim} compares the emulator with softmax-smoothed simulators, and Appendix \ref{sec:application-details} provides additional details on the application in Section \ref{sec:application}.

\section{Generic Separation and Universal Approximation}
\label{sec:separation}

This appendix develops the group-theoretic foundations underlying the per-alternative encoder architecture presented in the main text. We establish that certain invariants generically separate orbits under permutation groups, yielding universal approximation results via the Stone--Weierstrass theorem.

Let $\mathbb{S}_{\numalt}$ denote the space of symmetric real-valued $\numalt \times \numalt$ matrices. We consider an extended group $G_{\numalt}$ acting on $\mathbb{S}_{\numalt}$ by permutation, scaling, and translation by matrices in the null space of double-centering. Appendix~\ref{sec:preliminaries} introduces the necessary notation and establishes basic properties of the centering matrix and the normalized double-centering function $f$.

Appendices~\ref{sec:generic}--\ref{sec:basic} establish that the composition $g \circ f$, which first applies normalized double-centering and then extracts multisets of diagonal and off-diagonal elements, generically separates $G_{\numalt}$-orbits. The exceptional set where separation fails is closed, $G_{\numalt}$-invariant, and has measure zero.

Appendix~\ref{sec:restriction} considers the restriction to the centered, trace-normalized subspace $\mathbb{S}_{\numalt}^* = \mathrm{Im}(f)$, where $g$ generically separates $S_{\numalt}$-orbits and the Stone--Weierstrass theorem yields universal approximation.

Appendix~\ref{sec:submatrices} develops a variant for a distinguished index $j$, showing that an invariant $g_j$ generically separates orbits under $\permgroupsub$, the subgroup fixing $j$.

Appendices~\ref{sec:utility-covariance-global}--\ref{sec:utility-covariance} extend the theory to incorporate deterministic utilities, culminating in proofs of Theorem~\ref{thm:separation} and Theorem~\ref{thm:universal} from the main text.
Throughout, $\mat{X}$ and $\mat{Y}$ represent scale matrices. As in Section \ref{sec:architecture}, we sometimes refer to these matrices as covariance matrices although this interpretation is not strictly correct outside of MNP models.

\subsection{Preliminaries}\label{sec:preliminaries}

Throughout, let $\numalt \geq 2$ be the number of alternatives. Let $\mat{I}_{\numalt}$ denote the $\numalt \times \numalt$ identity matrix and $\ones_{\numalt} \in \R^{\numalt}$ the vector of all ones.
\vspace{6pt}

\begin{definition}[Centering Matrix]\label{def:centering}
The \emph{centering matrix} $\centermat \in \R^{\numalt \times \numalt}$ is defined as
$\centermat = \mat{I}_{\numalt} - \frac{1}{\numalt} \ones_{\numalt} \ones_{\numalt}^\top$.
\end{definition}

\vspace{3pt}
\begin{proposition}[Properties of the Centering Matrix]\label{prop:centering-properties}
The centering matrix $\centermat$ is symmetric ($\centermat^\top = \centermat$), idempotent ($\centermat^2 = \centermat$), and annihilates constant vectors ($\centermat \ones_{\numalt} = \vec{0}$).
\end{proposition}

\begin{proof}
Symmetry is immediate. For idempotency, $\centermat^2 = \mat{I}_{\numalt} - \frac{2}{\numalt}\ones_{\numalt}\ones_{\numalt}^\top + \frac{1}{\numalt^2}\ones_{\numalt}(\ones_{\numalt}^\top\ones_{\numalt})\ones_{\numalt}^\top = \mat{I}_{\numalt} - \frac{2}{\numalt}\ones_{\numalt}\ones_{\numalt}^\top + \frac{1}{\numalt}\ones_{\numalt}\ones_{\numalt}^\top = \centermat$. Finally, $\centermat\ones_{\numalt} = \ones_{\numalt} - \frac{1}{\numalt}\ones_{\numalt}(\ones_{\numalt}^\top\ones_{\numalt}) = \ones_{\numalt} - \ones_{\numalt} = \vec{0}$.
\end{proof}

\vspace{3pt}
\begin{definition}[Double-Centering Function]\label{def:double-centering}
Define the domain $\mathcal{U} = \{\mat{X} \in \mathbb{S}_{\numalt} : \tr(\centermat\mat{X}\centermat) \neq 0\}$ and the function $f: \mathcal{U} \to \mathbb{S}_{\numalt}$ by
\[
f(\mat{X}) = \frac{\numalt \cdot \centermat\mat{X}\centermat}{\tr(\centermat\mat{X}\centermat)}.
\]
\end{definition}

\vspace{3pt}
\begin{proposition}[Explicit Formula]\label{prop:explicit-formula}
For $\mat{X} \in \mathbb{S}_{\numalt}$, let $\bar{X}_{i\bigcdot} = \frac{1}{\numalt}\sum_{k=1}^{\numalt} X_{ik}$, $\bar{X}_{\bigcdot j} = \frac{1}{\numalt}\sum_{k=1}^{\numalt} X_{kj}$, and $\bar{X}_{\bigcdot\bigcdot} = \frac{1}{\numalt^2}\sum_{i,j} X_{ij}$. Then:
\begin{enumerate}
    \item[(i)] $[\centermat\mat{X}\centermat]_{ij} = X_{ij} - \bar{X}_{i\bigcdot} - \bar{X}_{\bigcdot j} + \bar{X}_{\bigcdot\bigcdot}$.
    \item[(ii)] $\tr(\centermat\mat{X}\centermat) = \tr(\mat{X}) - \numalt\bar{X}_{\bigcdot\bigcdot}$.
\end{enumerate}
\end{proposition}

\begin{proof}
Part (i) follows by direct expansion. For part (ii), use idempotency and the cyclic property of trace:
$\tr(\centermat\mat{X}\centermat) = \tr(\centermat^2\mat{X}) = \tr(\centermat\mat{X}) = \tr(\mat{X}) - \frac{1}{\numalt}\ones_{\numalt}^\top\mat{X}\ones_{\numalt} = \tr(\mat{X}) - \numalt\bar{X}_{\bigcdot\bigcdot}$.
\end{proof}

\vspace{6pt}
\begin{definition}[Extended Group Action]\label{def:group-action}
Let $S_{\numalt}$ denote the symmetric group on $\numalt$ elements with associated permutation matrices $\mat{P}_\pi$, and let $\R^* = \mathbb{R} \setminus \{0\}$. Define the extended group $G_{\numalt} = \R^{\numalt} \rtimes (\R^* \times S_{\numalt})$, which acts on $\mathbb{S}_{\numalt}$ by:
\[
(c, \vec{u}, \pi) \cdot \mat{X} = c \, \mat{P}_\pi \mat{X} \mat{P}_\pi^\top + \vec{u} \ones_{\numalt}^\top + \ones_{\numalt} \vec{u}^\top.
\]
\end{definition}

\vspace{5pt}
\begin{proposition}[Domain Properties]\label{prop:domain-measure}
The set $\mathbb{S}_{\numalt} \setminus \mathcal{U} = \{\mat{X} \in \mathbb{S}_{\numalt} : \tr(\centermat\mat{X}\centermat) = 0\}$ is a hyperplane in $\mathbb{S}_{\numalt}$ with Lebesgue measure zero. The domain $\mathcal{U}$ is open and $G_{\numalt}$-invariant.
\end{proposition}

\begin{proof}
By Proposition~\ref{prop:explicit-formula}(ii), $\tr(\centermat\mat{X}\centermat)$ is a nontrivial linear function, so its zero set is a measure-zero hyperplane. Openness follows from continuity. For $G_{\numalt}$-invariance: scaling by $c \neq 0$ gives $\tr(\centermat(c\mat{X})\centermat) = c\,\tr(\centermat\mat{X}\centermat) \neq 0$; null-space translation by $\mat{N} = \vec{u}\ones_{\numalt}^\top + \ones_{\numalt}\vec{u}^\top$ satisfies $\centermat\mat{N}\centermat = \mat{0}$ since $\centermat\ones_{\numalt} = \vec{0}$; permutation preserves the trace by the cyclic property.
\end{proof}

\vspace{5pt}
\begin{proposition}[$f$ and $G_{\numalt}$]\label{prop:f-and-G}
The function $f$ satisfies:
\begin{enumerate}
    \item[(i)] Equivariance under permutation: $f(\mat{P}_\pi \mat{X} \mat{P}_\pi^\top) = \mat{P}_\pi f(\mat{X}) \mat{P}_\pi^\top$.
    \item[(ii)] Invariance under scaling: $f(c \mat{X}) = f(\mat{X})$ for $c \neq 0$.
    \item[(iii)] Invariance under null-space translation: $f(\mat{X} + \vec{u}\ones_{\numalt}^\top + \ones_{\numalt}\vec{u}^\top) = f(\mat{X})$.
\end{enumerate}
\end{proposition}

\begin{proof}
For (i), since $\mat{P}_\pi\ones_{\numalt} = \ones_{\numalt}$, we have $\mat{P}_\pi\centermat = \centermat\mat{P}_\pi$, so $\centermat(\mat{P}_\pi\mat{X}\mat{P}_\pi^\top)\centermat = \mat{P}_\pi\centermat\mat{X}\centermat\mat{P}_\pi^\top$, and the trace is preserved. Part (ii) follows from $f(c\mat{X}) = \frac{c K\,\centermat\mat{X}\centermat}{c\cdot\tr(\centermat\mat{X}\centermat)} = f(\mat{X})$. For (iii), $\centermat(\vec{u}\ones_{\numalt}^\top + \ones_{\numalt}\vec{u}^\top)\centermat = \mat{0}$ since $\centermat\ones_{\numalt} = \vec{0}$.
\end{proof}

\vspace{5pt}
\begin{lemma}[Preimage Characterization]\label{lem:preimage}
If $f(\mat{X}) = f(\mat{Y})$ for $\mat{X}, \mat{Y} \in \mathcal{U}$, then there exist $c \in \R^*$ and $\vec{w} \in \R^{\numalt}$ such that $\mat{X} = c\mat{Y} + \vec{w}\ones_{\numalt}^\top + \ones_{\numalt}\vec{w}^\top$.
\end{lemma}

\begin{proof}
Let $t_X = \tr(\centermat\mat{X}\centermat)$ and $t_Y = \tr(\centermat\mat{Y}\centermat)$. The assumption $f(\mat{X}) = f(\mat{Y})$ implies $\centermat\mat{X}\centermat = c\,\centermat\mat{Y}\centermat$ where $c = t_X/t_Y \in \R^*$.

For symmetric $\mat{Z}$, let $\vec{\mu}_{\mat{Z}} = \mat{Z}\ones_{\numalt}/\numalt$ (row averages) and $m_{\mat{Z}} = \ones_{\numalt}^\top\mat{Z}\ones_{\numalt}/\numalt^2$ (grand mean). Expanding $\centermat\mat{Z}\centermat = \mat{Z} - \vec{\mu}_{\mat{Z}}\ones_{\numalt}^\top - \ones_{\numalt}\vec{\mu}_{\mat{Z}}^\top + m_{\mat{Z}}\ones_{\numalt}\ones_{\numalt}^\top$ and applying to both sides:
\[
\mat{X} = c\mat{Y} + (\vec{\mu}_{\mat{X}} - c\vec{\mu}_{\mat{Y}})\ones_{\numalt}^\top + \ones_{\numalt}(\vec{\mu}_{\mat{X}} - c\vec{\mu}_{\mat{Y}})^\top - (m_{\mat{X}} - cm_{\mat{Y}})\ones_{\numalt}\ones_{\numalt}^\top.
\]
Setting $\vec{w} = \vec{\mu}_{\mat{X}} - c\vec{\mu}_{\mat{Y}} - \frac{m_{\mat{X}} - cm_{\mat{Y}}}{2}\ones_{\numalt}$ gives the result.
\end{proof}

\vspace{6pt}
\begin{definition}[Centered, Trace-Normalized Subspace]\label{def:restricted-subspace}
Define $\mathbb{S}_{\numalt}^* = \{\mat{Y} \in \mathbb{S}_{\numalt} : \mat{Y}\ones_{\numalt} = \vec{0}, \, \tr(\mat{Y}) = \numalt\}$.
\end{definition}

\vspace{6pt}
\begin{definition}[Invariant Function $g$]\label{def:invariant}
For $\mat{Y} \in \mathbb{S}_{\numalt}^*$, define $g(\mat{Y}) = (\mathcal{D}_{\numalt}(\mat{Y}), \mathcal{O}_{\numalt}(\mat{Y}))$ where $\mathcal{D}_{\numalt}(\mat{Y}) = \{\!\{ Y_{ii} : i = 1, \ldots, \numalt \}\!\}$ is the multiset of diagonal entries and $\mathcal{O}_{\numalt}(\mat{Y}) = \{\!\{ Y_{ij} : 1 \leq i < j \leq \numalt \}\!\}$ is the multiset of off-diagonal entries.
\end{definition}

\subsection{Generic Injectivity}\label{sec:generic}

We now define the genericity conditions under which $g$ separates orbits.

\vspace{6pt}
\begin{definition}[Genericity Conditions]\label{def:genericity}
Let $\mat{Y} \in \mathbb{S}_{\numalt}$. We say $\mat{Y}$ satisfies:
\begin{enumerate}
    \item[(a)] \emph{Condition (a)} if all entries are distinct: $Y_{ij} \neq Y_{kl}$ for $(i,j) \neq (k,l)$ with $i \leq j$, $k \leq l$.
    \item[(b)] \emph{Condition (b)} if for each $i$, there exists exactly one $(\numalt - 1)$-subset $\mathcal{S}$ of off-diagonal positions such that $\sum_{(k,l) \in \mathcal{S}} Y_{kl} = -Y_{ii}$.
\end{enumerate}
\end{definition}

Condition (b) exploits the row-sum constraint $\mat{Y}\ones_{\numalt} = \vec{0}$: the off-diagonal entries in row $i$ sum to $-Y_{ii}$, so the canonical subset $\mathcal{S}_i^* = \{(i,j) : j \neq i\}$ always satisfies the equation. Condition (b) requires this subset to be unique among the $\binom{\numalt(\numalt - 1)/2}{\numalt-1}$ possible $(\numalt - 1)$-subsets.

\vspace{6pt}
\begin{lemma}[Recovery Lemma]\label{lem:recovery}
If $\mat{Y} \in \mathbb{S}_{\numalt}^*$ satisfies conditions (a) and (b), then $g(\mat{Y})$ determines $\mat{Y}$ up to $S_{\numalt}$-orbit.
\end{lemma}

\begin{proof}
Given multisets $\mathcal{D}_{\numalt}$ and $\mathcal{O}_{\numalt}$, we reconstruct $\mat{Y}$ up to relabeling:

\textbf{Step 1:} Order the distinct diagonal elements as $Y_{(1)(1)} > \cdots > Y_{(\numalt)(\numalt)}$, assigning canonical labels to alternatives.

\textbf{Step 2:} For each $(i)$, identify $\mathcal{R}_{(i)} \subset \mathcal{O}_{\numalt}$ as the unique $(\numalt - 1)$-subset summing to $-Y_{(i)(i)}$ (exists by condition (b)).

\textbf{Step 3:} For $i \neq j$, recover $Y_{(i)(j)}$ as the unique element in $\mathcal{R}_{(i)} \cap \mathcal{R}_{(j)}$ (singleton by condition (a)).

Any $\mat{Y}'$ with $g(\mat{Y}') = g(\mat{Y})$ satisfying (a) and (b) yields the same reconstruction up to relabeling.
\end{proof}

\subsection{The Exceptional Set}\label{sec:exceptional}

\vspace{6pt}
\begin{definition}[Exceptional Set]\label{def:exceptional}
Define $\exceptset = \mathcal{U}^C \cup \{\mat{X} \in \mathcal{U} : f(\mat{X}) \text{ violates (a) or (b)}\}$.
\end{definition}

\vspace{6pt}
\begin{lemma}[Non-Triviality of Constraints]\label{lem:nontriviality}
The following hold:
\begin{enumerate}
    \item[(i)] For distinct positions $(i,j) \neq (k,l)$, the function $f(\mat{X})_{ij} - f(\mat{X})_{kl}$ is not identically zero on $\mathcal{U}$.
    \item[(ii)] For each index $i$ and $(\numalt - 1)$-subset $\mathcal{S} \neq \mathcal{S}_i^*$, the function $\sum_{(k,l) \in \mathcal{S}} f(\mat{X})_{kl} + f(\mat{X})_{ii}$ is not identically zero on $\mathcal{U}$.
\end{enumerate}
\end{lemma}

\begin{proof}
Consider diagonal $\mat{X} = \mathrm{diag}(\lambda_1, \ldots, \lambda_{\numalt})$ with distinct $\lambda_j$ and $\sum_j \lambda_j \neq 0$. Then $[\centermat\mat{X}\centermat]_{ii} = \lambda_i(1 - 2/\numalt) + \bar{\lambda}/\numalt$ and $[\centermat\mat{X}\centermat]_{ij} = (\bar{\lambda} - \lambda_i - \lambda_j)/\numalt$ for $i \neq j$. For generic $\lambda_j$, all entries of $f(\mat{X})$ are distinct, proving (i). For (ii), the linear combination for $\mathcal{S} \neq \mathcal{S}_i^*$ has nonzero coefficient on $\lambda_i$, so it is nontrivial.
\end{proof}

\vspace{6pt}
\begin{proposition}[Properties of $\exceptset$]\label{prop:E-properties}
The set $\exceptset$ is closed, $G_{\numalt}$-invariant, and has Lebesgue measure zero.
\end{proposition}

\begin{proof}
\textbf{Closed:} Conditions (a) and (b) are open (strict inequalities and unique subset-sums are preserved under small perturbations), so $\exceptset^C$ is open.

\textbf{$G_{\numalt}$-invariant:} By Proposition~\ref{prop:f-and-G}, scaling and null-space translation leave $f(\mat{X})$ unchanged. Permutation conjugates $f(\mat{X})$, preserving distinctness and subset-sum structure.

\textbf{Measure zero:} Each entry $f(\mat{X})_{ij}$ is a rational (hence analytic) function on $\mathcal{U}$ by Proposition~\ref{prop:explicit-formula}. By Lemma~\ref{lem:nontriviality}, the constraints defining violations of (a) and (b) are nontrivial analytic functions. Their zero sets have measure zero by \citet{Mityagin2020}, and $\exceptset$ is a finite union of such sets.
\end{proof}

\subsection{Basic Result for the Full Space}\label{sec:basic}

\vspace{6pt}
\begin{theorem}[Generic Separation on $\mathbb{S}_{\numalt}$]\label{thm:main}
For any $\mat{X}, \mat{Y} \in \mathbb{S}_{\numalt} \setminus \exceptset$: $g\left\{f(\mat{X})\right\} = g\left\{f(\mat{Y})\right\}$ implies $\mat{X}$ and $\mat{Y}$ are in the same $G_{\numalt}$-orbit.
\end{theorem}

\begin{proof}
By Lemma~\ref{lem:recovery}, $g\left\{f(\mat{X})\right\} = g\left\{f(\mat{Y})\right\}$ implies $f(\mat{Y}) = \mat{P}_\pi f(\mat{X}) \mat{P}_\pi^\top$ for some $\pi \in S_{\numalt}$. By Proposition~\ref{prop:f-and-G}(i), $f(\mat{P}_\pi \mat{X} \mat{P}_\pi^\top) = f(\mat{Y})$. By Lemma~\ref{lem:preimage}, $\mat{Y} = c\,\mat{P}_\pi \mat{X} \mat{P}_\pi^\top + \vec{w}\ones_{\numalt}^\top + \ones_{\numalt}\vec{w}^\top$ for some $c \in \R^*$, $\vec{w} \in \R^{\numalt}$.
\end{proof}

\subsection{Restriction to Centered, Trace-Normalized Matrices}\label{sec:restriction}

\vspace{6pt}
\begin{proposition}[Properties of $\mathbb{S}_{\numalt}^*$]\label{prop:restricted-properties}
The subspace $\mathbb{S}_{\numalt}^*$ is a nonempty, closed affine subspace of dimension $(\numalt + 1)(\numalt - 2)/2$, satisfies $\mathbb{S}_{\numalt}^* = \mathrm{Im}(f)$, and is $S_{\numalt}$-invariant.
\end{proposition}

\begin{proof}
The dimension follows from $\dim(\mathbb{S}_{\numalt}) - (\numalt + 1) = \numalt(\numalt + 1)/2 - \numalt - 1 = (\numalt + 1)(\numalt - 2)/2$. For $\mathbb{S}_{\numalt}^* = \mathrm{Im}(f)$: if $\mat{X} \in \mathcal{U}$, then $f(\mat{X})\ones_{\numalt} = \vec{0}$ and $\tr(f(\mat{X})) = \numalt$, so $\mathrm{Im}(f) \subseteq \mathbb{S}_{\numalt}^*$; conversely, if $\mat{Y} \in \mathbb{S}_{\numalt}^*$, then $\centermat\mat{Y}\centermat = \mat{Y}$ and $f(\mat{Y}) = \mat{Y}$. For $S_{\numalt}$-invariance: $(\mat{P}_\pi\mat{Y}\mat{P}_\pi^\top)\ones_{\numalt} = \mat{P}_\pi\mat{Y}\ones_{\numalt} = \vec{0}$ and the trace is preserved.
\end{proof}
\vspace{12pt}

\begin{definition}[Restricted Exceptional Set]\label{def:restricted-exceptional}
Define $\exceptset^* = \{\mat{Y} \in \mathbb{S}_{\numalt}^* : \mat{Y} \text{ violates (a) or (b)}\}$.
\end{definition}
\vspace{6pt}

\begin{theorem}[Generic Separation on $\mathbb{S}_{\numalt}^*$]\label{thm:restricted-main}
The set $\exceptset^*$ is closed, $S_{\numalt}$-invariant, and has measure zero in $\mathbb{S}_{\numalt}^*$. For any $\mat{Y}_1, \mat{Y}_2 \in \mathbb{S}_{\numalt}^* \setminus \exceptset^*$: $g(\mat{Y}_1) = g(\mat{Y}_2)$ implies $\mat{Y}_1$ and $\mat{Y}_2$ are in the same $S_{\numalt}$-orbit.
\end{theorem}

\begin{proof}
The properties of $\exceptset^*$ follow from Proposition~\ref{prop:E-properties} restricted to $\mathbb{S}_{\numalt}^*$, using Lemma~\ref{lem:nontriviality} (which provides nontrivial constraints via $\mathbb{S}_{\numalt}^* = \mathrm{Im}(f)$). Separation follows from Lemma~\ref{lem:recovery}.
\end{proof}
\vspace{12pt}

\begin{corollary}[Universal Approximation on $\mathbb{S}_{\numalt}^*$]\label{cor:universal-approximation}
Any continuous, $S_{\numalt}$-invariant function on $\mathbb{S}_{\numalt}^* \setminus \exceptset^*$ can be uniformly approximated on compact subsets by an MLP taking the sorted diagonal and off-diagonal entries as input.
\end{corollary}

\begin{proof}
Following \citet{blumsmith2025machine}, encode $g$ via continuous functions: $g^d_k(\mat{Y}) =$ the $k$th largest diagonal entry, and $g^o_\ell(\mat{Y}) =$ the $\ell$th largest off-diagonal entry. These are $S_{\numalt}$-invariant and separate points in $(\mathbb{S}_{\numalt}^* \setminus \exceptset^*)/S_{\numalt}$ by Theorem~\ref{thm:restricted-main}. Since $S_{\numalt}$ is finite (hence compact), the quotient is Hausdorff \citep{bredon1972introduction}. By Stone--Weierstrass, the algebra generated by these functions is dense. By the universal approximation theorem \citep{cybenko1989approximation, hornik1989multilayer}, MLPs can approximate any continuous function of these inputs.
\end{proof}

\subsection{Restriction to Submatrices}\label{sec:submatrices}

For a distinguished index $j$, we develop an invariant separating $\permgroupsub$-orbits, where $\permgroupsub = \{\pi \in S_{\numalt} : \pi(j) = j\}$.
\vspace{20pt}

\begin{definition}[Submatrix Invariant $g_j$]\label{def:g-j}
For $\mat{Y} \in \mathbb{S}_{\numalt}^*$ and an index $j$, define the submatrix invariant as the pair of multisets $g_j(\mat{Y}) = (\mathcal{P}_j(\mat{Y}), \mathcal{O}_j(\mat{Y}))$, where
\begin{align*}
\mathcal{P}_j(\mat{Y}) &= \{\!\{ (Y_{kk}, Y_{jk}) : k \neq j \}\!\}, \\
\mathcal{O}_j(\mat{Y}) &= \{\!\{ Y_{kl} : 1 \leq k < l \leq \numalt, \, k \neq j, \, l \neq j \}\!\}.
\end{align*}
\end{definition}
\vspace{-9pt}
The multiset $\mathcal{P}_j$ couples each diagonal entry $Y_{kk}$ (for $k \neq j$) with the covariance $Y_{jk}$. The multiset $\mathcal{O}_j$ contains off-diagonal entries of $\mat{Y}_{-j,-j}$. Note that $Y_{jj} = \numalt - \sum_{k \neq j} Y_{kk}$ is implicitly recoverable.

\vspace{6pt}
\begin{definition}[Genericity Conditions for Submatrices]\label{def:genericity-submatrix}
We say $\mat{Y} \in \mathbb{S}_{\numalt}^*$ satisfies (a$_j$) if the pairs $(Y_{kk}, Y_{jk})$ for $k \neq j$ are pairwise distinct and the off-diagonal entries of $\mat{Y}_{-j,-j}$ are pairwise distinct, and (b$_j$) if for each $k \neq j$, there is a unique $(\numalt - 2)$-subset $\mathcal{S}$ of off-diagonal positions in $\mat{Y}_{-j,-j}$ with $\sum_{(l,m) \in \mathcal{S}} Y_{lm} = -Y_{kk} - Y_{jk}$.
\end{definition}

The target sum $-Y_{kk} - Y_{jk}$ arises from the row-sum constraint: $Y_{kk} + Y_{kj} + \sum_{l \neq k, l \neq j} Y_{kl} = 0$.

\vspace{6pt}
\begin{theorem}[Generic Separation via $g_j$]\label{thm:submatrix-main}
Define $\exceptset_j = \{\mat{Y} \in \mathbb{S}_{\numalt}^* : \mat{Y} \text{ violates } (\text{a}_j) \text{ or } (\text{b}_j)\}$. Then $\exceptset_j$ is closed, $\permgroupsub$-invariant, and has measure zero. For $\mat{Y}_1, \mat{Y}_2 \in \mathbb{S}_{\numalt}^* \setminus \exceptset_j$: $g_j(\mat{Y}_1) = g_j(\mat{Y}_2)$ implies $\mat{Y}_1, \mat{Y}_2$ are in the same $\permgroupsub$-orbit.
\end{theorem}

\begin{proof}
The proof follows Theorem~\ref{thm:restricted-main}. For separation, the recovery procedure parallels Lemma~\ref{lem:recovery}: order pairs in $\mathcal{P}_j$ lexicographically, identify each row's off-diagonal entries via the unique subset-sum, and recover shared entries as singleton intersections.
\end{proof}

\vspace{6pt}
\begin{corollary}[Universal Approximation via $g_j$]\label{cor:universal-approximation-submatrix}
Any continuous, $\permgroupsub$-invariant function on $\mathbb{S}_{\numalt}^* \setminus \exceptset_j$ can be uniformly approximated on compact subsets by an MLP taking the sorted pairs from $\mathcal{P}_j$ and sorted entries from $\mathcal{O}_j$ as input.
\end{corollary}

\begin{proof}
The proof is analogous to Corollary~\ref{cor:universal-approximation}, with pairs replacing diagonal entries.
\end{proof}

\subsection{Extension to Utility-Covariance Pairs: Global Invariants}\label{sec:utility-covariance-global}

We extend the theory to pairs $(\vec{v}, \mat{Y})$ consisting of a utility vector and covariance matrix.

\vspace{6pt}
\begin{definition}[Centered Utility-Covariance Space]\label{def:global-joint-space}
Define $\utilcovspace = \{(\vec{v}, \mat{Y}) \in \R^{\numalt} \times \mathbb{S}_{\numalt} : \sum_{i=1}^{\numalt} v_i = 0, \, \mat{Y} \ones_{\numalt} = \vec{0}, \, \tr(\mat{Y}) = \numalt\}$. This set differs from \eqref{eq:utilcovspace} in the main text by omitting the restriction $\covstar \succeq \mat{0}$; with a slight abuse of notation, we use the same symbol for both. Every result in this appendix for the larger set restricts to the smaller one, because the smaller set is a closed, $S_{\numalt}$-invariant subset with nonempty interior in the affine hull, and the exceptional sets remain closed and measure zero after restriction.\footnote{The text imposes $\covstar \succeq \mat{0}$ which is required for a valid covariance matrix. But the restriction $\covstar \succeq \mat{0}$ is not needed here as it does not provide a useful architectural simplification. It restricts the admissible input domain, but it does not create algebraic equalities that allow the network to ignore any particular part of $\covstar$.
The reason is that positive semidefiniteness is an inequality restriction ($\vec{a}^{\top}\covstar \vec{a} \geq 0$
for all $\vec{a}$) and does not impose additional
algebraic equalities among the entries of $\covstar$ beyond the centering and trace normalizations. 
Positive semidefiniteness does not imply, for example, that one row is determined by another row, or that some off-diagonal entries can be omitted from the architecture.}
\end{definition}

This space has dimension $(\numalt - 1) + (\numalt + 1)(\numalt - 2)/2 = (\numalt^2 + \numalt - 4)/2$. It is closed, Hausdorff, and $S_{\numalt}$-invariant under $\pi \cdot (\vec{v}, \mat{Y}) = (\mat{P}_\pi \vec{v}, \mat{P}_\pi \mat{Y} \mat{P}_\pi^\top)$.

\vspace{6pt}
\begin{definition}[Global Joint Invariant]\label{def:global-joint-g}
For $(\vec{v}, \mat{Y}) \in \utilcovspace$, define $\tilde{g}(\vec{v}, \mat{Y}) = (\mathcal{P}(\vec{v}, \mat{Y}), \mathcal{O}(\mat{Y}))$ where $\mathcal{P}(\vec{v}, \mat{Y}) = \{\!\{ (v_i, Y_{ii}) : i = 1, \ldots, \numalt \}\!\}$ and $\mathcal{O}(\mat{Y}) = \{\!\{ Y_{ij} : 1 \leq i < j \leq \numalt \}\!\}$.
\end{definition}

The key extension is that $\mathcal{P}$ contains \emph{pairs} $(v_i, Y_{ii})$ coupling utilities with variances.

\vspace{6pt}
\begin{definition}[Genericity Conditions]\label{def:global-genericity-joint}
We say $(\vec{v}, \mat{Y}) \in \utilcovspace$ satisfies ($\tilde{\text{a}}$) if the pairs $(v_i, Y_{ii})$ are pairwise distinct and the off-diagonal entries are pairwise distinct, and ($\tilde{\text{b}}$) if condition (b) holds for $\mat{Y}$.
\end{definition}

\vspace{6pt}
\begin{theorem}[Generic Separation on $\utilcovspace$]\label{thm:global-joint-main}
Define $\tilde{\exceptsetutil} = \{(\vec{v}, \mat{Y}) \in \utilcovspace : (\tilde{\text{a}}) \text{ or } (\tilde{\text{b}}) \text{ fails}\}$. Then $\tilde{\exceptsetutil}$ is closed, $S_{\numalt}$-invariant, and has measure zero. For $(\vec{v}_1, \mat{Y}_1), (\vec{v}_2, \mat{Y}_2) \in \utilcovspace \setminus \tilde{\exceptsetutil}$: $\tilde{g}(\vec{v}_1, \mat{Y}_1) = \tilde{g}(\vec{v}_2, \mat{Y}_2)$ implies the pairs are in the same $S_{\numalt}$-orbit.
\end{theorem}

\begin{proof}
The proof follows Theorem~\ref{thm:restricted-main}. The recovery procedure orders pairs lexicographically to assign labels, then recovers off-diagonal entries via subset-sum intersections. The measure-zero argument extends since $v_i - v_j$ is nontrivial on $\utilcovspace$.
\end{proof}

\subsection{Per-Alternative Invariants for Utility-Covariance Pairs}\label{sec:utility-covariance}

This section extends the symmetric-matrix invariants of \citet{blumsmith2025machine} to per-alternative representations of utility-covariance pairs, establishing Theorem~\ref{thm:separation} (Generic Separation) and Theorem~\ref{thm:universal} (MLP Universal Approximation) from the main text.

\vspace{6pt}
\begin{definition}[Per-Alternative Invariant $\invarutil_j$]\label{def:joint-g-j}
For $(\vec{v}, \mat{Y}) \in \utilcovspace$ and index $j$, define $\invarutil_j(\vec{v}, \mat{Y}) = (\mathcal{T}_j(\vec{v}, \mat{Y}), \mathcal{O}_j(\mat{Y}))$ where:
\begin{align*}
\mathcal{T}_j(\vec{v}, \mat{Y}) &= \{\!\{ (v_k, Y_{kk}, Y_{jk}) : k \neq j \}\!\}, \\
\mathcal{O}_j(\mat{Y}) &= \{\!\{ Y_{kl} : 1 \leq k < l \leq \numalt, \, k \neq j, \, l \neq j \}\!\}.
\end{align*}
\end{definition}

This combines extensions from Appendices~\ref{sec:submatrices} and~\ref{sec:utility-covariance-global}: $\mathcal{T}_j$ contains \emph{triples} $(v_k, Y_{kk}, Y_{jk})$ coupling each alternative $k$'s utility and variance with its covariance with $j$, while $\mathcal{O}_j$ contains off-diagonal entries of $\mat{Y}_{-j,-j}$. The values $v_j = -\sum_{k \neq j} v_k$ and $Y_{jj} = \numalt - \sum_{k \neq j} Y_{kk}$ are implicitly recoverable.

\vspace{6pt}
\begin{definition}[Genericity Conditions]\label{def:genericity-joint}
We say $(\vec{v}, \mat{Y}) \in \utilcovspace$ satisfies ($\tilde{\text{a}}_j$) if the triples $(v_k, Y_{kk}, Y_{jk})$ for $k \neq j$ are pairwise distinct and the off-diagonal entries of $\mat{Y}_{-j,-j}$ are pairwise distinct, and ($\tilde{\text{b}}_j$) if condition (b$_j$) holds.
\end{definition}

\vspace{6pt}
\begin{definition}[Exceptional Set]\label{def:exceptional-joint}
Define $\exceptsetutil_j = \{(\vec{v}, \mat{Y}) \in \utilcovspace : (\tilde{\text{a}}_j) \text{ or } (\tilde{\text{b}}_j) \text{ fails}\}$.
\end{definition}

\vspace{6pt}
\begin{proposition}[Properties of $\exceptsetutil_j$]\label{prop:joint-E-j-properties}
The set $\exceptsetutil_j$ is closed, $\permgroupsub$-invariant, and has measure zero in $\utilcovspace$.
\end{proposition}

\begin{proof}
The proof follows Proposition~\ref{prop:E-properties}. Closedness holds because ($\tilde{\text{a}}_j$) and ($\tilde{\text{b}}_j$) are open conditions. Invariance under $\permgroupsub$ holds because permutations fixing $j$ preserve distinctness and subset-sum structure. For measure zero: coincidence of two triples defines the intersection of three hyperplanes $\{v_k = v_l\} \cap \{Y_{kk} = Y_{ll}\} \cap \{Y_{jk} = Y_{jl}\}$, each nontrivial; the subset-sum constraints are similarly nontrivial by arguments analogous to Lemma~\ref{lem:nontriviality}(ii).
\end{proof}

\vspace{6pt}
\begin{lemma}[Recovery Lemma]\label{lem:recovery-joint}
If $(\vec{v}, \mat{Y}) \in \utilcovspace$ satisfies ($\tilde{\text{a}}_j$) and ($\tilde{\text{b}}_j$), then $\invarutil_j(\vec{v}, \mat{Y})$ determines $(\vec{v}, \mat{Y})$ up to $\permgroupsub$-orbit.
\end{lemma}

\begin{proof}
The proof extends Lemma~\ref{lem:recovery} to triples. Order triples in $\mathcal{T}_j$ lexicographically to assign labels $(1), \ldots, (\numalt - 1)$. For each $(k)$, identify $\mathcal{R}_{(k)} \subset \mathcal{O}_j$ as the unique $(\numalt - 2)$-subset summing to $-Y_{(k)(k)} - Y_{j(k)}$. Recover $Y_{(k)(l)}$ as the singleton $\mathcal{R}_{(k)} \cap \mathcal{R}_{(l)}$. Recover $v_j$ and $Y_{jj}$ from the constraints.
\end{proof}

We now prove the main results from Section~\ref{sec:theory}.

\begin{proof}[Proof of Theorem~\ref{thm:separation}]
By Proposition~\ref{prop:joint-E-j-properties}, $\exceptsetutil_j$ is closed, $\permgroupsub$-invariant, and has measure zero.
Note that these properties continue to hold when we add the restriction $\covstar \succeq \mat{0}$ as in \eqref{eq:utilcovspace}.
For $(\vec{v}_1, \mat{Y}_1), (\vec{v}_2, \mat{Y}_2) \in \utilcovspace \setminus \exceptsetutil_j$ with $\invarutil_j(\vec{v}_1, \mat{Y}_1) = \invarutil_j(\vec{v}_2, \mat{Y}_2)$, Lemma~\ref{lem:recovery-joint} implies both pairs are determined up to $\permgroupsub$-orbit by this common value.
\end{proof}

Before proving Theorem~\ref{thm:universal}, we establish the continuity of the choice probabilities on the relevant domain. Recall from Section~\ref{sec:universal-approximation} that $\exceptsetutil_j^{0} = \{(\vec{v}, \mat{Y}) \in \utilcovspace : Y_{jj} + Y_{kk} - 2Y_{jk} = 0 \text{ for some } k \neq j\}$ and $\exceptsetutilex_j = (\exceptsetutil_j \cup \exceptsetutil_j^{0}) \times (0, \infty)$.

\vspace{6pt}
\begin{lemma}[Continuity of Choice Probabilities]\label{lem:continuity}
Under Assumption~\ref{assumption:epsilon-star}, the choice probability $\choiceprob_j$, viewed as a function of $(\Cstar, \utilstar, \covstar)$, is continuous on $\utilcovspaceex \setminus \exceptsetutilex_j$, where $\utilcovspace$ is restricted to $\covstar \succeq \mat{0}$ as in \eqref{eq:utilcovspace}. Moreover, $\exceptsetutil_j^{0}$ is closed, $\permgroupsub$-invariant, and has measure zero in $\utilcovspace$.
\end{lemma}

\begin{proof}
The properties of $\exceptsetutil_j^{0}$ follow as in Proposition~\ref{prop:joint-E-j-properties}: it is a finite union of zero sets of nontrivial linear functions, each of which is preserved by permutations fixing $j$.
For continuity, the map $(\Cstar, \utilstar, \covstar) \mapsto (\tilde{\util}, \cov)$ described in Section~\ref{sec:preprocessing} is continuous (for factorization-invariant models, one may simply take $\cov = \covstar$), and so is $\cov \mapsto \cov^{1/2}$ on positive semidefinite matrices. By \eqref{eq:choice-integral} and location invariance,
\[
\choiceprob_j = \int \prod_{k \neq j} 1\left\{\vec{a}_k^{\top} \vec{\epsilon}^* \leq \tilde{v}_j - \tilde{v}_k\right\} f_{\vec{\epsilon}^*}(\vec{\epsilon}^*) \, d\vec{\epsilon}^*,
\qquad
\vec{a}_k = \big(\cov^{1/2}\big)^{\top} (\vec{e}_k - \vec{e}_j),
\]
where $\vec{e}_k$ denotes the $k$th standard basis vector. Since $\|\vec{a}_k\|^2 = (\vec{e}_k - \vec{e}_j)^{\top} \cov (\vec{e}_k - \vec{e}_j)$, which equals $\Sigma_{jj}^* + \Sigma_{kk}^* - 2\Sigma_{jk}^*$ up to a positive factor ($1/\Cstar$ in the general case, and $1$ when $\cov = \covstar$), every $\vec{a}_k$ is nonzero outside $\exceptsetutilex_j$. Along any convergent sequence of inputs in this set, each indicator converges pointwise except on the hyperplane $\{\vec{a}_k^{\top} \vec{\epsilon}^* = \tilde{v}_j - \tilde{v}_k\}$ at the limit, which has Lebesgue measure zero because $\vec{a}_k \neq \vec{0}$. The integrand is bounded by the integrable function $f_{\vec{\epsilon}^*}$, so the dominated convergence theorem yields continuity.
\end{proof}

\begin{proof}[Proof of Theorem~\ref{thm:universal}]

If $j=1$, then recover $v_1$ and $Y_{11} = Y_{j1}$ from the constraints.
Otherwise, there is exactly one element of $\overline{\mathcal{T}}_j(\util, \mat{Y})$ with $1_{\{k=1\}}=1$.
In this case, identify this element to recover $v_1$, $Y_{11}$, and $Y_{j1}$.
Encode this information via continuous functions:
\[
g_{0,1}^{\text{diag}} = v_1,\qquad g_{0,2}^{\text{diag}} = Y_{11},\qquad g_{0,3}^{\text{diag}} = Y_{j1}.
\]
Next, encode the remaining information in $\overline{\mathcal{T}}_j(\util, \mat{Y})$ by sorting the remaining tuples lexicographically and setting
\[
g_{k,1}^{\text{diag}} = v_{(k)},\qquad g_{k,2}^{\text{diag}} = Y_{(k)(k)},\qquad g_{k,3}^{\text{diag}} = Y_{j(k)}
\]
for $k=1,\ldots,\numalt-1-1_{\{j\neq1\}}$.
Next, sort the values in $\mathcal{O}_j(\mat{Y})$ and encode them such that $g_{m}^{\text{off}} = Y_{k[m],l[m]}$, where $k[m]$ and $l[m]$ are the indices $k$ and $l$ of the $m$th element in this ordering.
This produces $(\numalt - 1)(\numalt - 2)/2$ such functions.
Finally, set $g_C = \Cstar$.
These functions are continuous and $\permgroupsubsub$-invariant, and they separate points on $(\utilcovspaceex \setminus \exceptsetutilex_j)/\permgroupsubsub$: if two inputs outside $\exceptsetutilex_j$ share the same function values, then they have the same $\Cstar$ and the same $\overline{\invarutil}_j$, so by Theorem~\ref{thm:separation} they lie in a common $\permgroupsub$-orbit; and because the triples are distinct and the triple carrying $1_{\{k=1\}} = 1$ is matched to itself, the connecting permutation also fixes alternative $1$.

Let $\mathcal{K}$ be a compact subset of $\utilcovspaceex \setminus \exceptsetutilex_j$. Replacing $\mathcal{K}$ by $\cup_{\pi \in \permgroupsubsub} \pi \mathcal{K}$, which is compact and still avoids $\exceptsetutilex_j$ by its invariance, we may assume $\mathcal{K}$ is $\permgroupsubsub$-invariant. Since $\permgroupsubsub$ is finite, $\mathcal{K}/\permgroupsubsub$ is compact and Hausdorff \citep{bredon1972introduction}. By the Stone--Weierstrass theorem, the algebra generated by the functions above (together with the constants) is dense in the continuous functions on $\mathcal{K}/\permgroupsubsub$. By the universal approximation theorem \citep{cybenko1989approximation, hornik1989multilayer}, MLPs can approximate any continuous function of these inputs uniformly on $\mathcal{K}$. By Lemma~\ref{lem:continuity}, $\choiceprob_j$ is continuous on $\mathcal{K}$; it is also $\permgroupsubsub$-invariant by permutation equivariance and the fact that the preprocessing commutes with permutations (Section~\ref{sec:preprocessing}). Thus, the result holds.
\end{proof}

\section{Estimation Theory}
\label{sec:estimation-proofs}

This appendix contains additional information on the estimation theory results in Sections \ref{sec:estimation-theory} and \ref{sec:misspecification}.
Appendix~\ref{sec:approx-discussion} provides additional discussion of the approximation conditions: Assumptions~\ref{ass:approx} and~\ref{ass:grad-approx}.
Appendix \ref{sec:normality-proof} provides a proof of Theorem~\ref{thm:normality}.
Appendix \ref{sec:regularity-misspec-appendix} contains the regularity conditions for Section \ref{sec:misspecification}.
Appendix \ref{sec:proof-sandwich} provides proofs of Theorems~\ref{thm:consistency-misspec} and~\ref{thm:sandwich}.

\subsection{Emulator Approximation Assumptions}
\label{sec:approx-discussion}

This section provides additional discussion of Assumptions~\ref{ass:approx} and~\ref{ass:grad-approx}, showing how they can be derived from more primitive conditions on the emulator and justifying these conditions through related theoretical results.

\textbf{Primitive assumptions.} We consider a sieve regime in which a sequence of emulators $\hat{P}_n$ is trained with increasing precision as the sample size $n$ grows. This is achieved by progressively (a) increasing the number of simulated training examples $s = s_n$ and (b) increasing the neural network complexity $c = c_n$ (e.g., the number of layers or hidden units). We impose the following primitive conditions:

\begin{enumerate}
    \item[(P1)] \emph{Compact covariate support}: The covariate space $\mathcal{C}$ is compact.
    \item[(P2)] \emph{Smooth utility and covariance mappings}: The mappings $(\mat{X}, \param) \mapsto \Cstar(\mat{X}, \param)$, $(\mat{X}, \param) \mapsto \utilstar(\mat{X}, \param)$, and $(\mat{X}, \param) \mapsto \covstar(\mat{X}, \param)$ are continuous on $\mathcal{C} \times \paramspace$ and continuously differentiable in $\param$, with Jacobians $\nabla_{\param}(\Cstar, \utilstar, \covstar)$ that are jointly continuous in $(\mat{X}, \param)$.
    \item[(P3)] \emph{Interior probability condition}: There exist a compact set $\utilcovspaceexsubset \subset \utilcovspaceex$ and $\underline{p} > 0$ such that $\inf_{(\Cstar, \utilstar, \covstar) \in \utilcovspaceexsubset} \min_j P_j(\Cstar, \utilstar, \covstar) \geq \underline{p}$ and the image of $\mathcal{C} \times \paramspace$ under the mapping $(\mat{X}, \param) \mapsto (\Cstar(\mat{X}, \param), \utilstar(\mat{X}, \param), \covstar(\mat{X}, \param))$ is contained in $\utilcovspaceexsubset$.
    \item[(P4)] \emph{Probability-scale approximation}: For some $\alpha \geq 0$, the emulator satisfies
    \[
    \delta_n \equiv \sup_{(\Cstar, \utilstar, \covstar) \in \utilcovspaceexsubset} \max_{j=1,\ldots,\numalt} |\hat{P}_{nj}(\Cstar, \utilstar, \covstar) - P_j(\Cstar, \utilstar, \covstar)| = o_p\left(n^{-\alpha/2} / \sqrt{\log(n)}\right).
    \]
    \item[(P5)] \emph{Gradient approximation}: The true log choice probabilities $\log P_j$ and the emulator log probabilities $\log \hat{P}_{nj}$ are continuously differentiable in $(\Cstar, \utilstar, \covstar)$ on a neighborhood of $\utilcovspaceexsubset$, and the emulator gradients satisfy
    \[
    \delta_n^{\nabla} \equiv \sup_{(\Cstar, \utilstar, \covstar) \in \utilcovspaceexsubset} \max_{j=1,\ldots,\numalt} \left\| \nabla_{\Cstar, \utilstar, \covstar} \log \hat{P}_{nj} - \nabla_{\Cstar, \utilstar, \covstar} \log P_j \right\| = o_p(1).
    \]
\end{enumerate}

The logarithmic factor in (P4) is the price of the uniformity in $\param$ required by Assumption~\ref{ass:approx}; see Lemma~\ref{lem:uniform-fluctuation} below.
Throughout this section, the emulator is trained on simulated data that are independent of the observed sample, so we may condition on the trained emulator and treat $\delta_n$ and $\delta_n^{\nabla}$ as fixed. We also take $\paramspace$ to be convex; otherwise, the arguments below apply on a compact convex set containing $\paramspace$ to which (P2) and (P5) extend.

We next show how these assumptions imply Assumptions \ref{ass:approx} and \ref{ass:grad-approx}. We then provide justification of assumptions (P4) and (P5) using related results in approximation theory.

\textbf{Bounds on scores.} By the chain rule,
\[
\nabla_{\param} \log \hat{P}_{nj}(\param) = \left\{\nabla_{\param} (\Cstar, \utilstar, \covstar)\right\}^\top \nabla_{(\Cstar, \utilstar, \covstar)} \log \hat{P}_{nj},
\]
and similarly for $\nabla_{\param} \log P_j(\param)$. By condition (P2), the Jacobian $\nabla_{\param} (\Cstar, \utilstar, \covstar)$ is continuous on the compact set $\mathcal{C} \times \paramspace$ and hence uniformly bounded: $\sup_{(\mat{X}, \param) \in \mathcal{C} \times \paramspace} \|\nabla_{\param} (\Cstar, \utilstar, \covstar)\| \leq \kappa_1$ for some constant $\kappa_1 < \infty$. Hence, by (P5),
\begin{equation}\label{eq:score-diff-bound}
\sup_{(\mat{X}, \param) \in \mathcal{C} \times \paramspace} \max_j \left\| \nabla_{\param} \log \hat{P}_{nj}(\param) - \nabla_{\param} \log P_j(\param) \right\| \leq \kappa_1 \delta_n^{\nabla} = o_p(1).
\end{equation}
The true scores are also uniformly bounded: (P3) ensures that all relevant inputs $(\Cstar, \utilstar, \covstar)$ lie in the compact set $\utilcovspaceexsubset$, on which $\nabla_{(\Cstar, \utilstar, \covstar)} \log P_j$ is continuous by (P5) and hence bounded. Let $\kappa_2 = \sup_{(\mat{X}, \param) \in \mathcal{C} \times \paramspace} \max_j \|\nabla_{\param} \log P_j(\param)\| < \infty$.

\textbf{From probability-scale errors to log-likelihood errors.} We now show that conditions (P1)--(P5) imply Assumption~\ref{ass:approx}. Define the observation-level log-likelihood errors $e_i(\param) = \log \hat{P}_{ny_i}(\param) - \log P_{y_i}(\param)$, where $y_i \in \{1, \ldots, \numalt\}$ denotes the observed choice, so that
\[
\loglikhat(\param) - \loglik(\param) = \frac{1}{n} \sum_{i=1}^n e_i(\param).
\]
Conditional on covariates and the trained emulator, the $e_i(\param)$ are independent. We decompose into bias and centered components:
\[
\loglikhat(\param) - \loglik(\param) = \underbrace{\frac{1}{n} \sum_{i=1}^n \E\{e_i(\param) \mid \mat{X}_i\}}_{\text{bias}} + \underbrace{\frac{1}{n} \sum_{i=1}^n \left[e_i(\param) - \E\{e_i(\param) \mid \mat{X}_i\}\right]}_{\text{centered term}}.
\]

Here and below, we restrict attention to the event $\{\delta_n \leq \underline{p}/2\}$, whose probability approaches one by (P4); on this event, $\hat{P}_{nj} \geq \underline{p}/2$ on $\utilcovspaceexsubset$. For observation $i$, write $\Delta_{nj}(\param) = \hat{P}_{nj}(\param) - P_j(\param)$, suppressing the dependence on $\mat{X}_i$ as in the preceding paragraph, so that $|\Delta_{nj}(\param)| \leq \delta_n$ and $\sum_j \Delta_{nj}(\param) = 0$ because both probability vectors sum to one. Because $y_i$ is drawn from the true model, $\Prob(y_i = j \mid \mat{X}_i) = P_j(\paramtrue)$, and hence
\[
\E\{e_i(\param) \mid \mat{X}_i\} = \sum_{j=1}^{\numalt} P_j(\paramtrue) \log\left\{1 + \frac{\Delta_{nj}(\param)}{P_j(\param)}\right\}.
\]
By (P3) and the definition of $\delta_n$, $|\Delta_{nj}(\param)/P_j(\param)| \leq \delta_n/\underline{p} \leq 1/2$ on the event above. A first-order Taylor expansion of $\log(1 + x)$ around $x = 0$, for which $|\log(1 + x) - x| \leq x^2$ when $|x| \leq 1/2$, therefore gives
\[
\E\{e_i(\param) \mid \mat{X}_i\} = \sum_{j=1}^{\numalt} P_j(\paramtrue)\, \frac{\Delta_{nj}(\param)}{P_j(\param)} + T_i(\param), \qquad |T_i(\param)| \leq \sum_{j=1}^{\numalt} P_j(\paramtrue) \frac{\Delta_{nj}(\param)^2}{P_j(\param)^2} \leq \frac{\delta_n^2}{\underline{p}^2},
\]
where the last inequality uses $\sum_j P_j(\paramtrue) = 1$. Subtracting $\sum_j \Delta_{nj}(\param) = 0$ from the first term,
\[
\E\{e_i(\param) \mid \mat{X}_i\} = \sum_{j=1}^{\numalt} \frac{P_j(\paramtrue)}{P_j(\param)}\, \Delta_{nj}(\param) - \sum_{j=1}^{\numalt} \Delta_{nj}(\param) + T_i(\param) = \sum_{j=1}^{\numalt} \left\{\frac{P_j(\paramtrue)}{P_j(\param)} - 1\right\} \Delta_{nj}(\param) + T_i(\param).
\]
The first-order term vanishes at $\param = \paramtrue$, where the bias is therefore of order $\delta_n^2$. Away from $\paramtrue$, it is controlled by the Lipschitz continuity of the choice probabilities in $\param$: since $\nabla_{\param} P_j = P_j \nabla_{\param} \log P_j$ and $P_j \leq 1$, the mean value theorem on the convex set $\paramspace$ gives $|P_j(\paramtrue) - P_j(\param)| \leq \kappa_2 \|\param - \paramtrue\|$, so that $|P_j(\paramtrue)/P_j(\param) - 1| \leq \kappa_2 \|\param - \paramtrue\|/\underline{p}$. Therefore,
\begin{equation}\label{eq:bias-bound}
\left|\E\{e_i(\param) \mid \mat{X}_i\}\right| \leq \frac{\numalt \kappa_2}{\underline{p}}\, \delta_n \|\param - \paramtrue\| + \frac{\delta_n^2}{\underline{p}^2},
\end{equation}
and this bound holds uniformly over $\param \in \paramspace$ and $\mat{X}_i \in \mathcal{C}$ because $\delta_n$ is a supremum over $\utilcovspaceexsubset$, which contains the image of $\mathcal{C} \times \paramspace$ by (P3). The bias is thus quadratic in $\delta_n$ at $\paramtrue$ but only linear in $\delta_n$ elsewhere, because the observed choices are drawn from $P_j(\paramtrue)$ rather than from $P_j(\param)$; this is the source of the weighting in Assumption~\ref{ass:approx}.

For the centered term, the mean value theorem gives
\[
|\log \hat{P}_{nj} - \log P_j| \leq \frac{|\hat{P}_{nj} - P_j|}{\min(\hat{P}_{nj}, P_j)} \leq \frac{2\delta_n}{\underline{p}},
\]
and hence $|e_i(\param)| \leq 2\delta_n / \underline{p}$. At a fixed $\param$, Chebyshev's inequality then bounds the centered term by $O_p\{\delta_n/(\underline{p}\sqrt{n})\}$. Assumption~\ref{ass:approx}, however, requires a bound that holds uniformly over $\param$, because Theorem~\ref{thm:normality} evaluates the log likelihoods at data-dependent points. The following lemma provides it at the cost of a logarithmic factor.

\vspace{6pt}
\begin{lemma}[Uniform Sampling Fluctuations]\label{lem:uniform-fluctuation}
Under (P1)--(P5),
\[
\sup_{\param \in \paramspace} \left| \frac{1}{n} \sum_{i=1}^n \left[e_i(\param) - \E\{e_i(\param) \mid \mat{X}_i\}\right] \right| = O_p\left(\frac{\delta_n}{\underline{p}} \sqrt{\frac{\log n}{n}}\right) + O_p\left(n^{-\alpha - 1}\right).
\]
\end{lemma}

\begin{proof}
Write $Z_n(\param)$ for the centered term. Condition on the covariates and the trained emulator. For fixed $\param$, $Z_n(\param)$ is an average of $n$ independent, mean-zero terms, each lying in an interval of length $4\delta_n/\underline{p}$, so Hoeffding's inequality gives
\[
\Prob\{|Z_n(\param)| > t\} \leq 2 \exp\left(-\frac{n t^2 \underline{p}^2}{8 \delta_n^2}\right) \quad \text{for all } t > 0.
\]
Let $D$ denote the diameter of $\paramspace \subset \R^{\dimparam}$. Cover the compact set $\paramspace$ by $N \leq (1 + 2D/\eta)^{\dimparam}$ balls of radius $\eta = n^{-\alpha-1}$ with centers $\param_1, \ldots, \param_N \in \paramspace$, so that $\log N = O(\log n)$. A union bound over the centers yields, for any $u > 0$,
\[
\Prob\left\{ \max_{1 \leq r \leq N} |Z_n(\param_r)| > \frac{2\delta_n}{\underline{p}} \sqrt{\frac{2(\log N + u)}{n}} \right\} \leq 2 e^{-u},
\]
so $\max_r |Z_n(\param_r)| = O_p\{(\delta_n/\underline{p}) \sqrt{\log n / n}\}$.

It remains to control $Z_n$ between the centers. Each map $\param \mapsto e_i(\param) - \E\{e_i(\param) \mid \mat{X}_i\}$ is continuously differentiable on the convex set $\paramspace$ by (P2) and (P5), so by the mean value theorem it is Lipschitz with a constant $L_n$ that does not depend on $i$: writing $e_i^{(y)}(\param) = \log \hat{P}_{ny}(\param) - \log P_y(\param)$ and $\E\{e_i(\param) \mid \mat{X}_i\} = \sum_y P_y(\paramtrue)\, e_i^{(y)}(\param)$, where the weights $P_y(\paramtrue)$ do not depend on $\param$, the gradient of the centered map is $\nabla_{\param} e_i^{(y_i)}(\param) - \sum_y P_y(\paramtrue) \nabla_{\param} e_i^{(y)}(\param)$, and each score difference $\nabla_{\param} e_i^{(y)}$ is bounded by $\kappa_1 \delta_n^{\nabla}$ by \eqref{eq:score-diff-bound}. Hence $L_n \leq 2\kappa_1 \delta_n^{\nabla} = O_p(1)$. Since every $\param \in \paramspace$ lies within $\eta$ of some center,
\[
\sup_{\param \in \paramspace} |Z_n(\param)| \leq \max_{1 \leq r \leq N} |Z_n(\param_r)| + L_n \eta = O_p\left(\frac{\delta_n}{\underline{p}} \sqrt{\frac{\log n}{n}}\right) + O_p\left(n^{-\alpha-1}\right). \qedhere
\]
\end{proof}

Combining the bias bound \eqref{eq:bias-bound} with Lemma~\ref{lem:uniform-fluctuation} via the triangle inequality, with probability approaching one,
\[
|\loglikhat(\param) - \loglik(\param)| \leq \frac{\numalt \kappa_2}{\underline{p}}\, \delta_n \|\param - \paramtrue\| + \frac{\delta_n^2}{\underline{p}^2} + \sup_{\param' \in \paramspace} |Z_n(\param')| \quad \text{for all } \param \in \paramspace.
\]
By (P4), $\delta_n = o_p(n^{-\alpha/2})$, so $\delta_n^2 = o_p(n^{-\alpha})$; by Lemma~\ref{lem:uniform-fluctuation}, $\sup_{\param} |Z_n(\param)| = o_p\{n^{-(\alpha+1)/2}\} + O_p(n^{-\alpha-1})$, which is $o_p(n^{-\alpha})$ for $\alpha \in [0, 1]$. Hence
\[
\sup_{\param \in \paramspace} \frac{|\loglikhat(\param) - \loglik(\param)|}{n^{-\alpha} + n^{-\alpha/2} \|\param - \paramtrue\|} \leq n^{\alpha} \left\{\frac{\delta_n^2}{\underline{p}^2} + \sup_{\param \in \paramspace} |Z_n(\param)|\right\} + n^{\alpha/2}\, \frac{\numalt \kappa_2}{\underline{p}}\, \delta_n = o_p(1),
\]
which is Assumption~\ref{ass:approx}.
The weighting in Assumption~\ref{ass:approx} is essential to this conclusion: because the first-order term in \eqref{eq:bias-bound} does not vanish away from $\paramtrue$, the uniform condition $\sup_{\param \in \paramspace} |\loglikhat(\param) - \loglik(\param)| = o_p(n^{-\alpha})$ would instead require $\delta_n = o_p(n^{-\alpha})$, which by the discussion below corresponds to a training set growing faster than $n^{2\alpha}$ rather than $n^{\alpha}$.
For $\alpha = 1$, the argument above therefore delivers Assumption~\ref{ass:approx} under (P4) with $\delta_n = o_p(n^{-1/2} / \sqrt{\log(n)})$. 
\citet{kristensen2017higher} require a similar rate to demonstrate first-order equivalence between estimators based on approximate and exact objective functions: in their expansion, a deterministic approximation bias enters the estimator linearly (their Theorem~2), so equivalence requires $\sqrt{n}$ times the bias to vanish.
Their score-based argument avoids the logarithmic factor, which here is the price of uniformity in $\param$ (Lemma~\ref{lem:uniform-fluctuation}), but requires the Hessian of the approximate objective to converge to the true Hessian (their Assumption~A.2).
In contrast, the argument above requires only the first-derivative condition (P5) at the cost of a logarithmic factor for $\delta_n$.

\textbf{From gradient approximation to Fisher information consistency.} We now show that conditions (P1), (P2), (P3), and (P5) imply Assumption~\ref{ass:grad-approx}.
Recall the emulator and true outer-product-of-scores estimators are
\begin{align*}
\hat{\fisher}_n(\param) &= \frac{1}{n} \sum_{i=1}^n \nabla_{\param} \log \hat{P}_{ny_i}(\param) \nabla_{\param} \log \hat{P}_{ny_i}(\param)^\top, \\
\fisher_n(\param) &= \frac{1}{n} \sum_{i=1}^n \nabla_{\param} \log P_{y_i}(\param) \nabla_{\param} \log P_{y_i}(\param)^\top.
\end{align*}
Write $\scorehat_i = \scorehatn(y_i, \mat{X}_i; \param) = \nabla_{\param} \log \hat{P}_{ny_i}(\param)$ and $\score_i = \score(y_i, \mat{X}_i; \param) = \nabla_{\param} \log P_{y_i}(\param)$ for short. The identity
\[
\scorehat_i \scorehat_i^\top - \score_i \score_i^\top = (\scorehat_i - \score_i) \scorehat_i^\top + \score_i (\scorehat_i - \score_i)^\top
\]
implies
\[
\|\scorehat_i \scorehat_i^\top - \score_i \score_i^\top\| \leq \|\scorehat_i - \score_i\| \cdot (\|\scorehat_i\| + \|\score_i\|).
\]

By the definition of $\kappa_2$, $\|\score_i\| \leq \kappa_2$ for all $i$ and $\param$. Since $\|\scorehat_i - \score_i\| \leq \kappa_1 \delta_n^{\nabla} = o_p(1)$ by \eqref{eq:score-diff-bound}, the emulator scores satisfy $\|\scorehat_i\| \leq \|\score_i\| + \kappa_1 \delta_n^{\nabla} \leq \kappa_2 + o_p(1)$. Therefore,
\begin{align*}
\sup_{\param \in \paramspace} \|\hat{\fisher}_n(\param) - \fisher_n(\param)\| &\leq \sup_{\param \in \paramspace} \frac{1}{n} \sum_{i=1}^n \|\scorehat_i - \score_i\| \cdot (\|\scorehat_i\| + \|\score_i\|) \\
&\leq \kappa_1 \delta_n^{\nabla} \cdot \left\{2\kappa_2 + o_p(1)\right\} = o_p(1),
\end{align*}
which is Assumption~\ref{ass:grad-approx}.

\textbf{Justification from approximation theory.} The primitive conditions (P4) and (P5) are supported by related theoretical results on neural network approximation and estimation, though a complete theory tailored to our specific setting remains an area for future work.

The total emulator error can be decomposed into approximation error and estimation error:
\[
\delta_n \leq \underbrace{\operatorname{app}(c_n)}_{\text{approximation}} + \underbrace{\operatorname{est}(s_n, c_n)}_{\text{estimation}},
\]
where $\operatorname{app}(c)$ is the distance from the true choice probability function to the best approximator in the network class of complexity $c$, and $\operatorname{est}(s, c)$ is the deviation of the trained network from this best-in-class approximator due to training on $s$ simulated examples.

For the approximation error, provided the choice probabilities are real-analytic functions of $(\Cstar, \utilstar, \covstar)$ on the interior of the input space, neural networks with smooth activations can approximate them with exponentially small error. \citet{de2021approximation} establish that tanh networks with two hidden layers achieve approximation error decaying as $O\left\{N^{k/(d_{\mathcal{X}}+1)} \exp(-C \cdot N^{1/(d_{\mathcal{X}}+1)} \log N)\right\}$ in Sobolev norms $W^{k,\infty}$, where $N$ is a network-size parameter in which both hidden-layer widths grow polynomially, $d_{\mathcal{X}}$ is the input dimension, and $C > 0$ is a constant whose admissible values depend on the analyticity parameters of the target function. Earlier work established similar exponential rates for deep ReLU networks \citep{e2018exponential}, though under the stronger assumption that the target function admits an absolutely convergent power series expansion. Given these exponential convergence results, achieving approximation error $\varepsilon$ requires network complexity scaling only polylogarithmically in $1/\varepsilon$.

For the estimation error, standard results from empirical process theory yield bounds of order $O\left\{\sqrt{c \log(sc)/s}\right\}$ in $L^2$ norms, where $c$ is a complexity parameter \citep{schmidthieber2020, farrell2021, shen2023asymptotic}. For our setting, however, we require an $L^{\infty}$ bound. The results below assume minimally that $\operatorname{est}(s, c) = O\left\{c^{b}\, \polylog(s)/\sqrt{s}\right\}$ for some $b < \infty$. Then, with network complexity growing polylogarithmically in the target accuracy, the dependence on $c$ is negligible, and the estimation error is effectively of order $O\left(1/\sqrt{s}\right)$ up to logarithmic factors.

Achieving the rate given in (P4) requires that both $\operatorname{app}(c_n) = o\{n^{-\alpha/2}(\log n)^{-1/2}\}$ and $\operatorname{est}(s_n, c_n) = o_p\{n^{-\alpha/2}(\log n)^{-1/2}\}$; note that the approximation error is nonstochastic because it concerns pure function approximation. The exponential approximation rates for analytic functions imply that the first condition is satisfied with polylogarithmic growth in $c_n$. The second condition requires $s_n$ to grow faster than $n^\alpha$ up to logarithmic factors. In practice, since simulation from the MNP model is computationally inexpensive, generating a training set with $s_n \gg n$ is feasible.

The cases $\alpha = 0$ and $\alpha = 1$ correspond to different inferential goals. When $\alpha = 0$, Assumption~\ref{ass:approx} requires only that the log-likelihood error vanish in probability, which suffices for consistency of the emulator-based estimator (Theorem~\ref{thm:consistency}). When $\alpha = 1$, the stronger requirement that the error vanish faster than $n^{-1}$ near $\paramtrue$ ensures that the emulator-based estimator is asymptotically equivalent to the true MLE, inheriting its $\sqrt{n}$-consistency and asymptotic normality (Theorem~\ref{thm:normality}). The latter requires a larger training set: absorbing the logarithmic factors, it suffices for $s_n$ to grow as $n^{\alpha + \delta}$ for any $\delta > 0$, that is, as $n^{1+\delta}$ when $\alpha = 1$ versus $n^{\delta}$ when $\alpha = 0$.

For the gradient approximation in (P5), \citet{czarnecki2017sobolev} show that neural networks are universal approximators in Sobolev norms, and the bounds of \citet{de2021approximation} hold in $W^{k,\infty}$ norms, establishing that smooth networks can simultaneously approximate both function values and derivatives. Additionally, \citet{cocola2020global} prove global convergence of gradient flow for Sobolev training with overparameterized two-layer networks, ensuring that the trained network achieves small Sobolev loss under appropriate conditions. These results support the plausibility of (P5), though the precise rates and conditions for our setting merit further investigation.
\subsection{Proof of Theorem \ref{thm:normality}}
\label{sec:normality-proof}

\begin{proof}
We prove the result for approximate maximizers satisfying $\loglikhat(\paramest) \geq \sup_{\param \in \paramspace} \loglikhat(\param) - \eta_n$ where $\eta_n = o_p(n^{-1})$. Exact maximizers correspond to $\eta_n = 0$. We proceed in three steps: first bounding the difference in log likelihoods, then showing that the estimators are asymptotically equivalent, and finally appealing to the asymptotic normality of the exact MLE.

\textbf{Step 1: Bounding the difference in log likelihoods.}

Let $\varepsilon_n$ denote the supremum in Assumption~\ref{ass:approx}, so that $\varepsilon_n = o_p(1)$ and $|\loglikhat(\param) - \loglik(\param)| \leq \varepsilon_n \left(n^{-\alpha} + n^{-\alpha/2} \|\param - \paramtrue\|\right)$ for all $\param \in \paramspace$. Since $\parammle$ maximizes $\loglik$ and $\paramest$ is an approximate maximizer of $\loglikhat$,
\begin{align*}
0 \leq \loglik(\parammle) - \loglik(\paramest) &= \left\{\loglik(\parammle) - \loglikhat(\parammle)\right\} + \left\{\loglikhat(\parammle) - \loglikhat(\paramest)\right\} + \left\{\loglikhat(\paramest) - \loglik(\paramest)\right\} \\
&\leq |\loglikhat(\parammle) - \loglik(\parammle)| + \eta_n + |\loglikhat(\paramest) - \loglik(\paramest)| \\
&\leq \varepsilon_n \left\{2 n^{-\alpha} + n^{-\alpha/2} \left(\|\parammle - \paramtrue\| + \|\paramest - \paramtrue\|\right)\right\} + \eta_n.
\end{align*}

\textbf{Step 2: Asymptotic equivalence of the estimators.}

We show that $\sqrt{n}(\paramest - \parammle) \convp \vec{0}$. By Theorem~\ref{thm:consistency} and standard MLE theory, $\paramest \convp \paramtrue$ and $\parammle \convp \paramtrue$, where $\paramtrue \in \mathrm{int}(\paramspace)$ by Assumption~\ref{ass:regularity}(i). Hence, with probability approaching one, $\parammle$ is an interior point and the score vanishes there: $\gradparam \loglik(\parammle) = \vec{0}$.

By Assumption~\ref{ass:regularity}(ii) and (iv), the uniform law of large numbers gives $\sup_{\param \in \paramspace} \|\gradparam^2 \loglik(\param) - \E\{\gradparam^2 \loglikobs(y, \mat{X}; \param)\}\| \convp 0$, and $\param \mapsto \E\{\gradparam^2 \loglikobs(y, \mat{X}; \param)\}$ is continuous by dominated convergence. Under correct specification, the information matrix equality gives $\E\{\gradparam^2 \loglikobs(y, \mat{X}; \paramtrue)\} = -\fisher(\paramtrue)$; for discrete outcomes, the required interchange of differentiation and summation over $y$ is justified by Assumption~\ref{ass:regularity}(ii)--(iv). Let $\lambda_{\min} > 0$ denote the minimum eigenvalue of $\fisher(\paramtrue)$, which is positive by Assumption~\ref{ass:regularity}(v). By continuity, there exists a closed ball $N \subset \mathrm{int}(\paramspace)$ centered at $\paramtrue$ on which all eigenvalues of $\E\{\gradparam^2 \loglikobs(y, \mat{X}; \param)\}$ are at most $-3\lambda_{\min}/4$. Combining this with the uniform convergence above, the event
\[
A_n = \left\{\paramest \in N, \ \parammle \in N, \ \text{and } \gradparam^2 \loglik(\param) \preceq -\lambda_{\min} \mat{I} /2\text{ for all } \param \in N\right\}
\]
has probability approaching one.

On $A_n$, the exact second-order Taylor expansion of $\loglik$ around $\parammle$, evaluated at $\paramest$, gives
\[
\loglik(\paramest) = \loglik(\parammle) + \frac{1}{2}(\paramest - \parammle)^\top \gradparam^2 \loglik(\bar{\param}_n)(\paramest - \parammle),
\]
where $\bar{\param}_n$ lies on the segment between $\paramest$ and $\parammle$ and therefore in the convex set $N$. Hence, on $A_n$,
\[
\loglik(\parammle) - \loglik(\paramest) \geq \frac{\lambda_{\min}}{4}\, \|\paramest - \parammle\|^2.
\]
Combining the display above with Step 1 and using $\|\paramest - \paramtrue\| \leq \|\paramest - \parammle\| + \|\parammle - \paramtrue\|$, on $A_n$,
\[
\frac{\lambda_{\min}}{4}\, \|\paramest - \parammle\|^2 \leq \varepsilon_n \left\{2 n^{-\alpha} + 2 n^{-\alpha/2} \|\parammle - \paramtrue\| + n^{-\alpha/2} \|\paramest - \parammle\|\right\} + \eta_n.
\]
Multiplying by $n$ and using $n^{-\alpha} \leq n^{-1}$ and $n^{-\alpha/2} \leq n^{-1/2}$ for $\alpha \geq 1$,
\[
\frac{\lambda_{\min}}{4}\, n\|\paramest - \parammle\|^2 \leq 2\varepsilon_n + 2\varepsilon_n \sqrt{n}\|\parammle - \paramtrue\| + \varepsilon_n \sqrt{n}\|\paramest - \parammle\| + n \eta_n.
\]
Write $u_n = \sqrt{n}\|\paramest - \parammle\|$. Since $\varepsilon_n = o_p(1)$, $n \eta_n = o_p(1)$, and $\sqrt{n}\|\parammle - \paramtrue\| = O_p(1)$ by the asymptotic normality of the exact MLE (Step 3), the last display has the form
\[
\frac{\lambda_{\min}}{4}\, u_n^2 \leq a + b u_n, \qquad a = 2\varepsilon_n \left(1 + \sqrt{n}\|\parammle - \paramtrue\|\right) + n \eta_n = o_p(1), \qquad b = \varepsilon_n = o_p(1).
\]
An inequality of the form $c u^2 \leq a + b u$ with $c > 0$ and $a, b \geq 0$ implies $u \leq b/c + (a/c)^{1/2}$, so $u_n = o_p(1)$ on $A_n$. Since $\Prob(A_n) \to 1$, it follows that $\sqrt{n}(\paramest - \parammle) \convp \vec{0}$.
\vspace{12pt}

\textbf{Step 3: Asymptotic normality.}

By standard maximum likelihood theory (Theorem 5.39 and Lemma 7.6 of \citet{vaart1998asymptotic}), the regularity conditions in Assumption~\ref{ass:regularity} ensure that the model is differentiable in quadratic mean at $\paramtrue$ with nonsingular Fisher information. Since $\parammle$ is consistent,
$\sqrt{n}(\parammle - \paramtrue) \convd \mathcal{N}\left\{\vec{0}, \fisher(\paramtrue)^{-1}\right\}.$
Since $\sqrt{n}(\paramest - \parammle) \convp \vec{0}$, Slutsky's theorem gives
\[
\sqrt{n}(\paramest - \paramtrue) = \sqrt{n}(\paramest - \parammle) + \sqrt{n}(\parammle - \paramtrue) \convd \mathcal{N}\left\{\vec{0}, \fisher(\paramtrue)^{-1}\right\}. \qedhere
\]
\end{proof}

\subsection{Regularity Conditions for Inference Under Misspecification}
\label{sec:regularity-misspec-appendix}

This section states the regularity conditions required for the results in Section~\ref{sec:misspecification}. These conditions parallel standard assumptions for M-estimators \citep[Chapter 5]{vaart1998asymptotic} and are mild for neural network emulators with smooth activation functions. Recall that in Section~\ref{sec:misspecification}, we treat the neural network as fixed (not varying with $n$), so we use $\loglikobshat$ and $\scorehat$ without $n$ subscripts.

\begin{assumption}[Identification and Uniform Convergence]
\label{ass:pseudo-identification}
The following conditions hold:
\begin{enumerate}
    \item[(i)] The parameter space $\paramspace$ is compact.
    \item[(ii)] The pseudo-true parameter $\parampseudo = \argmax_{\param \in \paramspace} \E\left\{\loglikobshat(y, \mat{X}; {\param})\right\}$ exists uniquely and lies in the interior of $\paramspace$.
    \item[(iii)] The pseudo-true parameter is well-separated: for every $\epsilon > 0$,
    \[
    \sup_{\param: \|\param - \parampseudo\| \geq \epsilon} \E\left\{\loglikobshat(y, \mat{X}; {\param})\right\} < \E\left\{\loglikobshat(y, \mat{X}; \parampseudo)\right\}.
    \]
    \item[(iv)] The sample criterion function converges uniformly to its expectation:
    \[
    \sup_{\param \in \paramspace} \left| \frac{1}{n} \sum_{i=1}^n \loglikobshat(y_i, \mat{X}_i; \param) - \E[\loglikobshat(y, \mat{X}; {\param})] \right| \convp 0.
    \]
    \item[(v)] The estimator $\paramest$ is an approximate maximizer of the sample criterion function:
    \[
    \frac{1}{n} \sum_{i=1}^n \loglikobshat(y_i, \mat{X}_i; \paramest) \geq \sup_{\param \in \paramspace} \frac{1}{n} \sum_{i=1}^n \loglikobshat(y_i, \mat{X}_i; \param) - o_p(n^{-1}).
    \]
\end{enumerate}
\end{assumption}

\begin{assumption}[Regularity for Asymptotic Normality]
\label{ass:regularity-misspec}
The following conditions hold:
\begin{enumerate}
    \item[(i)] (Differentiability) For each $(y, \mat{X})$, the map $\param \mapsto \loglikobshat(y, \mat{X}; \param)$ is differentiable on $\paramspace$ with gradient $\scorehat(y, \mat{X}; \param) = \nabla_{\param} \loglikobshat(y, \mat{X}; \param)$.
    \item[(ii)] (Local dominance) There exists $\delta > 0$ and a measurable function $\dot{m}(y, \mat{X})$ having $\E[\dot{m}(y, \mat{X})^2] < \infty$ such that
    \[
    \sup_{\param: \|\param - \parampseudo\| \leq \delta} \|\scorehat(y, \mat{X}; \param)\| \leq \dot{m}(y, \mat{X}).
    \]
    \item[(iii)] (Second-order expansion) The map $\param \mapsto \E\left\{\loglikobshat(y, \mat{X}; \param)\right\}$ admits a second-order Taylor expansion at $\parampseudo$:
    \[
    \E\left\{\loglikobshat(y, \mat{X}; \param)\right\} = \E\left\{\loglikobshat(y, \mat{X}; \parampseudo)\right\} - \frac{1}{2}(\param - \parampseudo)^\top \Amat(\parampseudo) (\param - \parampseudo) + o(\|\param - \parampseudo\|^2),
    \]
    where $\Amat(\parampseudo) = -\E\left\{\nabla_{\param}^2 \loglikobshat(y, \mat{X}; \parampseudo)\right\}$ is positive definite.
\end{enumerate}
\end{assumption}

\subsection{Proof of Theorems~\ref{thm:consistency-misspec} and~\ref{thm:sandwich}}
\label{sec:proof-sandwich}

\begin{proof}[Proof of Theorem~\ref{thm:consistency-misspec}]
Define the sample and population criterion functions
\[
M_n(\param) = \frac{1}{n} \sum_{i=1}^n \loglikobshat(y_i, \mat{X}_i; \param), \qquad M(\param) = \E\left\{\loglikobshat(y, \mat{X}; \param)\right\}.
\]
By Assumption~\ref{ass:pseudo-identification}(iv), $\sup_{\param \in \paramspace} |M_n(\param) - M(\param)| \convp 0$. By Assumption~\ref{ass:pseudo-identification}(ii)--(iii), $\parampseudo$ is the unique maximizer of $M(\param)$ and is well-separated in the sense that $\sup_{\param: \|\param - \parampseudo\| \geq \epsilon} M(\param) < M(\parampseudo)$ for every $\epsilon > 0$. By Assumption~\ref{ass:pseudo-identification}(v), $M_n(\paramest) \geq \sup_{\param \in \paramspace} M_n(\param) - o_p(n^{-1})$, which implies $M_n(\paramest) \geq \sup_{\param \in \paramspace} M_n(\param) - o_p(1)$. The result $\paramest \convp \parampseudo$ now follows from Theorem 5.7 of \citet{vaart1998asymptotic}.
\end{proof}

\begin{proof}[Proof of Theorem~\ref{thm:sandwich}]
We verify the conditions of Theorem 5.23 of \citet{vaart1998asymptotic} and apply it to obtain the result.

By Assumption~\ref{ass:regularity-misspec}(i), the map $\param \mapsto \loglikobshat(y, \mat{X}; \param)$ is differentiable at $\parampseudo$ for every $(y, \mat{X})$, with derivative $\scorehat(y, \mat{X}; \parampseudo)$. By Assumption~\ref{ass:regularity-misspec}(ii), the score is locally dominated by a square-integrable function, which implies the Lipschitz condition
\[
|\loglikobshat(y, \mat{X}; \param_1) - \loglikobshat(y, \mat{X}; \param_2)| \leq \dot{m}(y, \mat{X}) \|\param_1 - \param_2\|
\]
for $\param_1, \param_2$ in a neighborhood of $\parampseudo$ via the mean value theorem. Assumption~\ref{ass:regularity-misspec}(iii) provides the required second-order Taylor expansion of $\param \mapsto \E\left\{\loglikobshat(y, \mat{X}; \param)\right\}$ at the point of maximum $\parampseudo$, with nonsingular second derivative matrix $-\Amat(\parampseudo)$. By Theorem~\ref{thm:consistency-misspec}, $\paramest \convp \parampseudo$. By Assumption~\ref{ass:pseudo-identification}(v), the near-maximization condition $M_n(\paramest) \geq \sup_{\param \in \paramspace} M_n(\param) - o_p(n^{-1})$ is satisfied.

Applying Theorem 5.23 of \citet{vaart1998asymptotic}:
\[
\sqrt{n}(\paramest - \parampseudo) = \Amat(\parampseudo)^{-1} \frac{1}{\sqrt{n}} \sum_{i=1}^n \scorehat(y_i, \mat{X}_i; \parampseudo) + o_p(1).
\]
By Assumption~\ref{ass:regularity-misspec}(ii), $\E\left\{\|\scorehat(y, \mat{X}; \parampseudo)\|^2\right\} \leq \E\left\{\dot{m}(y, \mat{X})^2\right\} < \infty$. The central limit theorem and Slutsky's lemma yield
\[
\sqrt{n}(\paramest - \parampseudo) \convd \mathcal{N}\left\{\vec{0}, \Amat(\parampseudo)^{-1} \Bmat(\parampseudo) \Amat(\parampseudo)^{-\top}\right\},
\]
where $\Bmat(\parampseudo) = \E\left\{\scorehat(y, \mat{X}; \parampseudo) \scorehat(y, \mat{X}; \parampseudo)^\top\right\}$.
\end{proof}

\section{Simulation Study Details}
\label{sec:simulation-details}

This appendix provides additional details on the emulator input features, the emulator training procedure, and the scaled Wishart distribution used in the simulation study.

\subsection{Input Feature Details}
\label{sec:input-feature-details}

This section documents the input features for each component of the per-alternative encoder described in Section~\ref{sec:architecture}. Table~\ref{tab:input-features} summarizes which features are passed to each component. We also describe transformations to reduce sensitivity to outliers.

\begin{table}[ptbh]
\centering
\caption{Per-Alternative Encoder Input Features}
\label{tab:input-features}
\begin{tabular}{lccc}
\toprule
Feature & Diagonal & Off-diagonal & Pass-through \\
\midrule
$v_j^*$ & \checkmark & & \checkmark \\
$v_k^*$ & \checkmark & & \\
$\Sigma_{jj}^*$ & \checkmark & & \checkmark \\
$\Sigma_{kk}^*$ & \checkmark & & \\
$\sigma_j$ & \checkmark & & \checkmark \\
$\sigma_k$ & \checkmark & & \\
$\Sigma_{jk}^*$ & \checkmark & & \\
$\rho_{jk}$ & \checkmark & & \\
$\ssign(z_{jk})$ & \checkmark & & \\
\midrule
$\Sigma_{kl}^*$ & & \checkmark & \\
$\rho_{kl}$ & & \checkmark & \\
$(v_k^* - v_l^*)^2$ & & \checkmark & \\
$\ssign(z_{kl})^2$ & & \checkmark & \\
$\Sigma_{kk}^* + \Sigma_{ll}^*$ & & \checkmark & \\
$\sigma_k + \sigma_l$ & & \checkmark & \\
\midrule
$\bar{\rho}_j$ & & & \checkmark \\
$\overline{\rho^2}_j$ & & & \checkmark \\
$\bar{\Sigma}_j$ & & & \checkmark \\
$\overline{\Sigma^2}_j$ & & & \checkmark \\
$\ssign(v_j^*/\sigma_j)$ & & & \checkmark \\ [5pt]
\midrule 
\multicolumn{4}{l}{\textit{Optional features}} \\[5pt]
$1_{\{k=1\}}$ & \checkmark & & \\
$\ssign(\log \Cstar)$ & & & \checkmark \\
$1_{\{j=1\}}$ & & & \checkmark \\
$\numalt$ & & & \checkmark \\
\bottomrule
\end{tabular}

\floatnote{Diagonal features are indexed by rival $k \neq j$; off-diagonal features by pairs $(k,l)$ with $k<l$ and $k,l\neq j$; and pass-through features by focal alternative $j$. The features $1_{\{k=1\}}$, $\ssign(\log \Cstar)$, and $1_{\{j=1\}}$ are used only for models that are not factorization-invariant. The feature $\numalt$ is used only when a single emulator is trained across several values of $\numalt$, namely for the multi-$\numalt$ emulator of Section~\ref{sec:sim-multi-k} and the application emulators of Appendix~\ref{sec:application-emulator-design}; the fixed-$\numalt$ emulators of Section~\ref{sec:simulation-results} omit it.}
\end{table}

\subsubsection*{Smooth Soft-Sign Function}

Several features are unbounded (e.g., z-scores, signal-to-noise ratios). We apply a smooth soft-sign function to map these to $(-1, 1)$:
\begin{equation}\label{eq:ssign}
\ssign(x) = \frac{0.1 x}{\sqrt{1 + (0.1 x)^2}}.
\end{equation}
This function is a $C^\infty$ bijection from $\R$ to $(-1,1)$, so applying it preserves the universal approximation property (Theorem~\ref{thm:universal}).

\subsubsection*{Feature Definitions}

Let $j$ denote the focal alternative and $k, l \neq j$ denote rival alternatives. We define:
\begin{align*}
\sigma_j &= \sqrt{\Sigma_{jj}^*}, &
\rho_{jk} &= \frac{\Sigma_{jk}^*}{\sigma_j \sigma_k}, &
z_{jk} &= \frac{v_j^* - v_k^*}{\sqrt{\Sigma_{jj}^* + \Sigma_{kk}^* - 2\Sigma_{jk}^*}}.
\end{align*}
In the implementation, we add a small constant ($10^{-8}$) to denominators before division and to arguments before taking square roots to ensure numerical stability.

\textbf{Diagonal DeepSet.} For each pair $(j, k)$ with $k \neq j$, the feature vector $\vec{d}_{jk}$ contains:
\begin{align*}
v_k^*, \quad \Sigma_{kk}^*, \quad \sigma_k, \quad \Sigma_{jk}^*, \quad \rho_{jk}, \quad v_j^*, \quad \Sigma_{jj}^*, \quad \sigma_j, \quad \ssign(z_{jk}).
\end{align*}

\textbf{Off-diagonal DeepSet.} For each pair $(k, l)$ with $k < l$ and $k, l \neq j$, the feature vector $\vec{o}_{kl}$ contains:
\begin{align*}
\Sigma_{kl}^*, \quad \rho_{kl}, \quad (v_k^* - v_l^*)^2, \quad \ssign(z_{kl})^2, \quad \Sigma_{kk}^* + \Sigma_{ll}^*, \quad \sigma_k + \sigma_l.
\end{align*}
These features are symmetric in $k$ and $l$, as required since the pairs have no natural ordering.

\textbf{Pass-through features.} For each alternative $j$, the feature vector $\vec{s}_j$ contains:
\begin{align*}
v_j^*, \quad \Sigma_{jj}^*, \quad \sigma_j, \quad \bar{\rho}_j, \quad \overline{\rho^2}_j, \quad \bar{\Sigma}_j, \quad \overline{\Sigma^2}_j, \quad \ssign(v_j^*/\sigma_j),
\end{align*}
where the summary statistics are
\begin{align*}
\bar{\rho}_j &= \frac{1}{\numalt-1}\sum_{k \neq j} \rho_{jk}, &
\overline{\rho^2}_j &= \frac{1}{\numalt-1}\sum_{k \neq j} \rho_{jk}^2, \\
\bar{\Sigma}_j &= \frac{1}{\numalt-1}\sum_{k \neq j} \Sigma_{jk}^*, &
\overline{\Sigma^2}_j &= \frac{1}{\numalt-1}\sum_{k \neq j} (\Sigma_{jk}^*)^2.
\end{align*}

\textbf{Optional features.} For models that are not factorization-invariant (Definition~\ref{def:factorization-invariant}), we include additional features to prevent information loss from the preprocessing transformation (Section~\ref{sec:preprocessing}): the normalizing constant $\ssign(\log \Cstar)$ and index indicators $1_{\{j=1\}}$ and $1_{\{k=1\}}$. When training a single emulator across different numbers of alternatives, as in Section~\ref{sec:sim-multi-k} and Appendix~\ref{sec:application-emulator-design}, we also include $\numalt$ as a pass-through feature.

\subsection{Emulator Training}
\label{sec:emulator-training}

We train the emulator using the Adam optimizer \citep{kingma2015adam} with the AdamW weight decay variant \citep{loshchilov2019decoupled}. The base learning rate is $0.003$ with momentum parameters $\beta_1 = 0.9$ and $\beta_2 = 0.999$, numerical stability constant $10^{-8}$, and weight decay $3 \times 10^{-6}$.

\textbf{Learning rate schedule.} We use a learning rate schedule with linear warmup followed by inverse-square-root decay. During the first $5{,}000$ optimizer steps, the learning rate increases linearly from zero to the base rate. After warmup, the learning rate decays as $\sqrt{t_{\text{warmup}}/t}$, where $t$ is the current optimizer step and $t_{\text{warmup}} = 5{,}000$.

\textbf{Dropout schedule.} We apply dropout during training to encourage diversity among learned representations. We set the dropout rate to $1/(10t)$, where $t$ indexes the optimizer steps from $1$ to $400{,}000$.

\textbf{Training procedure.} We pregenerate $400{,}000$ training examples, each consisting of latent inputs $(\util, \mat{L})$, simulated choice frequencies based on $10^6$ Monte Carlo draws, and target gradients computed using the soft relaxation described in Section~\ref{sec:training-procedure}. During training, we randomly generate directions in this latent space and compute directional derivatives as described in Section \ref{sec:training-procedure}.
Gradients are clipped to a maximum norm of $0.01$ and non-finite gradient penalties are filtered from the loss function.
We train for $400{,}000$ optimizer steps with a batch size of $10{,}000$.

\textbf{Network architecture.} Table~\ref{tab:emulator-architecture} summarizes the neural network architecture for each value of $\numalt$. All networks use the Swish activation function. Larger values of $\numalt$ require larger networks to maintain approximation accuracy.

\begin{table}[!htbp]
\centering
\caption{Simulation Emulator Architectures}
\label{tab:emulator-architecture}
\begin{tabular}{lccc}
\toprule
Component & $\numalt = 3$ & $\numalt = 5$ & $\numalt = 10$ \\
\midrule
\textit{Diagonal DeepSet} & & & \\[3pt]
\quad $\phi$ output dim & 8 & 24 & 32 \\
\quad $\phi$ hidden layers & 0 & 1 & 1 \\
\quad $\phi$ hidden dim & N/A & 24 & 32 \\
\quad $\rho$ output dim & 8 & 24 & 32 \\
\quad $\rho$ hidden layers & 0 & 0 & 0 \\
\quad $\rho$ hidden dim & N/A & N/A & N/A \\
\quad Parameters & 152 & 1{,}440 & 2{,}432 \\[6pt]
\textit{Off-diagonal DeepSet} & & & \\[3pt]
\quad $\phi$ output dim & 8 & 24 & 32 \\
\quad $\phi$ hidden layers & 0 & 1 & 1 \\
\quad $\phi$ hidden dim & N/A & 24 & 32 \\
\quad $\rho$ output dim & 8 & 24 & 32 \\
\quad $\rho$ hidden layers & 0 & 0 & 0 \\
\quad $\rho$ hidden dim & N/A & N/A & N/A \\
\quad Parameters & 128 & 1{,}368 & 2{,}336 \\[6pt]
\textit{Combining MLP} & & & \\[3pt]
\quad Output dim & 8 & 24 & 32 \\
\quad Hidden dim & 12 & 48 & 64 \\
\quad Hidden layers & 2 & 2 & 2 \\
\quad Parameters & 560 & 6{,}264 & 10{,}912 \\[6pt]
\textit{Equivariant layers} & & & \\[3pt]
\quad Hidden dim & 4 & 12 & 16 \\
\quad Hidden layers & 1 & 1 & 1 \\
\quad Parameters & 72 & 600 & 1{,}056 \\[6pt]
\textit{Total parameters} & 912 & 9{,}672 & 16{,}736 \\
\bottomrule
\end{tabular}

\floatnote{All configurations use an initial dropout rate of $0.1$ and Swish activations. N/A indicates that the hidden dimension is not applicable because the number of hidden layers is zero.}
\end{table}

\subsection{Derivation of Scaled Wishart PDF}
\label{sec:scaled-wishart}

This section derives the PDF of a scaled Wishart random matrix as described in Section \ref{sec:simulation-design}. The variables in this section, including $d$, $j$, $n$, and $\cov$, are defined locally and separately from those in the main paper.

Let $\mat{W} \sim W_d(\mat{V}, n)$ denote a $d \times d$ Wishart random matrix with scale matrix $\mat{V}$ and $n \geq d$ degrees of freedom. We derive the distribution of $\cov = \mat{W} / W_{jj}$, where $W_{jj}$ is the $(j,j)$ element of $\mat{W}$ for some $j \in \{1, \ldots, d\}$.

The Wishart density is
\[
f_{\mat{W}}(\mat{W}) = \frac{|\mat{W}|^{(n-d-1)/2} \exp\left\{-\frac{1}{2}\mathrm{tr}(\mat{V}^{-1}\mat{W})\right\}}{2^{nd/2} |\mat{V}|^{n/2} \Gamma_d(n/2)}
\]
where $\Gamma_d(\cdot)$ is the multivariate gamma function.

The inverse transformation is $\mat{W} = W_{jj} \cov$. Note that $\cov$ is positive definite with $\Sigma_{jj} = 1$, so it has $d(d+1)/2 - 1$ free parameters. Together with $W_{jj}$, this matches the $d(d+1)/2$ unique elements of $\mat{W}$. The Jacobian determinant of this transformation is $w_{jj}^{d(d+1)/2 - 1}$, where $w_{jj}$ denotes the value of $W_{jj}$.

Substituting $\mat{W} = W_{jj} \cov$ and applying the Jacobian yields the joint density
\[
f_{W_{jj}, \cov}(w_{jj}, \cov) = \frac{|\cov|^{(n-d-1)/2}}{2^{nd/2} |\mat{V}|^{n/2} \Gamma_d(n/2)} \cdot w_{jj}^{nd/2 - 1} \exp\left\{-\frac{w_{jj}}{2}\mathrm{tr}(\mat{V}^{-1} \cov)\right\}.
\]

Integrating over $w_{jj} > 0$ using the gamma integral gives the marginal density of $\cov$:
\[
f_{\cov}(\cov) = \frac{\Gamma(nd/2)}{2^{nd/2} |\mat{V}|^{n/2} \Gamma_d(n/2)} \cdot \frac{|\cov|^{(n-d-1)/2}}{\left\{\frac{1}{2}\mathrm{tr}(\mat{V}^{-1} \cov)\right\}^{nd/2}}
\]
for symmetric, positive definite $\cov$ with $\Sigma_{jj} = 1$ for some fixed $j \in \{1, \ldots, d\}$.

In the simulation study, we use $d = \numalt - 1$, $n = \numalt + 10$, $\mat{V} = \mat{I}_{\numalt-1}$, and $j = 1$, although the derivation is analogous for any $j$ by symmetry. Substituting these values and simplifying yields the density given in Section~\ref{sec:simulation-design}.

\section{Correlated-Gumbel Simulations}\label{sec:gumbel-sim}

This appendix provides simulations comparing our proposed emulator to softmax-smoothed sample averages. Appendix~\ref{sec:gumbel-sim-est} compares the two approaches when estimating the correlated-Gumbel model, and Appendix~\ref{sec:gumbel-sim-fixed} compares the accuracy of their choice probabilities at fixed inputs for both the MNP and correlated-Gumbel models.

\subsection{Simulation of Correlated-Gumbel Estimation}\label{sec:gumbel-sim-est}

This section repeats the simulation study of Section~\ref{sec:simulation} for the correlated-Gumbel model introduced in Section~\ref{sec:setup}. Its choice probabilities have no closed form, and the GHK simulator does not apply because the errors are not Gaussian. We consider $\numalt = 5$ alternatives with the Dense specification. The design matches that of Section~\ref{sec:simulation-design} except as described below.

\textbf{Model specification.} The errors take the form $\vec{\epsilon} = \cov^{1/2} \vec{\epsilon}^*$ with $\epsilon_j^* \overset{\text{iid}}{\sim} \text{Gumbel}(a, 1)$, where $a$ is the negative of the Euler--Mascheroni constant so that the errors have mean zero, and $\cov^{1/2}$ is the symmetric positive semidefinite square root.\footnote{Unlike in MNL, the location $a$ does not generally cancel in the correlated-Gumbel model because $\cov^{1/2} \ones_{\numalt}$ need not be proportional to $\ones_{\numalt}$. We use the same value of $a$ to train the emulator.} In the MNP study, the reference alternative (alternative $1$) has no error term: $\cov_{\bigcdot 1} = \cov_{1 \bigcdot}^{\top} = \vec{0}_{\numalt}$. Here, it instead has its own unit-scale Gumbel error, uncorrelated with the other errors: $\Sigma_{11} = 1$ and $\Sigma_{j1} = \Sigma_{1j} = 0$ for $j > 1$. Thus, $\cov$ has the level structure in \eqref{eq:Sigma-structure}, with the reference alternative as the base alternative and $\mat{\Omega} = \cov_{2:\numalt,2:\numalt}$. This ordering determines the indicators $1_{\{j=1\}}$ and $1_{\{k=1\}}$ supplied to the emulator (Section~\ref{sec:architecture}). As in the MNP study, $\mat{\Omega}$ is dense with $\Sigma_{22} = 1$ and is drawn from the normalized Wishart distribution of Section~\ref{sec:simulation-design}. When $\mat{\Omega} = \mat{I}_{\numalt-1}$, the model reduces to MNL with a reference alternative. Because the reference alternative's error already fixes the scale of utility, $\Sigma_{22} = 1$ is a restriction rather than a normalization in this model. We retain it so that the parameters, their true values, and the regularization match those of the MNP study. The two studies also use the same random seeds, so their covariates and true parameters coincide.

\textbf{Data generation.} The Gumbel DGP replaces the MNP DGP: it draws errors as described above and selects the utility-maximizing alternative. The Emulator DGP samples choices from the probabilities of the correlated-Gumbel emulator.

\textbf{Estimation.} The Emulator method uses a separately trained correlated-Gumbel emulator for $\numalt = 5$. We train it as described in Section~\ref{sec:training-procedure}, with the modifications noted there for error distributions that are not factorization-invariant. The emulator therefore includes the optional features $\ssign(\log \Cstar)$, $1_{\{j=1\}}$, and $1_{\{k=1\}}$ listed in Table~\ref{tab:input-features}. The training setup is otherwise the same as in Appendix~\ref{sec:emulator-training}.
Because GHK applies only to Gaussian errors, we compare the emulator with softmax-smoothed sample averages. This is the approximation used by \citet{Aouad2025} and \citet{Bagheri2025}, and it is the form of our Sobolev training targets:
\[
\tilde{\choiceprob}_j = \frac{1}{R} \sum_{r=1}^{R} \tilde{Y}_j^{(r)}, \qquad
\tilde{Y}_j^{(r)} = \frac{\exp(\tau U_j^{(r)})}{\sum_{k=1}^{\numalt} \exp(\tau U_k^{(r)})}, \qquad
\vec{U}^{(r)} = \util + \cov^{1/2} \vec{\epsilon}^{*(r)}.
\]
We denote this method Sim($\tau$, $R$) and consider $\tau \in \{5, 20\}$ and $R \in \{500, 2{,}000\}$. The temperature applies on the scale of \eqref{eq:Sigma-structure}, where the reference alternative's error has unit Gumbel scale. Larger values of $\tau$ reduce the smoothing bias but increase the simulation variance. The draws $\vec{\epsilon}^{*(r)}$ are held fixed across optimizer iterations and are independent across observations. Both methods use the regularized likelihood, optimizer, and standard-error formula of Section~\ref{sec:simulation-design}. For the simulator, $\hat{\fisher}_n$ is computed from the simulated scores, which treats the simulated log likelihood as exact.

\textbf{Results.} Table~\ref{tab:emulator_vs_sim_gumbel:5:all} reports the results. The emulator performs similarly under the Gumbel DGP and the Emulator DGP, where it is correctly specified. This suggests that its approximation error is small relative to sampling error at these sample sizes.

\begin{table}[htbp]
\centering
\caption{Dense Correlated-Gumbel Simulation Results: $\numalt=5$}
\label{tab:emulator_vs_sim_gumbel:5:all}
\footnotesize
{\setlength{\tabcolsep}{4.8pt}
\begin{tabular}{@{}rllrrrrr@{}}
\toprule
\multicolumn{1}{c}{$n$} & \multicolumn{1}{l}{DGP} & \multicolumn{1}{l}{Method} & \multicolumn{1}{c}{RMSE} & \multicolumn{1}{c}{RMS Bias} & \multicolumn{1}{c}{SE Ratio} & \multicolumn{1}{c}{Coverage} & \multicolumn{1}{c}{Time (s)} \\
\midrule
\multirow{6}{*}{1,000} & Emulator & Emulator & 0.351 (0.030) & 0.019 (0.012) & 1.24 (0.07) & 0.95 (0.01) & 0.47 (0.01) \\
 & Gumbel & Emulator & 0.411 (0.032) & 0.049 (0.021) & 1.21 (0.05) & 0.96 (0.01) & 0.45 (0.00) \\
 & Gumbel & Sim(5, 500) & 0.494 (0.041) & 0.091 (0.029) & 0.97 (0.04) & 0.92 (0.01) & 0.43 (0.01) \\
 & Gumbel & Sim(5, 2,000) & 0.446 (0.041) & 0.056 (0.023) & 1.14 (0.06) & 0.96 (0.01) & 1.37 (0.02) \\
 & Gumbel & Sim(20, 500) & 0.437 (0.041) & 0.054 (0.024) & 0.82 (0.04) & 0.90 (0.01) & 0.57 (0.01) \\
 & Gumbel & Sim(20, 2,000) & 0.409 (0.035) & 0.039 (0.019) & 1.09 (0.04) & 0.96 (0.01) & 1.58 (0.03) \\
\midrule
\multirow{6}{*}{10,000} & Emulator & Emulator & 0.142 (0.008) & 0.013 (0.006) & 1.03 (0.05) & 0.95 (0.01) & 2.71 (0.03) \\
 & Gumbel & Emulator & 0.150 (0.010) & 0.024 (0.009) & 1.07 (0.05) & 0.95 (0.01) & 2.65 (0.02) \\
 & Gumbel & Sim(5, 500) & 0.170 (0.015) & 0.014 (0.007) & 0.83 (0.06) & 0.89 (0.01) & 3.19 (0.04) \\
 & Gumbel & Sim(5, 2,000) & 0.157 (0.010) & 0.017 (0.008) & 1.00 (0.05) & 0.94 (0.01) & 12.34 (0.14) \\
 & Gumbel & Sim(20, 500) & 0.185 (0.018) & 0.023 (0.010) & 0.59 (0.04) & 0.77 (0.02) & 3.63 (0.07) \\
 & Gumbel & Sim(20, 2,000) & 0.164 (0.010) & 0.016 (0.008) & 0.84 (0.04) & 0.90 (0.01) & 13.22 (0.18) \\
\midrule
\multirow{6}{*}{100,000} & Emulator & Emulator & 0.051 (0.005) & 0.005 (0.002) & 0.94 (0.07) & 0.95 (0.01) & 20.50 (0.16) \\
 & Gumbel & Emulator & 0.049 (0.003) & 0.004 (0.002) & 1.02 (0.05) & 0.96 (0.01) & 20.27 (0.16) \\
 & Gumbel & Sim(5, 500) & 0.060 (0.003) & 0.021 (0.003) & 0.76 (0.04) & 0.86 (0.01) & 31.11 (0.31) \\
 & Gumbel & Sim(5, 2,000) & 0.054 (0.004) & 0.006 (0.002) & 0.93 (0.04) & 0.92 (0.01) & 150.01 (1.30) \\
 & Gumbel & Sim(20, 500) & 0.075 (0.004) & 0.007 (0.003) & 0.44 (0.02) & 0.64 (0.02) & 30.99 (0.35) \\
 & Gumbel & Sim(20, 2,000) & 0.053 (0.003) & 0.005 (0.002) & 0.83 (0.04) & 0.89 (0.01) & 144.72 (1.40) \\
\bottomrule
\end{tabular}
}

\floatnote{Performance metrics are averaged across coefficients, variances, and covariances. Monte Carlo standard errors are shown in parentheses. Sim($\tau$, $R$) denotes the softmax-smoothed simulator with temperature $\tau$ and $R$ draws.}
\end{table}

The emulator is substantially faster than the best-performing simulators. The simulators with $R = 2{,}000$ take about 3 to 3.5 times as long as the emulator at $n = 1{,}000$, about 5 times as long at $n = 10{,}000$, and about 7 times as long at $n = 100{,}000$. To limit memory use to 8~GB, the $R = 2{,}000$ fits at $n = 100{,}000$ recompute intermediate quantities during backpropagation, which inflates their times somewhat. The simulators with $R = 500$ take roughly as long as the emulator at $n = 1{,}000$ and about 1.5 times as long at $n = 100{,}000$, but they generally have the highest RMSE and the lowest coverage.

The emulator's point estimates are at least as accurate as those of the simulators. The emulator attains the lowest RMSE at $n = 10{,}000$ and $n = 100{,}000$ and effectively ties Sim(20, 2,000) at $n = 1{,}000$. Its margins over the $R = 2{,}000$ simulators are within Monte Carlo error, whereas its advantage over the $R = 500$ simulators is clear at $n = 100{,}000$. RMS Bias is generally comparable across methods. At $n = 100{,}000$, the emulator has the lowest RMS Bias, while Sim(5, 500) has a noticeably larger bias that does not decrease between $n = 10{,}000$ and $n = 100{,}000$.

The emulator also delivers reliable inference. Its coverage is 0.95--0.96 at every sample size under both DGPs, and its SE Ratios are close to one for $n \geq 10{,}000$. In contrast, the simulators' SE Ratios fall as $n$ grows and are below one in most settings, so their coverage deteriorates, reaching 0.64 for Sim(20, 500) at $n = 100{,}000$. The $R = 2{,}000$ simulators perform better, with coverage of at least 0.89 at all sample sizes, but they still undercover at $n = 100{,}000$. As with GHK(10) in Section~\ref{sec:simulation-results}, this pattern is consistent with simulation error that does not vanish as $n$ grows. Two sources contribute: the smoothing bias from a finite $\tau$, and the bias from taking the logarithm of an average over finitely many draws. For sufficiently large $n$, both become nonnegligible relative to $1/\sqrt{n}$ sampling error. The effect is most pronounced for $\tau = 20$ with $R = 500$, the configuration with the least smoothing and the fewest draws.

Overall, the emulator matches or exceeds the accuracy of the best simulator configurations, maintains nominal coverage, and runs several times faster than the $R = 2{,}000$ simulators. It also avoids choosing $\tau$ and $R$, choices to which the simulator's inference is sensitive.

\subsection{Accuracy of Choice Probabilities at Fixed Inputs}
\label{sec:gumbel-sim-fixed}

The estimation results above combine the approximation error of each method with sampling error. To isolate approximation error, we compare the emulator with Sim($\tau$, $R$) at fixed inputs $(\util, \cov)$ for both the MNP and correlated-Gumbel models. We use an MNP emulator and a correlated-Gumbel emulator, each trained on $\numalt \in \{2, \ldots, 12\}$ with factor-structured scale matrices as described in Appendix~\ref{sec:application-emulator-design}. We evaluate each at 36 configurations (72 in total across the two models) formed by crossing three choices. First, $\numalt \in \{3, 5, 10\}$. Second, $\mat{\Omega} = \mat{I}_{\numalt-1} + \mat{\Gamma}\mat{\Gamma}^\top/q$ for $q \in \{1, 2\}$ factors, or $\mat{\Omega} = \mat{I}_{\numalt-1}$ for $q = 0$; across alternatives $2, \ldots, \numalt$, the first factor's loadings increase linearly from 0 to 2 and the second factor's decrease from 2 to 0. Third, one of four utility profiles: $\util = \vec{0}$, $v_2 = 2$, $v_2 = 4$, or $v_j$ decreasing linearly from 0 to $-4$. The resulting choice probabilities range from $7 \times 10^{-5}$ to $0.995$. As reference probabilities, we use hard-choice frequencies from $10^8$ draws, whose standard errors are at most $5 \times 10^{-5}$.

Because $\tilde{\choiceprob}_j$ averages $R$ independent draws of $\tilde{Y}_j^{(r)}$, its mean squared error is exactly
\[
\E\big\{(\tilde{\choiceprob}_j - \choiceprob_j)^2\big\} = \big\{\E(\tilde{Y}_j^{(r)}) - \choiceprob_j\big\}^2 + \Var(\tilde{Y}_j^{(r)})/R.
\]
The first term is the squared smoothing bias, which does not shrink as $R$ grows; the second is the simulation variance. We estimate both terms from the same $10^8$ draws, which gives the simulator's error at any $R$, and we confirm the result with 100 independent replications of the simulator.

\begin{figure}[ptbh]
\centering
\caption{Choice-Probability Error by Number of Simulation Draws}
\label{fig:smoothed-error-vs-r}
\includegraphics[width=\textwidth]{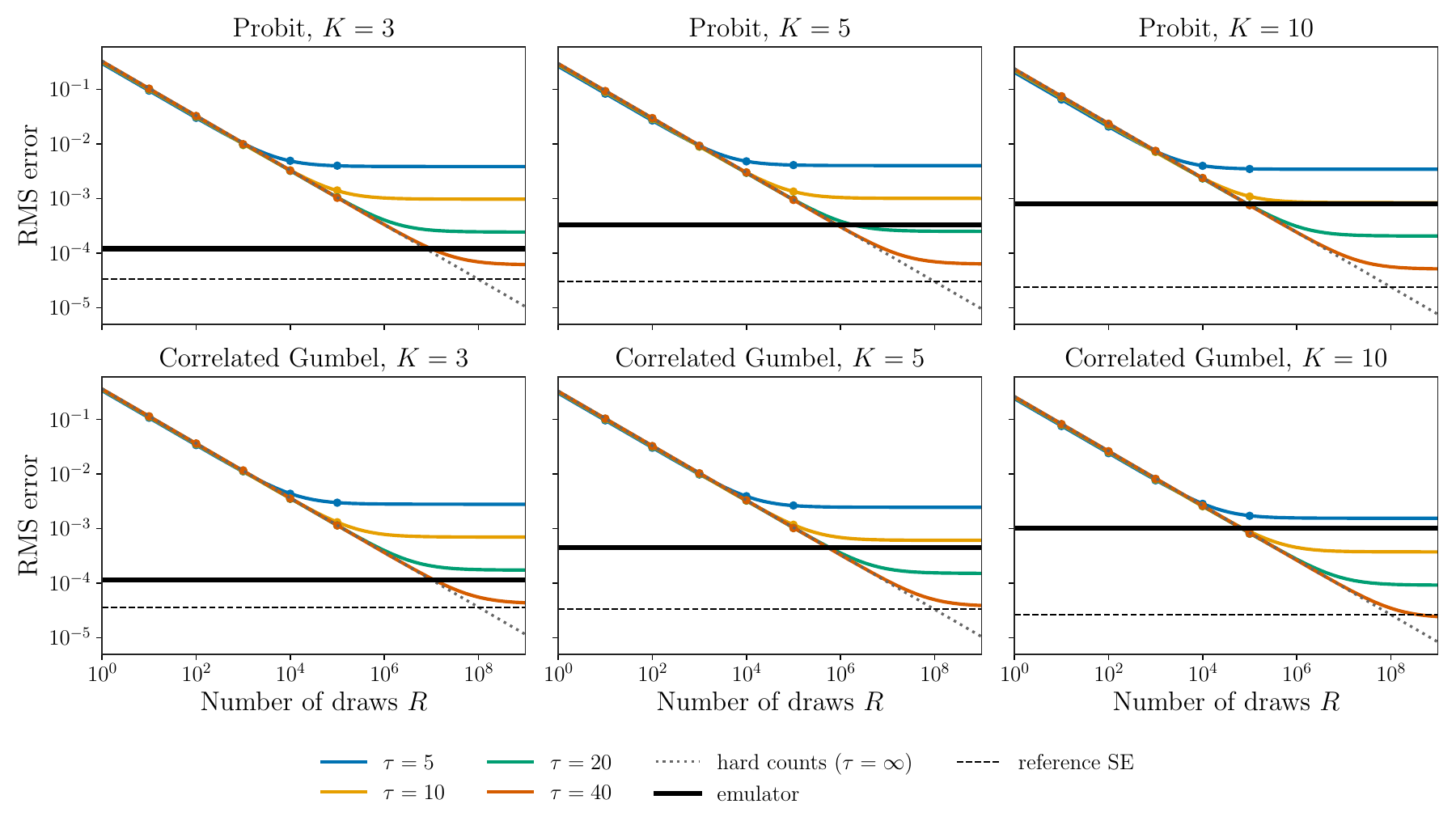}

\floatnote{Curves show the RMS error of the choice probabilities from Sim($\tau$, $R$) as a function of the number of draws $R$, and the thick horizontal line shows the RMS error of the emulator. Errors are averaged over alternatives and the 12 configurations in each panel. Dots show 100 independent replications of the simulator, and the dashed line shows the standard error of the reference probabilities. The flat right tails are the smoothing bias.}
\end{figure}

Figure~\ref{fig:smoothed-error-vs-r} shows the results. The emulator's RMS error is $1.1$--$1.2 \times 10^{-4}$ for $\numalt = 3$, $3.3$--$4.5 \times 10^{-4}$ for $\numalt = 5$, and $0.8$--$1.0 \times 10^{-3}$ for $\numalt = 10$.

Two features stand out. First, with $\tau = 5$, the simulator's error levels off at $1.6$--$4.0 \times 10^{-3}$, above the emulator's error for every model and $\numalt$, so no number of draws matches the emulator. Larger values of $\tau$ lower this floor roughly fourfold each time $\tau$ doubles. Second, even with little smoothing bias, matching the emulator requires many draws. For $\tau = 40$, the simulator matches the emulator at $R^* = 1.0$--$1.2 \times 10^{7}$ draws for $\numalt = 3$, $5.4$--$8.5 \times 10^{5}$ for $\numalt = 5$, and $6.6$--$8.5 \times 10^{4}$ for $\numalt = 10$.

Table~\ref{tab:smoothed-cost} translates these numbers into run time using the computation resources described in Section~\ref{sec:simulation-design}. Evaluating the choice probabilities for 1{,}000 observations takes the emulator 4--21 ms. A simulator with the same run time uses only about 120 to 940 draws, and its RMS error is then 8 to 200 times the emulator's. Conversely, matching the emulator's accuracy takes roughly 70 to 50{,}000 times as long as the emulator. The largest ratios occur for $\numalt = 3$, where the emulator is most accurate.

Accuracy at fixed inputs does not translate one-for-one into the accuracy of an estimator, because simulation noise that is independent across observations partly averages out in the likelihood. Nonetheless, the smoothing bias and simulation variance documented here are consistent with the simulators' undercoverage in Table~\ref{tab:emulator_vs_sim_gumbel:5:all}.

\begin{table}[htbp]
\centering
\caption{Cost of Matching Emulator Accuracy with Sim(40, $R$)}
\label{tab:smoothed-cost}
\footnotesize
{\setlength{\tabcolsep}{5pt}
\begin{tabular}{@{}llrrrrrrr@{}}
\toprule
 & & & \multicolumn{3}{c}{CPU} & \multicolumn{3}{c}{MPS} \\
\cmidrule(l){4-6} \cmidrule(l){7-9}
\multicolumn{1}{l}{Model} & \multicolumn{1}{c}{$\numalt$} & \multicolumn{1}{c}{$R^*$} & \multicolumn{1}{c}{Emulator (ms)} & \multicolumn{1}{c}{$R_{\mathrm{eq}}$} & \multicolumn{1}{c}{Ratio} & \multicolumn{1}{c}{Emulator (ms)} & \multicolumn{1}{c}{$R_{\mathrm{eq}}$} & \multicolumn{1}{c}{Ratio} \\
\midrule
\multirow{3}{*}{MNP} & 3 & $1.0 \times 10^{7}$ & 5.4 & 160 & 50{,}000 & 4.1 & 510 & 10{,}000 \\
 & 5 & $8.5 \times 10^{5}$ & 5.6 & 120 & 5{,}900 & 5.1 & 690 & 820 \\
 & 10 & $8.5 \times 10^{4}$ & 9.8 & 120 & 650 & 5.7 & 280 & 76 \\
\midrule
\multirow{3}{*}{Gumbel} & 3 & $1.2 \times 10^{7}$ & 9.1 & 680 & 14{,}000 & 7.4 & 940 & 8{,}800 \\
 & 5 & $5.4 \times 10^{5}$ & 11.2 & 720 & 680 & 5.2 & 380 & 740 \\
 & 10 & $6.6 \times 10^{4}$ & 21.3 & 840 & 75 & 10.3 & 700 & 68 \\
\bottomrule
\end{tabular}
}

\floatnote{$R^*$ is the number of draws at which the simulator's RMS error equals the emulator's (Figure~\ref{fig:smoothed-error-vs-r}). For each device, Emulator (ms) is the time of one emulator evaluation of the choice probabilities for 1{,}000 observations; $R_{\mathrm{eq}}$ is the number of draws at which the simulator takes the same time; and Ratio is the simulator's time at $R^*$ divided by the emulator's time. Simulator times beyond $R = 10^4$ are extrapolated linearly in $R$.}
\end{table}

\FloatBarrier
\section{Application Details}\label{sec:application-details}

This appendix provides additional details on the application to retail laundry-detergent purchases in Section \ref{sec:application} of the main paper.

\subsection{Emulator Design}\label{sec:application-emulator-design}

For both the MNP and correlated-Gumbel models, we train emulators on randomly generated data with $\numalt \in \{2, \ldots, 12\}$ alternatives.
To improve the generalization to factor-based models, we modify the data-generation approach given in Section \ref{sec:training-procedure} to sample $\mat{\Omega}$ according to a factor structure with $q \in \{0, \ldots, 8\}$ factors.
Specifically, we set $\mat{\Omega} = \mat{\Gamma} \mat{\Gamma}^{\top} + \mat{\Psi}$, where $\mat{\Gamma} \in \R^{(\numalt-1) \times q}$ is a matrix with standard normal entries and $\mat{\Psi} \in \R^{(\numalt-1) \times (\numalt-1)}$ is diagonal with each diagonal element distributed as $\chi^2(1)$.
The other difference in the data generation is that we use the symmetric positive semidefinite square-root matrix to correlate the errors in the Gumbel case.
In contrast, we use the Cholesky factorization for MNP.

For each value of $\numalt \in \{2, 3, \ldots, 12\}$, we generate $200{,}000$ training examples, each with $10^6$ simulated choices.
The training procedure loops over all values of $\numalt$, summing the loss values and their gradients before taking an optimizer step.
We train the emulators for $200{,}000$ optimizer steps, each consisting of a loop over $\numalt$ values with a batch size of $10{,}000$ for each value of $\numalt$.
We train all of the emulators for the application using NVIDIA L40S GPUs.

Because these emulators must generalize over a large number of choice-set sizes, we use larger neural networks than we do in the simulation study of Section \ref{sec:simulation}.
Table \ref{tab:application-emulator-architecture} shows the size of the emulators.
We train a single emulator for the MNP model, but we train two (Gumbel, Gumbel-large) for the Gumbel model.
The Gumbel emulator achieves a final loss of $1.2 \times 10^{-5}$, whereas the MNP and Gumbel-large emulators achieve a lower loss of $5 \times 10^{-6}$.
The results in Section \ref{sec:application-results} use the Gumbel-large emulator instead of the Gumbel emulator for this reason.
The Gumbel-large emulator requires a much larger number of parameters than the MNP emulator to achieve comparable performance.
We also train additional copies of each of these emulators on dense covariance matrices as in Section \ref{sec:training-procedure}; we use these emulators to fit the dense models reported in Table \ref{tab:agg-results}.

\begin{table}[htb]
\centering
\caption{Application Emulator Architectures}
\label{tab:application-emulator-architecture}
\begin{tabular}{lccc}
\toprule
Component & MNP & Gumbel & Gumbel-large \\
\midrule
\textit{Diagonal DeepSet} & & & \\[3pt]
\quad $\phi$ output dim & 64 & 64 & 256 \\
\quad $\phi$ hidden layers & 1 & 1 & 2 \\
\quad $\phi$ hidden dim & 64 & 64 & 64 \\
\quad $\rho$ output dim & 32 & 32 & 128 \\
\quad $\rho$ hidden layers & 1 & 1 & 2 \\
\quad $\rho$ hidden dim & 32 & 32 & 128 \\
\quad Parameters & 7{,}936 & 8{,}000 & 87{,}424 \\[6pt]
\textit{Off-diagonal DeepSet} & & & \\[3pt]
\quad $\phi$ output dim & 64 & 64 & 256 \\
\quad $\phi$ hidden layers & 1 & 1 & 2 \\
\quad $\phi$ hidden dim & 64 & 64 & 64 \\
\quad $\rho$ output dim & 32 & 32 & 128 \\
\quad $\rho$ hidden layers & 1 & 1 & 2 \\
\quad $\rho$ hidden dim & 32 & 32 & 128 \\
\quad Parameters & 7{,}744 & 7{,}744 & 87{,}168 \\[6pt]
\textit{Combining MLP} & & & \\[3pt]
\quad Output dim & 32 & 32 & 64 \\
\quad Hidden dim & 64 & 64 & 128 \\
\quad Hidden layers & 2 & 2 & 2 \\
\quad Parameters & 10{,}976 & 11{,}104 & 59{,}072 \\[6pt]
\textit{Equivariant layers} & & & \\[3pt]
\quad Hidden dim & 16 & 16 & 32 \\
\quad Hidden layers & 1 & 1 & 2 \\
\quad Parameters & 1{,}056 & 1{,}056 & 6{,}240 \\[6pt]
\textit{Total parameters} & 27{,}712 & 27{,}904 & 239{,}904 \\
\bottomrule
\end{tabular}

\floatnote{All configurations use an initial dropout rate of $0.1$ and Swish activations. All three emulators receive $\numalt$ as a pass-through feature. The Gumbel emulators also receive the scale term $\Cstar$ and first-alternative indicators, adding one input to the diagonal DeepSet $\phi$ and two inputs to the combining MLP.}
\end{table}

\subsection{Model-Fitting Algorithm}\label{sec:application-fitting-algorithms}

As in the simulation study, we fit the choice models via a modified L-BFGS optimizer.
The penalties are set low enough that they are negligible compared to the optimization tolerance, but they help the optimizer avoid problematic regions in the parameter space, such as very small uniquenesses.
As in Section \ref{sec:simulation}, the optimizer operates on an unconstrained latent scale, transforming these values back to the natural scale as needed.

We fit the models in an iterative fashion, beginning with the IID models.
The uniqueness-only models use the IID estimates to generate starting values.
The 1-factor models then start from the uniqueness-only estimates.
For $q\geq2$, we use starting values from the fitted model with $q-1$ factors.

We evaluate the log likelihood using 200 independent draws for the random effects. All results in Table \ref{tab:agg-results} come from emulator-based log-likelihood evaluations, including the IID models, which we evaluate with the emulators by fixing the error covariance at the identity. The data set includes two decision points with exactly two alternatives. For these rare cases in the MNP models, we use the exact probit log likelihood, which is available in closed form.

To check the accuracy of the emulators, we also fit the IID models using exact log-likelihood evaluations. Conditional on the sampled random effects, the log likelihood is available in closed form for the mixed-MNL (IID Gumbel) model, and we evaluate it for the IID probit model via Gauss--Hermite quadrature with 20 nodes. The maximized emulator-based log likelihoods differ from their exact counterparts by at most 2.04, a small discrepancy relative to the differences between the models in Table \ref{tab:agg-results}.

To address potential nonconvexity of the log-likelihood function, we use 19 random restarts to fit the initial IID models (20 fits total).
Due to computational constraints, we do not use random restarts for the subsequent models.
However, after fitting the models with both the Gumbel and Gumbel-large emulators, we refit the lower-performing model using starting values from the higher-performing model.
We also fit the dense models from two sets of starting values, the estimates of the corresponding IID and 7-factor models, and retain the fit with the higher log likelihood. For the correlated-Gumbel dense models, the refitting step follows these two fits: we compare the best fit from each emulator and refit the lower-performing model using starting values from the higher-performing model. The timing results in Table \ref{tab:agg-results} report the total optimization time with the given emulator (Gumbel-large for the correlated-Gumbel models). For the dense models, this includes both fits and, where applicable, the refit. For the correlated-Gumbel models, the refit time is included whenever the refit was run, even if it did not improve the fit. The timings exclude the time spent fitting the models that supply starting values.

\end{document}